\documentclass[acmsmall,screen,nonacm]{acmart}

\usepackage{subcaption} %
\usepackage{amsmath, amsthm}
\usepackage{epsfig,amsfonts,latexsym,upref}
\usepackage{stmaryrd}
\usepackage{centernot}
\usepackage{hyphenat}
\usepackage[textsize=tiny, %
]{todonotes}
\setuptodonotes{fancyline}
\usepackage[shortlabels, inline]{enumitem}

\definecolor{darkgreen}{rgb}{0.0, 0.5, 0.0} %

\usepackage[bibliography=common,%
appendix=append%
]{apxproof}
\theoremstyle{plain}
\newtheoremrep{theorem}{Theorem}[section]
\newtheoremrep{lemma}[theorem]{Lemma}
\newtheoremrep{proposition}[theorem]{Proposition}
\newtheoremrep{corollary}[theorem]{Corollary}
\newtheoremrep{remark}[theorem]{Remark}

\newtheoremstyle{lipicsdefinition-noparen}%
   {3pt}   %
   {3pt}   %
   {\normalfont} %
   {}      %
   {\sffamily\bfseries} %
   {.}     %
   { }     %
   {\thmname{#1}\thmnumber{ #2}\thmnote{ #3}} %

\newtheoremstyle{noparen}
  {3pt}
  {3pt}
  {\normalfont} 
  {}
  {\itshape}
  {.}
  { .5em}
  {\thmname{#1}\ \thmnumber{#2}\ \thmnote{ (#3)}}

\theoremstyle{noparen}
\newtheorem{instance}[theorem]{Instance}

\usepackage{galois}
\usepackage[ruled,vlined,linesnumbered]{algorithm2e}
\usepackage[inference]{semantic}

\usepackage{tikz}
\usepackage{mathtools,tikz-cd}
\usetikzlibrary{fit}

\usetikzlibrary{datavisualization.formats.functions}
\usetikzlibrary{decorations.markings}
\usetikzlibrary{arrows,backgrounds}
\usetikzlibrary{positioning}
\usetikzlibrary{arrows.meta}

\usepackage{pgfplots}
\pgfplotsset{compat=1.17}
\usetikzlibrary{decorations.pathreplacing,calligraphy}

\usepackage{mathrsfs}
\usepackage[all]{nowidow}
\usepackage{xfrac}
\usepackage{mdframed}
\usepackage{alltt}
\usepackage{upgreek}

\usepackage{galois}
\usepackage{graphicx}
\usepackage{svg}
\usepackage{cancel}
\usepackage{mathpartir}

\usepackage{multicol}

\renewcommand{\qed}{\quad\qedsymbol}

\newcommand{\set}[1]{\left \{ #1 \right \}}

\newcommand{\lat}{D} 

\newcommand{\N}{\mathbb{N}}
\newcommand{\Z}{\mathbb{Z}}

\renewcommand{\L}{\mathcal{L}}

\newcommand{\as}[1]{ \langle #1 \rangle}

\newcommand{\aslightdef}[3]{ \textcolor{gray}{\langle} #1  \textcolor{gray}{#2 \{\,} #3 \textcolor{gray}{\,\}\rangle}}

\newcommand{\Var}{\ensuremath{\mathit{Var}}}
\newcommand{\subst}[3]{\ensuremath{{#1}[{#3}/{#2}]}}

\newcommand{\Int}{\ensuremath{\mathsf{Int}}}
\newcommand{\interval}[2]{\ensuremath{\left [{#1},{#2} \right ]}}

\newcommand{\Sign}{\ensuremath{\mathsf{Sign}}}

\newcommand{\e}{\mathsf{e}}
\renewcommand{\r}{\mathsf{r}}
\newcommand{\s}{\mathsf{s}}
\newcommand{\ti}{\mathsf{t}}
\newcommand{\q}{\mathsf{q}}

\newcommand{\fix}[1]{\mathsf{fix} (#1)}

\newcommand{\lfp}{\operatorname{lfp}}
\newcommand{\1}{\mathsf{skip}}
\newcommand{\0}{\mathsf{0}}
 
\newcommand{\down}{\mathord{\downarrow}}

\newcommand{\defiff}{\mathrel{\stackrel{\mathrm{def}}{\Longleftrightarrow}}}

\newcommand{\APPL}{\text{APPL}}
\newcommand{\Reg}{\mathsf{Reg}}
\newcommand{\Ecom}{\mathsf{Elc}}

\newcommand{\base}{\ensuremath{B}} 
\newcommand{\ebase}[1]{\ensuremath{B{(#1)}}} 

\newcommand{\Ir}[1]{\operatorname{Irr}({#1})}
\newcommand{\mon}[1]{\ensuremath{#1^\oplus}}

\newcommand{\w}[1]{\ensuremath{w(#1)}}

\newcommand{\post}[1]{\mathrm{post} ({#1})}

\newcommand{\sem}[2][\!]{\llbracket{#2} \rrbracket_{\scriptscriptstyle{#1}}\,}
\newcommand{\bsem}[2][\!]{\llparenthesis #2 \rrparenthesis_{\scriptscriptstyle{#1}}\,}

\newcommand{\isem}[2][]{\llbracket{#2} \rrbracket_{\scriptscriptstyle#1}}
\newcommand{\pow}[1]{\wp ({#1})}

\newcommand{\add}[1]{\ensuremath{\widehat{#1}}}

\newcommand{\oplcut}[1]{}

\renewcommand{\state}[1]{\sigma_{#1}}

\newcommand{\Oct}{\ensuremath{\mathsf{Oct}}}
\newcommand{\po}{\texttt{p}}     %
\newcommand{\bo}{\texttt{b}}     %
\newcommand{\Ostd}{O_{\mathrm{std}}}
\newcommand{\Onew}{O_{\mathrm{new}}}

\newcommand{\colcode}[1]{\textcolor{blue}{#1}}

\AtBeginDocument{%
  }

\usepackage{amsmath,amsthm}
\usepackage{tikz}
\usepackage{mathtools,tikz-cd}
\usepackage{galois}

\usepackage{graphicx}
\usepackage{svg}
\usepackage{centernot}
\usepackage[inference]{semantic}
\usepackage{cancel}
\usepackage[textsize=tiny]{todonotes}
\usepackage{mathpartir}

\usepackage{soul}

\newcommand{\Tmpl}[1]{\ensuremath{\mathsf{T}_{#1}}}

\definecolor{lvlA}{RGB}{102,166,30}    %
\definecolor{lvlB}{RGB}{240,170,30}    %
\definecolor{lvlC}{RGB}{220,70,30}     %
\definecolor{lvlD}{RGB}{204,0,153}     %
\definecolor{lvlE}{RGB}{0,114,178}     %
\definecolor{lvlF}{RGB}{0,150,136}     %

\usepackage{subcaption}

\setcopyright{acmlicensed}
\copyrightyear{2026}
\acmYear{2026}
\acmDOI{XXXXXXX.XXXXXXX}
\acmConference[Conference acronym'XX]{%
  }{%
  }{%
  }

\acmISBN{%
}

\allowdisplaybreaks

\begin{document}

\title{The Program Logic of Abstract Interpretation}

\author{Paolo Baldan}
\email{baldan@math.unipd.it}
\orcid{https://orcid.org/0000-0001-9357-5599}
\affiliation{
  \institution{Dipartimento di Matematica, University of Padova}
  \country{Italy}
}

\author{Roberto Bruni}
\email{bruni@di.unipi.it}
\orcid{https://orcid.org/0000-0002-7771-4154}
\affiliation{
  \institution{Dipartimento di Informatica, University of Pisa}
   \country{Italy}
}

\author{Francesco Ranzato}
\email{francesco.ranzato@unipd.it}
\orcid{https://orcid.org/0000-0003-0159-0068}
\affiliation{
  \institution{Dipartimento di Matematica, University of Padova}
  \country{Italy}
}

\author{Diletta Rigo}
 \email{diletta.rigo@math.unipd.it}
 \orcid{https://orcid.org/0009-0008-4473-3202}
\affiliation{
  \institution{Dipartimento di Matematica, University of Padova}
  \country{Italy}
}

\renewcommand{\shortauthors}{}

\begin{abstract}
Hoare-style program logics and abstract interpretation are two of the principal
frameworks for reasoning about programs, and their affinity has long been
recognised informally: both derive over-approximate judgements about program
behaviour by structural, syntax-directed reasoning.
Turning this analogy into a formal correspondence runs into two challenges.
First, the disjunction (or join) rule, sound and admissible in Hoare logic, becomes unsound over non-additive abstract semantics. 
Second, the extrapolation operators that make abstract interpretation effective on unbounded computations, i.e.\ widening and narrowing, have no analogue in the derivations of existing program logics.
We introduce \APPL{} (Abstract Program Property Logic), a Hoare-style program logic that provides a proof-theoretic foundation for abstract interpretation while addressing both challenges.
The key technical idea is density: abstract semantics is additive over dense families of logical assertions.
This leads to a density-guarded join rule that enables sound disjunctive
reasoning over arbitrary abstract domains.  
We also internalise the application of widening and narrowing through inference rules for iterative commands.
As a result, an entire analysis by abstract interpretation can be viewed as a logical derivation in \APPL{}. The logic is sound over any abstract domain and, when the
assertion language is expressive enough, it is also relatively complete.
Finally, by interpreting nondeterminism through a monoidal operator, \APPL{} provides a single logical framework that uniformly subsumes Hoare logic, incorrectness logic, and Hoare-style logics for hyperproperties, thereby enabling their abstraction.
Remarkably, the density-guarded join rule introduced to handle non-additive abstract semantics is also fundamental in guaranteeing completeness for hyperproperties.
\end{abstract}

\ccsdesc[500]{Theory of computation~Logic and verification}
\ccsdesc[500]{Theory of computation~Abstraction}
\ccsdesc[500]{Theory of computation~Programming logic}
\ccsdesc[500]{Theory of computation~Semantics and reasoning}
\ccsdesc[500]{Theory of computation~Program analysis}
\ccsdesc[300]{Theory of computation~Hoare logic}
\ccsdesc[500]{Theory of computation~Program reasoning}

\maketitle

\section{Introduction}\label{sec:intro}

Reasoning about programs has long relied on two frameworks that, although complementary to some extent, are traditionally regarded as distinct.
Beginning with Floyd and Hoare~\cite{Floyd1967Flowcharts,Hoare69}, \emph{Hoare-style program logics} are based on pre- and postconditions and establish correctness through syntax-directed inference rules that axiomatise judgements about the over-approximate behaviour of program statements.
They are modular, compositional, and proof-theoretically elegant, but not by themselves algorithmic: discovering loop invariants, and choosing the weakenings and strengthenings applied through the consequence rule, which control the degree of approximation and hence the precision of the reasoning, are entirely left to the ingenuity of the user.
Cousot and Cousot's \emph{abstract interpretation}~\cite{CC77} takes the complementary approach. 
Starting from the concrete semantics, it fixes 
once and for all a class of properties, the abstract domain,
and mechanises the computation of sound over-approximations, invariants included, through fixpoint iteration accelerated by widening and refined by narrowing. The resulting program analyses are fully automatic, but they are traditionally presented as semantic or algorithmic constructions rather than as deductive systems.

Despite these methodological differences, the two frameworks ultimately reason about the same mathematical object. 
A derivable Hoare triple and a sound abstract postcondition both establish an over-approximation of the collecting semantics, and the rules of Hoare logic closely mirror the inductive structure of an abstract interpreter. 
Their relationship has therefore long been recognised, and several works have explored different facets of it, including: calculational designs of program logics by abstract interpretation~\cite{Cousot24}, later organised into taxonomies of predicate transformers~\cite{VerschtK25}; 
local completeness logics, which combine under-approximate postconditions with their abstractions~\cite{BruniGGR23,AscariBG25,CampionPG22}; outcome logic, for general branching effects~\cite{Zil:toplas};
a program logic for best correct approximations~\cite{GiacobazziR25}; and
program logics for abstract (hyper)properties~\cite{CousotW25,GiacobazziM10}.

\paragraph{Research question.}
This raises a natural foundational question:
\begin{center}
  \emph{Can abstract interpretation itself be presented as a program logic?}
\end{center}
More precisely, can one design a Hoare-style proof system whose assertions are taken from
arbitrary abstract domains and whose derivations can use %
the effective tools of program
analysis, notably widening and narrowing? %
And can such
a system retain the proof-theoretic guarantees of classical program logics, namely
soundness and relative completeness? We answer both positively. %

Such a correspondence is not merely a matter of presentation. It places abstract
interpretation on the same proof-theoretic footing as Hoare logic, allows abstract static
program analyses to be viewed as logical derivations, and reveals the logical principles
underlying abstract reasoning.

\paragraph{The challenges.}
A full correspondence, however, turns out to be surprisingly subtle to achieve.
Two challenges stand in the way,
both intrinsic to the abstract setting rather than matters of presentation.

\medskip
\noindent\emph{Challenge 1: Unsoundness of standard proof rules in abstract domains.}\
A logical view of abstract interpretation seems immediate at first glance: one could simply lift the inductive semantic transfer functions into logical rules.
The resulting system strongly resembles Hoare logic. However, the latter traditionally relies on structural rules like consequence, conjunction (meet), and disjunction (join), and, as originally pointed out by~\citet{CCLB:AIFR},
using an abstract domain directly as the assertion language for a Hoare-style logic introduces severe structural hazards. First, if the domain is not closed under meets (i.e., it is not defined through a Galois connection or a Moore family), the standard (meet) rule becomes unsound. More importantly, the (join) rule is in general unsound over non-disjunctive domains.
In concrete powerset-based Hoare logic the disjunction rule
\[
  \small
  \inferrule{\{P_1\}\,\r\,\{Q_1\} \\ \{P_2\}\,\r\,\{Q_2\}}
            {\{P_1\vee P_2\}\,\r\,\{Q_1\vee Q_2\}}
\]
is sound and admissible: the collecting semantics distributes over unions, so the analysis of a disjunction of inputs is exactly the disjunction of the individual analyses. 
Over an abstract domain this property often breaks. 
Abstract transfer functions are, in general, \emph{merely monotone}, \emph{not additive}: the abstract interpretation of a join of inputs may strictly over-approximate the join of 
their separate 
abstract interpretations. 
Decomposing an abstract precondition, analysing the components, and rejoining the results can thus be 
unsound. 
This is far from an isolated phenomenon: indeed \citet{CCLB:AIFR} note, in their abstract interpretation framework for refactoring, 
that ``\emph{common inference rules in Hoare logic, like the conjunction and disjunction rules, are false in general}'', 
attributing this to the failure of concretisation to preserve, respectively, finite meets and joins.
Their counterexample~\cite[Ex.~5]{CCLB:AIFR} is worth reproducing, as it isolates the phenomenon in a simple domain.

\begin{figure}[t]
  \centering
  \begin{subfigure}{0.30\textwidth}
    \centering
    \[
      \begin{tikzcd}[row sep=2mm, column sep=3mm]
        & \Z \ar[dl, dash] \ar[d, dash] \ar[dr, dash]& \\
        \Z_{\leq 0}  \ar[d, dash] \ar[dr, dash]& \Z_{\neq 0} \ar[dl, dash] \ar[dr, dash]& \Z_{\geq 0} \ar[dl, dash] \ar[d, dash]\\
        \Z_{< 0} \ar[dr, dash]& \Z_{= 0} \ar[d, dash]& \Z_{> 0} \ar[dl, dash] \\
        & \varnothing &
      \end{tikzcd}
    \]
    \caption{The abstract domain $\Sign$.}\label{fig:intro-sign}
  \end{subfigure}
  \hfill
  \begin{subfigure}{0.30\textwidth}
    \centering
    \[
      \begin{tikzcd}[row sep=2mm, column sep=3mm]
        & \Z \ar[dl, dash] \ar[dr, dash]& \\
        \Z_{\leq 0}  \ar[d, dash] & & \Z_{\geq 0} \ar[d, dash]\\
        \Z_{< 0} & & \Z_{> 0}  \\
        & \varnothing \ar[ul, dash] \ar[ur, dash] &
      \end{tikzcd}
    \]
    \caption{The reduced domain $\Sign^{-}$.}\label{fig:intro-sign-minus}
  \end{subfigure}
  \hfill
  \begin{subfigure}{0.30\textwidth}
    \centering
    \[
      \begin{tikzcd}[row sep=2mm, column sep=3mm]
        & \Z \ar[dl, dash]  \ar[dr, dash]& \\
        \Z_{\leq 0}  \ar[d, dash] \ar[dr, dash]& & \Z_{\geq 0} \ar[dl, dash] \ar[d, dash]\\
        \Z_{< 0} \ar[dr, dash]& \Z_{= 0} \ar[d, dash]& \Z_{> 0} \ar[dl, dash] \\
        & \varnothing &
      \end{tikzcd}
    \]
    \caption{The reduced domain $\Sign^{\mathsf{Irr}}$.
     }\label{fig:intro-signs-no-neq}
  \end{subfigure}
  \caption{The abstract domain $\Sign$ and its reductions $\Sign^{-}$ and $\Sign^{\mathsf{Irr}}$. 
  $\Sign^{-}$ is obtained by removing
  the elements $\Z_{=0}$ and $\Z_{\neq 0}$. In $\Sign^{-}$ the concretisation preserves
  neither binary joins nor binary meets: $\Z_{<0}\sqcup_{\Sign^{-}}\Z_{>0}=\Z$ while
  $\gamma(\Z_{<0})\cup\gamma(\Z_{>0})=\Z_{\neq 0}\subsetneq\Z$, and dually
  $\Z_{\leq 0}\sqcap_{\Sign^{-}}\Z_{\geq 0}=\varnothing$ while
  $\gamma(\Z_{\leq 0})\cap\gamma(\Z_{\geq 0})=\Z_{=0}\supsetneq\varnothing$.}
  \label{fig:intro-signs}
\end{figure}

Consider the reduced sign domain $\Sign^{-}$ shown in Figure~\ref{fig:intro-sign-minus}, 
obtained from $\Sign$ by removing the two elements $\Z_{=0}$ and $\Z_{\neq 0}$. The obvious concretisation function $\gamma:\Sign^{-}\to\wp(\Z)$ preserves neither binary joins nor binary meets. In~\cite{CCLB:AIFR}, a triple $\as{P}\,\r\,\as{Q}$ takes specifications $P$ and $Q$ as abstract values, and letting $\sem{\cdot}$ denote the collecting semantics, $\as{P}\,\r\,\as{Q}$ is valid when $\sem{\r}\gamma(P)\subseteq\gamma(Q)$. 
Both the following triples are valid in $\Sign^{-}$:
\[
  \as{\Z_{<0}}\ x:=x*x\ \as{\Z_{>0}}
  \qquad\text{and}\qquad
  \as{\Z_{>0}}\ x:=x*x\ \as{\Z_{>0}}.
\]
Their join along the (unsound) disjunction rule yields
$\as{\Z_{<0}\sqcup_{\Sign^{-}}\Z_{>0}}\ x:=x*x\ \as{\Z_{>0}\sqcup_{\Sign^{-}}\Z_{>0}}$, that is,
$\as{\Z}\ x:=x*x\ \as{\Z_{>0}}$, which is clearly invalid. The rule breaks precisely
because $\gamma(\Z_{<0}\sqcup_{\Sign^{-}}\Z_{>0})=\Z\supsetneq\Z\smallsetminus \{0\}=
\gamma(\Z_{<0})\cup\gamma(\Z_{>0})$: joining $\Z_{<0}$ and $\Z_{>0}$ in $\Sign^{-}$
\emph{fills the hole} $\{0\}$, an element covered by neither disjunct. 
A dual problem also arises with the conjunction rule. %
While the latter is naturally solved by asking the concretisation
function to preserve meets, 
which is a quite common assumption, the
former would require the domain to be disjunctive. Quoting~\citet[\S~5]{CCLB:AIFR},
this would be ``\emph{a very restrictive hypothesis essentially forbidding
the approximation of joins which is the basis for static analysis.}''

Any faithful logic of abstract interpretation must therefore either
forgo disjunctive reasoning, thus losing precision that is often
essential,
or find a principled side condition to guarantee that decomposing is safe. As we explain below, the notion of
density provides this side condition: it is the property that the decomposition
$\{\Z_{<0},\Z_{>0}\}$ lacks %
in $\Sign^{-}$.

\medskip
\noindent\emph{Challenge 2: Widening and narrowing are not integrated in program logics}.\ 
The operators that ensure that abstract interpretation terminates on loops, called \emph{widening} ($\nabla$) and \emph{narrowing} ($\Delta$)~\cite{CC77,cousot21}, are absent from program logics. They
are extrapolation and interpolation operators, not monotone abstractions of concrete
operations, and they are generally applied at the meta-level of the analysis algorithm, not
as inference steps. 
Typically, an abstract interpretation run involves an ascending Kleene iteration accelerated by $\nabla$ up to a post-fixpoint, i.e., an
over-approximation of the least fixpoint,
followed by a descending iteration refined by $\Delta$.
Consequently, it has no representation as a proof derivation in existing program logics.

\paragraph{Our contribution.}
We answer the research question in the affirmative:
density and a density-guarded (join) rule address the failure of additivity, 
while the use of widening and narrowing is internalised 
through inference rules for the least 
fixpoint command $\fix{\r}$, which characterises
the standard loop. %

\paragraph{Joining over dense decompositions.}
In our semantic model, we work over an \emph{interpretation domain}
$(D,\sqsubseteq,\base)$, which is a complete lattice with join
$\sqcup$ and equipped with a \emph{join basis} $\base$. A set
$X\subseteq D$ is defined to be \emph{dense} when every basis element below
$\sqcup X$ lies below some member of $X$, intuitively meaning that no ``hole''
is filled by joining.
In the domain $\Sign^{-}$, the set $\{\Z_{<0},\Z_{>0}\}$
is not dense. In fact, its abstract join $\Z$ is above the basis element
$\Z_{\leq 0}$ (join-irreducible in this lattice  $\Sign^{-}$ and thus necessarily included in every basis),
which is not below any of the joined elements. Indeed,
$\gamma(\Z)$ strictly exceeds
$\gamma(\Z_{<0})\cup\gamma(\Z_{>0})=\Z\smallsetminus\{0\}$, the
gap that makes the disjunction rule unsound. The key observation is
that the full semantics is \emph{additive on dense
sets}, i.e., $\sem{\r}(\sqcup X)=\bigsqcup_{d\in X}\sem{\r}d$
whenever $X$ is dense.  Density is therefore the hypothesis
under which decomposing a precondition is sound, and plays the role of
side condition of our novel (join) rule:
\[
  \small
  \inferrule*[left=\text{\rm (join)}]{\{h_i\}_{i\in I}\ \text{dense} \\
             h\sqsubseteq\textstyle\bigsqcup_{i\in I} h_i \\
             \forall i\in I.\ \vdash\as{h_i}\,\r\,\as{k_i} \\
             \textstyle\bigsqcup_{i\in I} k_i\sqsubseteq k}
            {\vdash\as{h}\,\r\,\as{k}}
\]
that recovers sound disjunctive reasoning over possibly non-additive abstract domains.
The rule (join) turns
case analysis into a sound and, as we show, essential proof principle. This is
our answer to the unsoundness pointed out in~\cite[Ex.~5]{CCLB:AIFR}: the
disjunction rule is not false in the abstract; it must be \emph{guarded}, and
density is the right guard.

\medskip
\noindent\emph{Widening and narrowing rules.}\
The following widening rule %
\[
  \small
  \inferrule*[left=\text{\rm (widen)}]
  {\vdash \as{h} \ \1 + \r \ \as{h'} \quad
   \vdash \as{h \ \nabla \ h'} \ \fix{\r} \ \as {k} \\
   }
  {\vdash \as{h} \ \fix{\r} \ \as {k}}
\]
composes a one-step body triple $\vdash\as{h}\ \1 + \r \ \as{h'}$, meaning that $h' \sqsupseteq h \sqcup h''$
with  $\vdash\as{h}\ \r \ \as{h''}$,  
 with a recursive call
on $\fix{\r}$ from the widened state $h \ \nabla \ h'$: iterating (widen) therefore builds an ascending,
$\nabla$-accelerated approximation of the loop invariant, and once a post-fixpoint is
reached the invariant rule (inv-fix), %
that is, $\vdash\as{h}\,\r\,\as{h}$\: $\Rightarrow$\: $\vdash\as{h}\,\fix{\r}\,\as{h}$, 
closes the recursion. On the other hand, the following narrowing rule
\[
  \small
  \inferrule*[left=\text{\rm (narrow)}]
  {\vdash \as{h} \ \fix{\r} \ \as{k} \\
   \vdash \as{k} \ \r \ \as {k'} \\
   h \sqcup k' \sqsubseteq h' }
  {\vdash \as{h} \ \fix{\r} \ \as {k \ \Delta \ h'}}
\]
allows us to refine a post-fixpoint $k$ by one descending, $\Delta$-controlled step, deriving
the tighter narrowed invariant $k \ \Delta \ h'$, where $\vdash\as{k}\ \r \ \as{k'}$ is a body iteration 
and $h'$ is an over-approximation of the join $h \sqcup k'$. 
Together, (inv-fix), (widen) and (narrow) capture the
abstract interpretation of a loop, with an ascending phase followed by a descending phase,
as a single derivation tree in \APPL{}. A companion, purely fixpoint-based rule
handles Kleene iteration in abstract domains directly:
\[
  \small
  \inferrule*[left=\text{\rm (rec-fix)}]
  {\vdash \as{h} \ \1+ \r \ \as{h'} \\
   \vdash \as{h'} \ \fix{\r} \ \as{k} }
  {\vdash \as{h} \ \fix{\r} \ \as{k}}
\]
unfolding the loop one step through $\vdash \as{h} \ \1+ \r \ \as{h'}$ and recurring on $\fix{\r}$ starting from $h'$. 
This (rec-fix) rule is the abstract domain counterpart of
Kleene iteration. %
Furthermore, observe that (widen) is exactly (rec-fix) with  $h'$ replaced by
$h \ \nabla \ h'$ to enforce convergence.

\paragraph{The logic {\rm \APPL{}}.}
The resulting logic, called \APPL{} (Abstract Program Property Logic), is a
Hoare-style proof system whose judgements are triples $\as{h}\,\r\,\as{k}$
with $\r$ a regular command and $h,k$ ranging over an assertion
language $\L\subseteq D$ for the interpretation domain
$(D,\sqsubseteq,\base)$.  Alongside the expected standard rules
(basic)
for elementary commands (typically, boolean guards and assignments), (seq) 
for sequential compositions, (cons) for Hoare's consequence rule, (choice) for non-deterministic
branching, (iter) and (rec-star) for Kleene's star iteration, the %
logic includes the %
density-guarded (join), the fixpoint rules (inv-fix), (rec-fix),
(fix-star), and the extrapolation rules (widen) and
(narrow).
The ability to perform %
disjunctive reasoning 
enables \APPL{} to derive more precise abstract properties in a domain $D$ 
than standard abstract interpretation on the same $D$. 

The expressive power of \APPL{} is characterised in terms of a program semantics given over the
interpretation domain $(D, \sqsubseteq, \base)$. It is defined first over the elements of the basis $\base$, and then extended by joins to the full domain $D$.  Such semantics is, in general, finer than the standard inductive abstract semantics over $D$. 
Its precision is governed by the choice of the basis: a smaller basis gives a finer semantics. 
This produces a hierarchy of semantics, with the ordinary inductively defined abstract semantics over $D$ being the coarsest.

We prove that \APPL{} is sound over every interpretation domain, i.e., 
each derivable triple is semantically valid. When the assertion language $\L\subseteq D$ is
\emph{expressive} enough (cf.~Definition~\ref{de:expressive-language})---in particular when $\L=D$,
the natural case for an abstract domain such as intervals or octagons---\APPL{} is
relatively complete, i.e., 
every valid triple is derivable. 
The completeness proof reveals a clear separation of roles: 
(basic), (seq), (choice), (iter), (inv-fix), (cons), and (join) 
suffice for completeness, 
whereas the extrapolation and
recursion rules are sound auxiliaries, as they derive no new triples
but make reasoning finitary and mechanisable, exactly as widening and narrowing do for static program analysis. We further show that the (join)
rule is needed only in the non-additive case: when the basis
is irreducible and the semantics is additive, it can be dropped without losing relative completeness, 
which recovers the known redundancy of the
disjunction rule in ordinary Hoare logic.

The problem raised in~\cite{CCLB:AIFR},
i.e., using ordinary Hoare logics over assertions ranging in an
abstract domain, is covered as an instance of the above, where $D$ is the concrete domain
and $\L = A$ is an abstraction of $D$. 
The resulting logic is sound but, in general, not complete. Derivations rely on concrete semantic operations and are therefore non-effective. 
To enable mechanisation, one can instead work with a purely abstract variant of the logic whose operators act directly on the abstract domain and which  
we prove to be as powerful as its concrete counterpart under mild conditions. 

\paragraph{One rule, many logics: incorrectness and hyperproperties.}
The framework as described so far forces nondeterminism to be
interpreted as the lattice join. This excludes notable settings such as
program logics for hyperproperties and under-approximation logics like
incorrectness logic. 
We generalise the framework
by interpreting nondeterminism through the monoidal operation of an
\emph{interpretation monoid}, yielding a
version of \APPL{} that uniformly subsumes Hoare
logic~\cite{Hoare69}, incorrectness logic~\cite{OHearn20}, and Hyper
Hoare Logic~\cite{DM:HHL}. One consequence is worth highlighting: the (join)
rule, which we introduce to tame non-additive abstract
semantics, turns out to be the very device on which completeness of
hyperproperty logics hinges. The reason is the same in both settings:
the operation combining computations is not additive on the relevant
structure. Moreover, since the monoidal operator need not coincide with
the lattice join, reversing the order of the interpretation structure
yields under-approximate, incorrectness-style reasoning in the sense of~\citet{OHearn20}.

\paragraph{Summary.}
This paper makes the following contributions.
\begin{itemize}
  \item We identify additivity as the obstacle to a proof-theoretic account of
    abstract interpretation, and introduce density as the property that
    guarantees additivity of the abstract semantics.
  \item We design \APPL{}, a program logic over arbitrary abstract interpretation
    domains, whose density-guarded (join) rule restores sound disjunctive
    reasoning over non-additive domains, and whose widening and narrowing
    rules internalise %
    the use of the extrapolation operators of abstract interpretation as
    inference steps.
  \item We prove \APPL{} sound over every interpretation domain and relatively
    complete over expressive assertion languages, isolate the rules responsible
    for completeness, and show when the (join) rule can be dropped.
  \item We show how to abstract an interpretation domain through a Galois
    insertion, and relate the strength of the logic over the concrete domain to
    its strength over the abstract one.
  \item We demonstrate the logic on interval and octagon analyses, deriving loop
    invariants strictly more precise than those of standard abstract
    interpretation over the same domain.
  \item We show that \APPL{} uniformly captures Hoare logic, incorrectness logic,
    and Hoare-style logics for hyperproperties, revealing that the (join) rule
    plays the same foundational role in non-additive abstract interpretation and
    in hyperproperty reasoning.
\end{itemize}

\section{Illustrative Example}\label{subsec:example}

Consider the program $\q \triangleq \textbf{while~} x<5 \textbf{~do~} (z:=x*x;\ x:=x+1)$. We analyse it over the interval domain $\Int$, whose abstract states map each integer variable $x_i$ to an interval $[l_i,r_i]$ bounding the values it may take, written $\{x_1/[l_1,r_1], \dots,\, x_n/[l_n,r_n]\}$. 
For the sake of readability, singleton intervals $[v,v]$ are abbreviated $[v]$. 
Suppose the abstract input state is $h_0\triangleq \{\,x/[{-}4],\ z/[0]\,\}$
(e.g., from the initialisation $x:={-}4;\ z:=0$). 
As a regular command, $\q$ is $\fix{\r}; (x \geq 5?)$, where 
$$\r \triangleq (x<5?);\ z:=x*x;\ x:=x+1$$
\noindent
 is the guarded loop body, and the least fixpoint construct $\fix{\r}$ computes the least inductive invariant of $\r$ above the input.
A loop invariant $a$ is inferable in $\APPL$ when $\vdash \as{a}\,\fix{\r}\,\as{a}$
is derivable. Viewed as a set of integer pairs $(x,z)$, the concrete loop invariant is
\[\{({-}4,0),({-}3,16),({-}2,9),({-}1,4),(0,1),(1,0),(2,1),(3,4),(4,9),(5,16)\},\]
whose interval abstraction, i.e.\ the best correct interval invariant, 
is $k\triangleq \{\,x/[{-}4,5],\ z/[0,16]\,\}$. Let us recall that the abstract domain $\Int$ is non-relational and therefore cannot capture dependencies between program variables.

We first show that $\APPL$ faithfully reproduces standard interval analysis, which infers the strictly
weaker invariant 
$\{\,x/[{-}4,5],\ z/[{-}16,16]\,\}$. We then show that a single use of the (join) rule lets
$\APPL$ infer the optimal invariant 
$k$, 
unreachable by standard interval analysis with the compositional evaluation of $x*x$.

\begin{figure}[!t]
\resizebox{\textwidth}{!}{%
\begin{tikzpicture}[
  jd/.style ={anchor=west, inner sep=1.6pt},
  lbl/.style={anchor=east, inner sep=1.6pt},
  >=Stealth,
]
\def\ind{.75}\def\rh{0.65}\def\base{-0.2}
\pgfmathsetmacro\arl{0.30}
\pgfmathsetmacro\txtgap{0.42}

\newcommand{\trunk}[5]{%
  \pgfmathsetmacro\cx{\base + #2*\ind - 0.16}%
  \pgfmathsetmacro\ty{-#3*\rh + 0.32}%
  \pgfmathsetmacro\by{-#4*\rh}%
  \draw[#1, line width=0.8pt] (\cx,\ty) -- (\cx,\by);%
  \foreach \rrow in {#5}{%
    \pgfmathsetmacro\ry{-\rrow*\rh}%
    \draw[#1, line width=0.8pt]
      (\cx,\ry) -- (\cx+\arl,\ry);%
  }%
}
\newcommand{\row}[4]{%
 \node[lbl] at (-2,-#1*\rh) {\scriptsize #1};%
  \node[lbl] at (0,-#1*\rh) {\rm #3};%
  \node[jd]  at (\base+#2*\ind+\txtgap,-#1*\rh) {$#4$};%
}

\trunk{lvlA}{1}{2}{10}{2,9,10}   %
\trunk{lvlB}{2}{3}{4}{3,4}     %
\trunk{lvlC}{3}{5}{6}{5,6}     %
\trunk{lvlD}{4}{7}{7}{7}         %
\trunk{lvlE}{5}{8}{8}{8}       %

\row{1}{0}{(narrow)}{\vdash \aslightdef{h_0}{\triangleq}{x/[{-}4],\ z/[0]}\ \colcode{\fix{\r}}\ \aslightdef{k_0}{\triangleq k_1\ \Delta \ h' =}{x/[{-}4,5],\ z/[{-}16,16]}}
\row{2}{1}{(widen)}{{\color{lvlA} \mathbf{s}_1  \triangleq \ } \vdash \aslightdef{h_0}{=}{x/[{-}4],\ z/[0]}\ \colcode{\fix{\r}}\ \aslightdef{k_1}{=}{x/[{-}4,{+}\infty],\ z/[{-}\infty,{+}\infty]}}
\row{3}{2}{(choice*)}{{\color{lvlB} \mathbf{s}_2 \triangleq }\ \vdash \aslightdef{h_0}{=}{x/[{-}4],\ z/[0]}\ \colcode{\1 + \r}\ \aslightdef{h'_0}{\triangleq}{x/[{-}4,{-}3],\ z/[0,16]}}
\row{4}{2}{(widen)}{{\color{lvlB} \mathbf{s}_3  \triangleq \ } \vdash \aslightdef{h_1}{\triangleq h_0\nabla h'_0=}{x/[{-}4,{+}\infty],\ z/[0,{+}\infty]}\ \colcode{\fix{\r}}\ \aslightdef{k_1}{=}{x/[{-}4,{+}\infty],\ z/[{-}\infty,{+}\infty]}}
\row{5}{3}{(choice*)}{{\color{lvlC} \mathbf{s}_4 \triangleq }\ \vdash \aslightdef{h_1}{=}{x/[{-}4,{+}\infty],\ z/[0,{+}\infty]}\ \colcode{\1 + \r}\ \aslightdef{h'_1}{\triangleq}{x/[{-}4,{+}\infty],\ z/[{-}16,{+}\infty]}}
\row{6}{3}{(inv-fix)}{{\color{lvlC} \mathbf{s}_5 \triangleq }\vdash \aslightdef{k_1}{\triangleq h_1\ \nabla\ h'_1=}{x/[{-}4,{+}\infty],\ z/[{-}\infty,{+}\infty]} \ \colcode{\fix{\r}} \ \aslightdef{k_1}{=}{x/[{-}4,{+}\infty],\ z/[{-}\infty,{+}\infty]}}
\row{7}{4}{(cons)}{{\color{lvlD} \mathbf{s}_6 \triangleq \ }\vdash \aslightdef{k_1}{=}{x/[{-}4,{+}\infty],\ z/[{-}\infty,{+}\infty]}\ \colcode{\r}\ \aslightdef{k_1}{=}{x/[{-}4,{+}\infty],\ z/[{-}\infty,{+}\infty]}}
\row{8}{5}{(seq)+(basic)}{{\color{lvlE} \mathbf{s}_7 \triangleq \ } \vdash \aslightdef{k_1}{=}{x/[{-}4,{+}\infty],\ z/[{-}\infty,{+}\infty]}\ \colcode{\r}\ \aslightdef{k'_1}{\triangleq}{x/[{-}3,5],\ z/[{-}16,16]}}
\row{9}{1}{(seq)+(basic)}{{\color{lvlA} \mathbf{s}_7 \triangleq \ } \vdash \aslightdef{k_1}{=}{x/[{-}4,{+}\infty],\ z/[{-}\infty,{+}\infty]}\ \colcode{\r}\ \aslightdef{k'_1}{=}{x/[{-}3,5],\ z/[{-}16,16]}}
\row{10}{1}{}{{\color{lvlA} h'\triangleq}\ h_0\sqcup k'_1 =\{\,x/[{-}4,5],\ z/[{-}16,16]\,\}}

\end{tikzpicture}%
}
\caption{Derivation for the standard interval analysis in $\APPL$. For conciseness, we use a vertical layout instead of the standard derivation tree.
We mark with a $*$ the rules whose derivation trees are omitted for readability. Each colored vertical trunk denotes one inference rule application. The judgment immediately preceding the trunk is the rule conclusion, whereas the judgments connected to the trunk are its premises.}
\label{fig:linear-derivation1}
\end{figure}

\paragraph{Standard interval analysis in $\APPL$.}
We begin by deriving the abstract interval invariant $k_1=\{x/[{-}4,{+}\infty],\ z/[{-}\infty,{+}\infty]\}$ using widening, and then apply one narrowing step to refine it to $k_0=\{\,x/[{-}4,5],\ z/[-16,16]\,\}$.
The corresponding derivation tree is represented in Figure~\ref{fig:linear-derivation1}. 
The invariant $k_1$ is derived in lines 2--8. Starting from the goal $\as{h_0=\{\,x/[{-}4],\ z/[0]\,\}}\ \colcode{\fix{\r}}\ \as{k_1}$, 
we first apply rule (widen) to determine $k_1$. Its first premise is the judgement
\begin{center}
  ${{\color{lvlB} \mathbf{s}_2 \triangleq }\ \vdash \as{h_0}\ \colcode{\1 + \r}\ \aslightdef{h'_0}{\triangleq}{x/[{-}4,{-}3],\ z/[0,16]}}$
\end{center}
  \noindent
obtained using rule (choice) to compose the trivial judgement $\as{h_0}\ \colcode{\1}\ \as{h_0}$ and the one for the loop body $\as{h_0}\ \colcode{\r}\ \as{\{\,x/[{-}3],\ z/[16]\,\}}$ derived using (basic) and (seq), see line 3 (although some details are omitted for readability). Note that here $[{-}4]*^{\sharp}[{-}4]=[16]$ is trivially exact. 

Mimicking the abstract interpretation analysis of the loop, we widen the result against the input, and try to solve the goal
$\as{h_1\triangleq h_0 \ \nabla\  h'_0=\{x/[{-}4,{+}\infty],\ z/[0,{+}\infty]\}}\ \colcode{\fix{\r}}\ \as{k_1}$ to determine $k_1$ (line~4).
Applying the (widen) rule again repeats the analysis of the guarded loop and yields:
$${\color{lvlC} \mathbf{s}_4 \triangleq }\ \vdash \aslightdef{h_1}{=}{x/[{-}4,{+}\infty],\ z/[0,{+}\infty]}\ \colcode{\1 + \r}\ \aslightdef{h'_1}{\triangleq}{x/[{-}4,{+}\infty],\ z/[-16,{+}\infty]}\ .$$

In fact, by rule (choice), $h'_1$ is the join of the postconditions of the judgements $\as{h_1}~\1~\as{h_1}$ and $\as{h_1}~\r~\as{\{\,x/[{-}3,5],\ z/[{-}16,16]\,\}}$. The latter is obtained by rules (basic) and (seq), after filtering $h_1$ by the guard $(x <5 ?)$, thus getting $\{x/[{-}4,4],\ z/[0,{+}\infty]\}$, and then by computing the square compositionally, $[{-}4,4]*^{\sharp}[{-}4,4]=[{-}16,16]$. Then, we widen the result against the input, reaching 
$k_1 \triangleq h_1\ \nabla\ h'_1=\{x/[{-}4,{+}\infty],\ z/[{-}\infty,{+}\infty]\}$.
It is an invariant because
$${\color{lvlE} \mathbf{s}_7 \triangleq \ } \vdash \aslightdef{k_1}{=}{x/[{-}4,{+}\infty],\ z/[{-}\infty,{+}\infty]}\ \colcode{\r}\ \aslightdef{k'_1}{\triangleq}{x/[{-}3,5],\ z/[{-}16,16]}\ \textrm{ with } k'_1 \sqsubseteq k_1.$$
\noindent
This is shown in lines 6--8 and it concludes the derivation in lines 2--8 of the judgement:
$${\color{lvlA} \mathbf{s}_1  \triangleq \ } \vdash \aslightdef{h_0}{=}{x/[{-}4],\ z/[0]}\ \colcode{\fix{\r}} \ \aslightdef{k_1}{=}{x/[{-}4,{+}\infty],\ z/[{-}\infty,{+}\infty]}\ .$$
Finally, we can refine the result by applying the (narrow) rule with premises $\mathbf{s}_1$, $\mathbf{s}_7$, and $h'\triangleq\ h_0\sqcup k'_1 =\{\,x/[{-}4,5],\ z/[{-}16,16]\,\}$  (see lines 2 and 9--10), to conclude:
$$\vdash \aslightdef{h_0}{\triangleq}{x/[{-}4],\ z/[0]}\ \colcode{\fix{\r}}\ \aslightdef{k_0}{\triangleq k_1\ \Delta \ h' =}{x/[{-}4,5],\ z/[{-}16,16]}\ .$$

This derivation reproduces, step by step, the invariant computed by standard interval abstract
interpretation of $\q$ with widening and narrowing, spurious negative values $z\in[{-}16,{-}1]$
included.  The inaccuracy stems from the compositional abstract 
evaluation of $x*x$ as $[x]*^{\sharp}[x]$, which forgets
that the two operands are the same variable and thus carry the same value.

\begin{figure}[!t]
\resizebox{\textwidth}{!}{%
\begin{tikzpicture}[
  jd/.style ={anchor=west, inner sep=1.6pt},
  lbl/.style={anchor=east, inner sep=1.6pt},
  >=Stealth,
]
\def\ind{.75}\def\rh{0.65}\def\base{-0.2}
\pgfmathsetmacro\arl{0.30}
\pgfmathsetmacro\txtgap{0.42}

\newcommand{\trunk}[5]{%
  \pgfmathsetmacro\cx{\base + #2*\ind - 0.16}%
  \pgfmathsetmacro\ty{-#3*\rh + 0.32}%
  \pgfmathsetmacro\by{-#4*\rh}%
  \draw[#1, line width=0.8pt] (\cx,\ty) -- (\cx,\by);%
  \foreach \rrow in {#5}{%
    \pgfmathsetmacro\ry{-\rrow*\rh}%
    \draw[#1, line width=0.8pt]
      (\cx,\ry) -- (\cx+\arl,\ry);%
  }%
}
\newcommand{\row}[4]{%
 \node[lbl] at (-2,-#1*\rh) {\scriptsize #1};%
  \node[lbl] at (0,-#1*\rh) {\rm #3};%
  \node[jd]  at (\base+#2*\ind+\txtgap,-#1*\rh) {$#4$};%
}

\trunk{lvlA}{1}{2}{10}{2,9,10}   %
\trunk{lvlB}{2}{3}{4}{3,4}     %
\trunk{lvlC}{3}{5}{5}{5}     %
\trunk{lvlD}{4}{6}{6}{6}         %
\trunk{lvlE}{5}{7}{8}{7,8}    

\row{1}{0}{(narrow)}{\vdash \aslightdef{h_0}{\triangleq}{x/[{-}4],\ z/[0]}\ \colcode{\fix{\r}}\ \aslightdef{k_0}{\triangleq k_1\ \Delta \ h' =}{x/[{-}4,5],\ z/[0,16]}}
\row{2}{1}{(widen)}{{\color{lvlA} \mathbf{t}_1  \triangleq \ } \vdash \aslightdef{h_0}{=}{x/[{-}4],\ z/[0]}\ \colcode{\fix{\r}}\ \aslightdef{k_1}{=}{x/[{-}4,{+}\infty],\ z/[0,{+}\infty]}}
\row{3}{2}{(choice*)}{{\color{lvlB} \mathbf{t}_2 \triangleq }\ \vdash \aslightdef{h_0}{=}{x/[{-}4],\ z/[0]}\ \colcode{\1 + \r}\ \aslightdef{h'_0}{\triangleq}{x/[{-}4,{-}3],\ z/[0,16]}}
\row{4}{2}{(inv-fix)}{{\color{lvlB} \mathbf{t}_3  \triangleq \ } \vdash \aslightdef{k_1}{\triangleq h_0\nabla h'_0=}{x/[{-}4,{+}\infty],\ z/[0,{+}\infty]}\ \colcode{\fix{\r}}\ \aslightdef{k_1}{=}{x/[{-}4,{+}\infty],\ z/[0,{+}\infty]}}
\row{5}{3}{(cons)}{{\color{lvlC} \mathbf{t}_4 \triangleq }\ \vdash \aslightdef{k_1}{=}{x/[{-}4,{+}\infty],\ z/[0,{+}\infty]}\ \colcode{\r}\ \aslightdef{k_1}{=}{x/[{-}4,{+}\infty],\ z/[0,{+}\infty]}}
\row{6}{4}{(join)}{{\color{lvlD} \mathbf{t}_5 \triangleq }\ \vdash \aslightdef{k_1}{=}{x/[{-}4,{+}\infty],\ z/[0,{+}\infty]}\ \colcode{\r}\ \aslightdef{k'_1}{\triangleq}{x/[{-}3,5],\ z/[0,16]}}
\row{7}{5}{(seq)+(basic)}{{\color{lvlE} \mathbf{t}_6 \triangleq} \  \vdash \aslightdef{b_1}{=}{x/[{-}4,0],\ z/[0,{+}\infty]}\ \colcode{\r}\ \aslightdef{b'_1}{\triangleq}{x/[{-}3,1],\ z/[0,16]}}
\row{8}{5}{(seq)+(basic)}{{\color{lvlE} \mathbf{t}_7 \triangleq} \  \vdash \aslightdef{b_2}{=}{x/[0,{+}\infty],\ z/[0,{+}\infty]}\ \colcode{\r}\ \aslightdef{b'_2}{\triangleq}{x/[1,5],\ z/[0,16]}}
\row{9}{1}{(join*)}{{\color{lvlA} \mathbf{t}_5 \triangleq \ } \vdash \aslightdef{k_1}{=}{x/[{-}4,{+}\infty],\ z/[0,{+}\infty]}\ \colcode{\r}\ \aslightdef{k'_1}{=}{x/[{-}3,5],\ z/[0,16]}}
\row{10}{1}{}{{\color{lvlA} h'\triangleq}\ h_0\sqcup k'_1 =\{\,x/[{-}4,5],\ z/[0,16]\,\}}

\end{tikzpicture}%
}
\caption{Improved derivation for the interval 
analysis in $\APPL$. Side condition for (join) on lines 6 and 9: $\set{b_1, b_2}$ dense, $b_1 \cup b_2 = k_1$, and $b'_1 \cup b'_2 = k'_1$.}
\label{fig:linear-derivation2}
\end{figure}

\paragraph{The {\rm (join)} rule improves the invariant.}
In Figure~\ref{fig:linear-derivation1}, the loss of precision occurs in the derivations at lines 5, 8 and 9 for the triples $\mathbf{s}_4$ and $\mathbf{s}_7$, where it is calculated that $[{-}4,4]*^{\sharp}[{-}4,4]=[{-}16,16]$, yielding $z\in[{-}16,16]$.
We improve the precision by using the (join) rule. The set 
\begin{center} %
  $\big\{\, b_1 \triangleq \{\,x/[{-}4,{0}],\ z/[0,{+}\infty]\,\},\ b_2 \triangleq \{\,x/[0,{+}\infty],\ z/[0,{+}\infty]\,\}\,\big\}$
\end{center}
is dense in $\Int$, as its interval lub $k_1=\{\,x/[{-}4,{+}\infty],\ z/[0,{+}\infty]\,\}$ coincides with the union of its
concretisations.
No check is needed here: $b_1 = k_1 \cap \{x \le 0\}$ and
$b_2 = k_1 \cap \{x \ge 0\}$ split $k_1$ along the half-spaces $x \le 0$ and
$x \ge 0$, and such a split always yields a dense decomposition of $k_1$
(Proposition~\ref{pr:template-split}).
Therefore, we can analyse the two separate branches with preconditions $b_1$, $b_2$ (lines~7--8), getting the judgements 
$\color{lvlE} \mathbf{t}_6$ and $\color{lvlE} \mathbf{t}_7$, which are joined to infer
$${\color{lvlD} \mathbf{t}_5 \triangleq }\ \vdash \aslightdef{k_1}{=}{x/[{-}4,{+}\infty],\ z/[0,{+}\infty]}\ \colcode{\r}\ \aslightdef{k'_1}{\,\triangleq}{x/[{-}3,5],\ z/[0,16]} .$$

On each sign-definite piece $[{-}4,0]$, $[0,4]$ the square is exact and non-negative, so their join
certifies $z\in[0,16]$, namely that $z$ cannot be negative, which the compositional evaluation misses.
Now, we have that
$\{\,x/[{-}3,5],\ z/[0,16]\,\}\sqsubseteq\{\,x/[{-}4,{+}\infty],\ z/[0,{+}\infty]\,\}$. Thus, in contrast
to $\mathbf{s}_4$, the \emph{first} widening limit $\{\,x/[{-}4,{+}\infty],\ z/[0,{+}\infty]\,\}$ is
already an invariant of $\r$.

\medskip
Overall, an application of rule (join), through the user-supplied dense decomposition of the
precondition $[{-}4,{+}\infty]$ into $[{-}4,0]$ and $[0,{+}\infty]$, improves the imprecise step by keeping $z$ non-negative throughout, and lets $\APPL$ infer the optimal interval invariant
$\{\,x/[{-}4,5],\ z/[0,16]\,\}$. 
Standard interval analysis cannot recover $[0,16]$,
as it would require a dedicated optimal abstraction of squaring, or a domain refinement such as trace partitioning or disjunctive completion. 
Our program logic $\APPL$ attains it with rule (join), which
achieves local state-level disjunctive precision within one derivation,
without changing the abstract domain.
We will return to this point in \S~\ref{sec:related}.

\section{Semantic Model}
\label{sec:language}

We introduce the core language and its semantics over complete lattices, called \emph{interpretation domains}, which serve as the reference model for the Hoare-style logic developed 
in \S~\ref{sec:logic}.
An interpretation domain is endowed with a \emph{join basis} that controls the granularity of the semantics: a larger basis yields a coarser semantics (Proposition~\ref{pr:hierarchy}). This basis underlies the central notion of \emph{density}, which copes with the non-additivity of the semantics over non-disjunctive abstract domains: on dense sets the semantics is additive (Lemma~\ref{le:add-dense}).

\paragraph{Order theoretic notions.}
Posets are denoted
$(\lat,\sqsubseteq)$, with join and meet of a subset $X\subseteq \lat$ denoted by 
$\sqcup X$ and $\sqcap X$, respectively, whenever they exist.
The down closure of a subset $X\subseteq \lat$ is defined as $\down X \triangleq \{d\in \lat \mid \exists x\in X.\, d\sqsubseteq x\}$.
If join and meet exist for any subset $X\subseteq \lat$, then $\lat$ is a complete lattice, and 
$\bot$ and $\top$ denote its bottom and top elements, respectively.  

A function $f: D \to D$ is monotone when for $d, d' \in D$, if $d \sqsubseteq d'$ then $f(d) \sqsubseteq f(d')$,
while it is additive (co-additive) if, for any subset $X \subseteq D$, 
$f(\sqcup X) = \bigsqcup_{d \in X}f(d)$ ($f(\sqcap X) = \bigsqcap_{d \in X}f(d)$). 
By Knaster-Tarski's theorem, any
monotone function on a complete lattice has a least fixpoint, denoted
$\lfp{f}$, characterised as the meet of all post-fixpoints
$\lfp_\sqsubseteq(f) = \sqcap \{ d \mid f(d) \sqsubseteq d \}$.  It
can also be obtained by iterating the function on the bottom element
of the lattice, $\lfp_\sqsubseteq(f) = \bigsqcup_{\alpha} f^\alpha(\bot)$, 
where $\alpha$ ranges over ordinal numbers. 
For continuous functions, where $\alpha$ may range over natural numbers, this 
is often referred to as Kleene's iteration, and it is a
major pillar of abstract interpretation~\cite{CC:CCVTFP}.

\paragraph*{Abstract interpretation basics.}
Given two complete lattices $C$ and $A$, called the
\emph{concrete} and the \emph{abstract}
domain, respectively, a Galois connection (GC)
$\langle \alpha,\gamma \rangle: C \rightleftarrows A$ is a pair of
functions $\alpha : C \to A$ and $\gamma : A \to C$ s.t.\
$\alpha(c) \sqsubseteq_A a$
$\Leftrightarrow$ $c \sqsubseteq_C \gamma(a)$ for any $c\in C$ and $a\in A$.
When $\alpha(C)=A$, this is called Galois insertion.
The function $\alpha$ is referred to as abstraction map and turns out to be
additive, while $\gamma$ is the concretisation map which is always co-additive.
Intuitively, any abstract element $a\in A$ such that $c\sqsubseteq
\gamma(a)$ is a sound over-approximation for the concrete value $c$, while the
abstraction $\alpha(c)$ is the most precise over-approximation of
$c$ in the abstract domain $A$, i.e.,
$\alpha(c) = \bigsqcap_A \set{a \mid c \sqsubseteq_C \gamma(a)}$ holds.

Given a function $f: C^n
\to C$, an abstract function $f^{\sharp}: A^n \to A$ is a
correct approximation of $f$ if $\alpha \circ f \sqsubseteq
f^{\sharp}\circ \alpha$ (where $\alpha$ is
applied componentwise to tuples), and it is
a \emph{complete} approximation of $f$ when $\alpha
\circ f = f^{\sharp}\circ \alpha$. Any concrete operation $f : C^n \to C$ admits a canonical abstract 
counterpart known as its \emph{best correct approximation} (BCA): it is denoted $f^A : A^n \to A$ and defined by
$f^A(a_1, \ldots, a_n) \triangleq \alpha(f(\gamma(a_1), \ldots, \gamma(a_n)))$. 

An operator $\nabla: A \times A \to A$ is a \emph{widening} if (i) it computes upper bounds, i.e., 
$\forall a, a' \in A : a \sqcup a' \sqsubseteq_A a\  \nabla \ a'$ 
and (ii) it enforces convergence, i.e. for any $(b_i)_{i \in \N}$ in $A$ the sequence $(a_i)_{i \in \N}$ defined by $a_0 \triangleq b_0$, $a_{i+1} \triangleq a_i \ \nabla \ b_{i+1}$ is 
ultimately stationary, i.e., $\exists k \geq 0$ such that $\forall j\geq k$, $a_{j+1} = a_j$. 

An operator $\Delta: A \times A \to A$ is a \emph{narrowing} if (i) 
$\forall a, a' \in A :  a \sqcap a' \sqsubseteq_A a \ \Delta \ a' \sqsubseteq_A a$, 
and (ii)~for any $(b_i)_{i \in \N}$ in $A$ the sequence $(a_i)_{i \in \N}$ defined by $a_0 \triangleq 
b_0$, $a_{i+1} = a_i \ \Delta \ b_{i+1}$ is 
ultimately stationary. 

\subsection{Command Syntax and Collecting Semantics}\label{ss:commands}
We consider a core 
language of \emph{regular commands} 
 $\r \in \Reg$ defined as
\begin{center}
     $\r := \e \ \mid\ \r;\r \ \mid\ \r + \r \ \mid\ \r^* \mid\ \fix{\r}$
 \end{center}
where $\e\in \Ecom$ is a fixed set of elementary commands, assumed to always
include an \emph{identity} command $\1$.
Sequential composition associates to the left, i.e., $\r_1;\r_2;\r_3$ stands for
$(\r_1;\r_2);\r_3$. This matters because sequential composition need not be associative in 
our semantics (Example~\ref{ex:notassociative}).
As
elementary commands we use assignments $x := \mathit{exp}$ and
Boolean filters $b?$, where $\mathit{exp}$ is an arithmetic expression and $b$ a
Boolean one. 
Standard control-flow constructs are derived as usual:
\begin{align*}
  \hspace{-0pt} 
   \textbf{if}~b~\textbf{then}~\r_1~\textbf{else}~\r_2 %
    \triangleq (b? ; \r_1) + (\neg b? ; \r_2) %
    \qquad
    \textbf{while}~b~\textbf{do}~\r  %
     \triangleq \fix{b?;\r};\neg b?
\end{align*}

We write $\r^0 \triangleq \1$ and $\r^{i+1} \triangleq \r^i ; \r$, so
that $\r^i$ is the $i$-fold sequential composition of $\r$.
Besides
Kleene's star, we also include a least fixpoint construct $\fix{\r}$,
whose semantics is the least inductive invariant of $\r$ above the
input. In the concrete semantics 
this construct
is equivalent to the
Kleene star, but this fails in the abstract semantics (cf.\
Example~\ref{ex:fix-and-star}). Keeping both lets us show that $\r^*$ is more precise than $\fix{\r}$ but less tractable: $\r^*$
requires an infinitary iteration rule for logical completeness,
whereas finitary iteration rules suffice for $\fix{\r}$
(cf. Corollary~\ref{co:remove-iter} and Example~\ref{ex:incompleteness}).

\begin{example}[Collecting semantics]\label{ex:concr-semantics}
  We consider an imperative language with integer
  variables $x, y, z, \ldots$ ranging over $\Var$, and program states
  $\sigma, \tau, \ldots$ ranging over
  $\Sigma \triangleq \Var \to \mathbb{Z}$. For a state $\sigma$,
  variable $x$, and value $v \in \mathbb{Z}$, we write
  $\subst{\sigma}{v}{x}$ to denote a state update. The generalisation to other value types (e.g. bounded integers, floating-point numbers, or rationals) is straightforward.
  The semantics of
  elementary commands is assumed to be given as a function $\bsem{\e} : \Sigma \to \pow{\Sigma}$.
  Then, for $\r \in \Reg$, we define $\bsem{\r} : \Sigma \to \pow{\Sigma}$ inductively as follows:
  \vspace{3pt}
  \begin{center}
    \begin{tabular}{l l }
          $\bsem{\r_1;\r_2} \sigma
    \triangleq \bigcup_{\sigma' \in \bsem{\r_1}\sigma} \bsem{\r_2} \sigma'$
    &
    $\bsem{\r_1 + \r_2} \sigma
    \triangleq \bsem{\r_1} \sigma \cup \bsem{\r_2} \sigma$ \\[1ex]
    $\bsem{\r^*} \sigma  
    \triangleq \bigcup_{i \geq 0} \bsem{\r^i}\sigma$
   & %
    $\bsem{\fix{\r}} \sigma  
    \triangleq \lfp_{\subseteq} \big( \lambda X.\ \{\sigma\} \cup \bigcup_{\sigma' \in X} \bsem{\r}(\sigma')\big)$
    \end{tabular}
  \end{center}
  \vspace{3pt}

This inductive definition extends pointwise to sets of stores. Specifically, for any $X \subseteq \Sigma$, the collecting semantics $\sem{\r}: \pow{\Sigma} \to \pow{\Sigma}$ is defined as
$\sem{\r} X \triangleq  \bigcup_{\sigma \in X} \bsem{\r}\sigma$. \hfill $\lozenge$
\end{example}

\subsection{Semantics on Interpretation Domains}

The principle of defining the semantics on the full domain by extending the one defined just on ``basic'' elements generalises naturally to any complete lattice. In the powerset
lattice, the basic elements are the singletons $\set{\sigma} \in \pow{\Sigma}$, one for each state $\sigma \in \Sigma$, because they
generate every other set by join. 
We make the basic elements a parameter of the model, called a \emph{basis}.

\begin{definition}[(Join) Basis]\label{def:basis}
Let $(\lat,\sqsubseteq)$ be a complete lattice. A (\emph{join}) \emph{basis} for
$\lat$ is a subset $\base \subseteq \lat$ such that each $d \in \lat$ is the join of
the basis elements below it, i.e., $d = \sqcup\, \ebase{d}$, where
$\ebase{d} \triangleq \set{b \in \base \mid b \sqsubseteq d}$. 
\hfill $\lozenge$
\end{definition}

The whole lattice $\lat$ is trivially the largest basis. It is natural to
seek minimal bases consisting of elements that cannot be further decomposed by joins. Recall
that $d\in\lat$ is \emph{completely join-irreducible} if $d=\sqcup X$ implies $d\in X$,
for every $X\subseteq\lat$. Let $\Ir{\lat}$ denote the set of completely
join-irreducible elements. Every basis contains $\Ir{\lat}$: if $d$ is completely
join-irreducible then $d=\sqcup\,\ebase{d}$ forces $d\in\ebase{d}\subseteq\base$.
Hence, whenever $\Ir{\lat}$ is itself a basis, it is the unique minimal one, called the \emph{irreducible basis} of $\lat$. An \emph{interpretation domain} is then
a complete lattice with a chosen basis.

\begin{definition}[Interpretation domain]\label{de:based-lat}
  An \emph{interpretation domain} is a tuple $(D, \sqsubseteq, \base)$ where
  $(D, \sqsubseteq)$ is a complete lattice and $\base$ is a basis for~$D$.
  \hfill$\lozenge$
\end{definition}

Mimicking the collecting semantics of Example~\ref{ex:concr-semantics}, the command
semantics over an interpretation domain $(D, \sqsubseteq, \base)$ is first defined on
the basis, as monotone functions $\bsem[\base]{\r} : \base \to D$, and then lifted to
all of $D$, which is join-generated by $\base$. Indeed, any monotone map
$f : \base \to D$ induces a monotone \emph{join extension} $\add{f}: D \to D$
(Lemma~\ref{le:additive-sq}), defined by
$\add{f}(d) \triangleq \bigsqcup_{b \in \ebase{d}} f(b)$.

\begin{toappendix}
\begin{lemmarep}[Monotonicity of additive extensions]\label{le:additive-sq}
  Let $f : \base \to D$ be a function monotone w.r.t.  $\sqsubseteq$.
  Then also $\add{f} : D \to D$ is monotone w.r.t.  $\sqsubseteq$.
\end{lemmarep}

\begin{proof}
Let $d, d' \in D$ be such that $d \sqsubseteq d'$. Then $\ebase{d} \subseteq \ebase{d'}$ and thus it follows that: 
  \begin{align*}
    &f(d)  = \bigsqcup_{b \in \ebase{d}} f(b) %
    \sqsubseteq \bigsqcup_{b' \in \ebase{d'}} f(b') = \add{f}(d') \tag*{\qedhere}
  \end{align*}
\end{proof}
\end{toappendix}
 
\begin{definition}[Basis and Full Semantics]\label{de:semantics}
  Let $(D, \sqsubseteq, \base)$ be an interpretation domain, and assume a monotone
  semantics $\bsem[\base]{\e} : \base \to D$ for the elementary commands
  $\e \in \Ecom$ is given, with $\bsem[\base]{\1} \triangleq \lambda x . x$.
  For $\r \in \Reg$ and
  $b \in \base$, the \emph{basis semantics} $\bsem[\base]{\r}: \base \to D$ is defined
  inductively by \vspace{3pt}
  \begin{center}
    \begin{tabular}{l l}
    $\bsem[\base]{\r_1;\r_2} b
      \triangleq \add{\bsem[\base]{\r_2}}(\bsem[\base]{\r_1}b)$ &
    $\bsem[\base]{\r_1 + \r_2} b
      \triangleq \bsem[\base]{\r_1} b \sqcup \bsem[\base]{\r_2} b$ \\[1ex]
    $\bsem[\base]{\r^*} b
      \triangleq \bigsqcup_{i \geq 0} \bsem[\base]{\r^i}b$ &
    $\bsem[\base]{\fix{\r}} b
      \triangleq \lfp_{\sqsubseteq} \big( \lambda x.\ b \sqcup \add{\bsem[\base]{\r}}x\big)$
    \end{tabular}
  \end{center}
  \vspace{3pt}

  Each $\bsem[\base]{\r}$ is monotone (thus join extensions and the least fixpoint  are
  well defined) and the \emph{full semantics} $\sem[\base]{\r} \triangleq \add{\bsem[\base]{\r}} : D \to D$ is the
  join extension of $\bsem[\base]{\r}$,
  i.e.,
  $\sem[\base]{\r}d \triangleq \bigsqcup_{b \in \ebase{d}} \bsem[\base]{\r}b$.
  \hfill $\lozenge$
\end{definition} 
 
When the basis $\base$ is clear from context we omit the subscript $\base$. Two key
instances are obtained by taking the smallest and largest bases.

\begin{instance}[Irreducible basis]\label{inst:irr-collecting}
  If a complete lattice $(D, \sqsubseteq)$ admits the irreducible basis $\Ir{D}$, then we have that
  $(D, \sqsubseteq, \Ir{D})$ is an interpretation domain. An instance is $(\pow{\Sigma}, \subseteq, \Ir{\pow{\Sigma}})$, with
  $\Ir{\pow{\Sigma}} = \set{ \set{\sigma} \mid \sigma \in \Sigma}$. The corresponding
  semantics is the ordinary collecting (or Hoare) semantics of
  Example~\ref{ex:concr-semantics}, giving the strongest postcondition of commands:
  for all $\e\in \Ecom$ and $\{\sigma\}\in \Ir{\pow{\Sigma}}$, we set
  $\bsem{\e}\set{\sigma} \triangleq \post{\e}(\set{\sigma})$. 
  
  For assignments and
  Boolean filters, the strongest postcondition is as usual:
  \begin{center}
    $\post{x:=\mathit{exp}}(\set{\sigma}) \triangleq
      \set{\subst{\sigma}{\sigma(\mathit{exp})}{x}}, \qquad
    \post{b?}(\set{\sigma}) \triangleq
      \begin{cases} \set{\sigma} & \text{if } \sigma(b)=\mathit{tt}\\
                    \varnothing  & \text{otherwise,}\end{cases}$
  \end{center}
  where $\sigma(\mathit{exp})\in \mathbb{Z}$ and $\sigma(b)\in \{\mathit{tt},\mathit{ff}\}$
  denote the evaluation of $\mathit{exp}$ and of $b$ in $\sigma$. For all
  $\r\in \Reg$, $\bsem{\r} : \Ir{\pow{\Sigma}} \to \pow{\Sigma}$ is the semantics on
  singleton states, extended additively to sets by $\sem{\r} \triangleq \add{\bsem{\r}}$.

  As a second example the abstract domain $\Sign$
  (Figure~\ref{fig:intro-sign}) has
  irreducible basis
  $\Ir{\Sign} = \set{\Z_{<0}, \Z_{=0}, \Z_{>0}}$. In general, for an
  interpretation domain $(D, \sqsubseteq, \Ir{D})$ with $D$ an
  abstract domain (e.g., signs, intervals, octagons, convex polyhedra,
  etc.), the elementary commands are interpreted as sound
  over-approximations of their strongest postcondition. For
  assignments and tests, such an over-approximation is typically
  obtained inductively over the structure of expressions (e.g.\ the
  compositional abstract sum for intervals:
  $[a,b] +^\sharp [c,d] = [a+c,b+d]$).
  Nevertheless, 
  up to \S~\ref{ss:ai-interpretation} we make no use of the
  (sound) relationship between abstract and concrete semantics. \hfill
  $\lozenge$
\end{instance}

\begin{instance}[Trivial basis]\label{inst:simple-inductive}
  For a complete lattice $(D, \sqsubseteq)$, the interpretation domain
  $(D, \sqsubseteq, D)$ takes the trivial basis $\base=D$. Here the join extension is
  vacuous, since $\ebase{d} = \down{d}$, and thus, by monotonicity,
  $\sem[D]{\r}d = \bigsqcup \set{\bsem[D]{\r}d' \mid d' \sqsubseteq d} = \bsem[D]{\r}d$, i.e.,
  full and basis semantics coincide. This is the standard inductive semantics of
  regular commands over $D$ (e.g.,~\cite{OHearn20,KozenKAT,cousot21}), satisfying
  all the usual equations below:
  \begin{align*}
      &\sem[D]{\e} d=\bsem[D]{\e}d &
      &\sem[D]{\r_1;\r_2} d= \sem[D]{\r_2}(\sem[D]{\r_1}d) &
      &\sem[D]{\r_1 + \r_2} d=\sem[D]{\r_1} d \sqcup \sem[D]{\r_2} d\\
      &\sem[D]{\r^*} d = \bigsqcup_{i \geq 0} \sem[D]{\r^i} d &
      &\sem[D]{\fix{\r}} d = \lfp_{\sqsubseteq} \big( \lambda x.\ d \sqcup \sem[D]{\r}x\big)&
      & \hfill \tag*{$\lozenge$}
  \end{align*}
\end{instance}

With elementary commands interpreted by their strongest postcondition,
the interpretation domain $(\pow{\Sigma}, \subseteq, \pow{\Sigma})$ again induces the semantics of
Example~\ref{ex:concr-semantics}. Likewise, for $\Sign$
(Fig.~\ref{fig:intro-sign}) the irreducible and trivial bases yield the same
semantics. This is no accident: when the domain is disjunctive and the elementary
commands are additive, the choice of the basis does not affect the semantics
(Proposition~\ref{pr:inductive}). In general, however, different bases can yield different
semantics.

\begin{example}[Different bases, different semantics]\label{ex:signs-hole}
  Let $\mathsf{p} \triangleq (x \neq 0?);\,(x = 0?)$ over an integer variable $x$, on
  the $\Sign$ domain without $\Z_{\neq 0}$ (depicted in Fig.~\ref{fig:intro-signs-no-neq}). On
  input $\Z$, the irreducible basis $\Ir{\Sign} = \set{\Z_{<0}, \Z_{=0}, \Z_{>0}}$
  gives
  \begin{align*}
    \sem{\mathsf{p}}\Z
    &= \bsem{\mathsf{p}}\Z_{<0}\sqcup \bsem{\mathsf{p}}\Z_{=0} \sqcup \bsem{\mathsf{p}}\Z_{>0}
    = \add{\bsem{(x{=}0?)}}\Z_{<0} \sqcup \add{\bsem{(x{=}0?)}}\varnothing
       \sqcup \add{\bsem{(x{=}0?)}}\Z_{>0}
    = \varnothing,
  \end{align*}
  since $\bsem{(x{\neq}0?)}$ maps $\Z_{<0}$ and $\Z_{>0}$ to themselves and $\Z_{=0}$ to
  $\varnothing$, and $x{=}0?$ then filters out $\Z_{<0}$ and $\Z_{>0}$. The trivial basis
  $\base = \Sign$ instead gives the coarser standard sign analysis:
  \[
    \sem{(x{=}0?)}\big(\sem{(x{\neq}0?)}\Z\big) = \sem{(x{=}0?)}\Z = \Z_{=0}.
    \tag*{$\lozenge$}
  \]
\end{example}

We use the classical domain of intervals \cite{CC77} as running example, to illustrate properties
of the semantics induced by different bases and the corresponding logical derivations.

\begin{example}[Intervals]\label{ex:intervals}
  Let $\Int \triangleq \{ \interval{l}{u} \mid l \in \Z \cup \{-\infty\},\, u \in \Z
  \cup \{+\infty\},\, l\leq u \} \cup \set{\varnothing}$ be the complete lattice of
  integer intervals, with irreducible basis
  $\Ir{\Int} = \set{\interval{z}{z} \mid z \in \Z}$. For notational simplicity we
  sometimes take a single variable, $\Var = \set{x}$, so that $\Sigma = \Z$. The usual
  interval abstraction $\alpha: \pow{\Z} \rightarrow \Int$ is
  $\alpha(S) \triangleq \interval{\inf S}{\sup S}$ (with
  $\alpha(\varnothing) = \varnothing$), and the concretisation
  $\gamma: \Int \rightarrow \pow{\Z}$ is 
  $\gamma(\interval{l}{u}) \triangleq \set{ z \in \Z \mid l \le z \le u}$, where both extend
  componentwise to several variables. The interval semantics of
  $\e\in \Ecom$ is $\sem{\e}a \triangleq \alpha(\post{\e}(\gamma(a)))$ for $a \in \Int$.
  When the trivial basis is chosen (\S~\ref{inst:simple-inductive}), 
  $(\Int, \sqsubseteq, \Int)$ recovers the inductive abstract interval
  semantics~\cite{cousot21}, while with the irreducible basis
  $(\Int, \sqsubseteq, \Ir{\Int})$ (\S~\ref{inst:irr-collecting}) it yields a strictly
  finer one (cf.\ Example~\ref{ex:interval-BCA}). \hfill $\lozenge$
\end{example}

More generally, for a fixed complete lattice $(D, \sqsubseteq)$, smaller bases induce
finer semantics, yielding a hierarchy that can be strict (cf.\ Example~\ref{ex:hierarchy}).
\begin{propositionrep}[Hierarchy of semantics]
  \label{pr:hierarchy}
  Let $(D, \sqsubseteq)$ be a complete lattice and $\base \subseteq \base' \subseteq D$
  bases with $\bsem[\base]{\e}$ the restriction of $\bsem[\base']{\e}$ to $\base$. Then,
  $\sem[\base]{\r} d \sqsubseteq \sem[\base']{\r} d$ for all $d \in D$.
\end{propositionrep}

\begin{proof}
  We first observe that $\base \subseteq \base'$ implies
  $\ebase{d} \subseteq \base'(d)$ for all $d \in D$.  Also notice that
  it is sufficient to prove that
  $\bsem[\base]{\r}b \sqsubseteq \bsem[\base']{\r}b$ for all
  $b \in \base$, since then
  \begin{align*}
    \sem[\base']{\r} d & = \bigsqcup_{b' \in \base'(d)} \bsem[\base']{\r}b'\\
    & \sqsupseteq \bigsqcup_{b \in \ebase{d}} \bsem[\base']{\r}b \\ 
    & \sqsupseteq \bigsqcup_{b \in \ebase{d}} \bsem[\base]{\r}b \\
    & = \sem[\base]{\r} d
  \end{align*}
  We then prove that for all $b \in B$,
  $\bsem[\base]{\r}b \sqsubseteq \bsem[\base']{\r}b$ by induction on
  the command syntax: \begin{itemize}
  \item ($\r = \e$)\\
    Then $\bsem[\base]{\e} b = \bsem[\base']{\e} b$  by hypothesis.

    \smallskip
    
  \item ($\r = \r_1 + \r_2$)\\
    By inductive hypothesis
    $\bsem[\base]{\r_i} b \sqsubseteq \bsem[\base']{\r_i} b$ for
    $i = 1, 2$. Therefore:
    \[
      \bsem[\base]{\r_1 + \r_2} b = \bsem[\base]{\r_1} b \sqcup
      \bsem[\base]{\r_2} b \sqsubseteq \bsem[\base']{\r_1} b \sqcup
      \bsem[\base']{\r_2} b = \bsem[\base']{\r_1 + \r_2} b
    \]

    \smallskip
    
  \item ($\r = \r_1; \r_2$)\\
    We have
    \begin{align*}
      \bsem[\base]{\r_1;\r_2} b
      &  = \bigsqcup_{b' \in \ebase{\bsem[\base]{\r_1} b}} \bsem[\base]{\r_2} b' \\
      & \sqsubseteq \bigsqcup_{b' \in \ebase{\bsem[\base']{\r_1} b}} \bsem[\base]{\r_2} b'
      & \mbox{[by ind. hyp. $\bsem[\base]{\r_1} b \sqsubseteq \bsem[\base']{\r_1}b$]} \\
      & \sqsubseteq \bigsqcup_{b' \in \ebase{\bsem[\base']{\r_1} b}} \bsem[\base']{\r_2} b'
      & \mbox{[by ind. hyp. $\bsem[\base]{\r_2} b' \sqsubseteq \bsem[\base']{\r_2}b'$]}\\
      & \sqsubseteq \bigsqcup_{b' \in \base'(\bsem[\base']{\r_1} b)} \bsem[\base']{\r_2} b'
      & \mbox{[$\base \subseteq \base'$, thus $\ebase{\bsem[\base']{\r_1} b} \subseteq \base'(\bsem[\base']{\r_1} b)$]} \\
      & = \bsem[\base']{\r_1;\r_2}b
    \end{align*}

    \smallskip
    
  \item ($\r = \r_1^*$)\\
    By inductive hypothesis $\bsem[\base]{\r_1} b \sqsubseteq \bsem[\base']{\r_1} b$. 
    We show by induction on $i$ that for all $i \in \N$ it holds $\bsem[\base]{\r_1^i} b \sqsubseteq \bsem[\base']{\r_1^i} b$. 
    \begin{itemize}
    \item $i = 0$: $\bsem[\base]{\r_1^0} b = b = \bsem[\base']{\r_1^0} b$. 
    \item $i \rightarrow i +1$: as in the sequential composition case:
      \begin{align*}
        \bsem[\base]{\r_1^{i+1}} b & = \add{\bsem[\base]{\r_1}}(\bsem[\base]{\r_1}b)\\
        &  = \bigsqcup_{b' \in \ebase{\bsem[\base]{\r_1^i} b}} \bsem[\base]{\r_1} b'\\
        & \sqsubseteq \bigsqcup_{b' \in \ebase{\bsem[\base]{\r_1^i} b}} \bsem[\base']{\r_1} b'
          & \mbox{[by outer ind. hyp. $\bsem[\base]{\r_1} b' \sqsubseteq \bsem[\base']{\r_1}b'$]}\\
        & \sqsubseteq \bigsqcup_{b' \in \ebase{\bsem[\base']{\r_1^i} b}} \bsem[\base']{\r_1} b'
        & \mbox{[by  ind. hyp. $\bsem[\base]{\r_1^i} b' \sqsubseteq \bsem[\base']{\r_1^i}b'$]}
        \\
        & \sqsubseteq \bigsqcup_{b' \in \base'(\bsem[\base']{\r_1^i} b)} \bsem[\base']{\r_1} b'
        & \mbox{[$\base \subseteq \base'$, thus $\ebase{\bsem[\base']{\r_1^i} b} \subseteq \base'(\bsem[\base']{\r_1^i} b)$]}
         \\
        & = \add{\bsem[\base']{\r_1}}(\bsem[\base']{\r_1}b) = \bsem[\base']{\r_1^{i+1}}b
      \end{align*}
    \end{itemize}

    \smallskip
    
  \item ($\r = \fix{\r_1}$)\\
    By inductive hypothesis, for all $b \in B$, we have
    $\bsem[\base]{\r_1} b \sqsubseteq \bsem[\base']{\r_1} b$.
    Define
    the functions
    $\psi_{\r_1, b, \tilde{\base}} = \lambda x. b \sqcup
    \add{\bsem[\tilde{\base}]{\r_1}}x$. We have    
    $\psi_{\r_1, b, \base} \sqsubseteq \psi_{\r_1, b, \base'}$
    and therefore
    $\sem[\base]{\fix{\r_1}} b = \lfp(\psi_{\r_1, b, \base})
    \sqsubseteq \lfp(\psi_{\r_1, b, \base'}) = \sem[\base']{\fix{\r_1}}
    b$. \hfill    \qedhere
  \end{itemize}
\end{proof}

\begin{example}[Hierarchy of semantics]\label{ex:hierarchy}
  Consider $\s \triangleq y := x;\ w := y - x$ over the interval
  domain, and let $\base_k = \set{ I \in \Int \mid 1 \leq |I| \leq k }$ for
  some given $k \geq 1$, so that
  $\Ir{\Int} = \base_1 \subseteq \base_2 \subseteq \dots \subseteq
  \Int$ is an increasing chain of bases, with progressively coarser
  induced semantics. Writing an abstract state as a triple
  $(I_x, I_y, I_w)$ of intervals and
  $a_\top \triangleq (\top, \top, \top)$, where
  $\top = [-\infty,+\infty]$, we have that:
  \begin{align*}
    \sem[\base_1]{\s} a_\top = (\top, \top, [0, 0]), \qquad
    \sem[\base_{k+1}]{\s} a_\top = (\top, \top, [-k, k]), \qquad
    \sem[\Int]{\s} a_\top = (\top, \top, \top).
    \tag*{$\lozenge$}
  \end{align*}
\end{example}

We note some semantic properties, later used to prove the soundness of the program logic.%

\begin{lemmarep}[Properties of $\sem{\cdot}$]\label{le:kleene-axioms}
  Let $(D, \sqsubseteq, \base)$ be an interpretation domain. For all
  $\r,\r_1,\r_2\in \Reg$ and $d \in D$ all the following relations hold:
  \vspace*{-10pt}
  \begin{multicols}{2}
  \begin{enumerate}[{\rm (1)}]
  \item\label{le:kleene-axioms:seq}
    $\sem[\base]{\r_1;\r_2}d \sqsubseteq \sem[\base]{\r_2}(\sem[\base]{\r_1}d)$

  \item\label{le:kleene-axioms:sum}
    $\sem[\base]{\r_1 + \r_2}d = \sem[\base]{\r_1}d \sqcup \sem[\base]{\r_2}d$
  \item\label{le:kleene-axioms:star}
    $\sem[\base]{\r^*}d = \bigsqcup_{i \in \N}\sem[\base]{\r^i} d$
  \item\label{le:kleene-axioms:fix-ind}
    $\sem[\base]{\fix{\r}}d \sqsubseteq \lfp_\sqsubseteq( \lambda x.\ d \sqcup \sem[\base]{\r}x )$
  \item\label{le:kleene-axioms:seq-assoc}
    $\sem{\r_1;(\r_2;\r_3)}d \sqsubseteq \sem{(\r_1;\r_2);\r_3}d$
  
  \item\label{le:kleene-axioms:fix}
    $\sem[\base]{\r^*}d \sqsubseteq \sem[\base]{\fix{\r}}d$
\item\label{le:kleene-axioms:left-exp}
    $\sem[\base]{\r^*} \sqsubseteq \sem[\base]{\1 + (\r^*; \r)} $
  \item\label{le:kleene-axioms:left-dist}
    $\sem[\base]{\r; (\r_1 + \r_2)} = \sem[\base]{(\r; \r_1) + (\r; \r_2)}$
  \item\label{le:kleene-axioms:right-dist}
    $\sem[\base]{\r_1;\r + \r_2; \r} \sqsubseteq \sem[\base]{(\r_1 + \r_2);\r}$
  \end{enumerate}
  \end{multicols}
\end{lemmarep}
\begin{proof}%

  \ref{le:kleene-axioms:seq}
  For all $d \in D$
  \begin{align*}
    \sem{\r_1; \r_2}d 
    &= \add{\bsem{\r_1; \r_2}}d \\
    & 
    = \bigsqcup_{b \in \ebase{d}} \bsem{\r_1; \r_2}b \\
    & 
    = \bigsqcup_{b \in \ebase{d}} \add{\bsem{\r_2}} (\bsem{\r_1} b) \\
    & 
    \sqsubseteq \add{\bsem{\r_2}} \Bigl(\bigsqcup_{b \in \ebase{d}}  \bsem{\r_1} b\Bigl) \\
    & 
    =  \add{\bsem{\r_2}} (\add{\bsem{\r_1}} d) \\
    &
    = \sem{\r_2} (\sem{\r_1} d) 
  \end{align*}

  \bigskip

\noindent
  \ref{le:kleene-axioms:sum} For all  $d \in D$:
    \begin{align*}
      \sem{\r_1 + \r_2}d
      & = \add{\bsem{\r_1 + \r_2}}d\\
      & = \bigsqcup_{b \in \ebase{d}} \Bigl(\bsem{\r_1}b \sqcup \bsem{\r_2}b \Bigr)\\
      & =  \Bigl ( \bigsqcup_{b \in \ebase{d}} \bsem{\r_1}b \Bigr) 
      \sqcup \Bigl (\bigsqcup_{b \in \ebase{d}} \bsem{\r_2}b \Bigr) & [\mbox{associativity and commutativity of $\sqcup$}] \\ 
      & = \sem{\r_1}d \sqcup \sem{\r_2}d
    \end{align*}
    
    \bigskip

\noindent
    \ref{le:kleene-axioms:star} For all $d \in D$
    \begin{align*}
       \sem{\r^*}d 
      = \bigsqcup_{b \in \ebase{d}}\Biggl(\bigsqcup_{i\geq 0} \bsem{\r^i}b\Biggr) %
      = \bigsqcup_{i \geq 0} \bigsqcup_{b \in \ebase{d}} \bsem{\r^i}b %
      = \bigsqcup_{i \geq 0} \sem{\r^i}d
    \end{align*}
    
  \bigskip

\noindent    
\ref{le:kleene-axioms:fix-ind} For $d \in D$ define the function
$\psi_{\r, d} = \lambda x. d \sqcup \sem{\r}x$ and let $d' = \lfp_\sqsubseteq (\psi_{\r, d})$. We want to show that $\sem{\fix{\r}}d \sqsubseteq d'$.
Since $d' = d \sqcup \sem{\r} d'$ clearly $d \sqsubseteq d'$ and $\sem{\r}d' \sqsubseteq d'$.

Now, for all $b \in \base(d)$ we have 
\begin{center}
  $\psi_{\r,b}(d') = b \sqcup \sem{\r}d' \sqsubseteq d \sqcup \sem{\r}d' = d'$
\end{center}
i.e., $d'$ is a fixpoint of $\psi_{\r,b}$ and thus $\sem{\fix{\r}}d = \lfp_\sqsubseteq (\psi_{\r, d}) \sqsubseteq d'$. We conclude the desired inequality
\begin{center}  
  $\sem{\fix{\r}}d = \bigsqcup_{b \in \base(d)} \bsem{\fix{\r}}b \sqsubseteq d'$
\end{center}
  \bigskip
    
  \noindent For %
  \ref{le:kleene-axioms:seq-assoc}--\ref{le:kleene-axioms:right-dist}
  Below we just prove the results
  for the semantics on basis elements. The
  desired result immediately follows, since the full semantics $\sem{\cdot}$ 
  is the additive extension of the basis semantics.

  \bigskip

\noindent
  \ref{le:kleene-axioms:seq-assoc}
  Let $b \in B$, $d = \bsem[]{\r_1}b$, then \begin{align*}
  \bsem[]{\r_1; (\r_2; \r_3)}b & = \add{\bsem[]{\r_2;\r_3}}d = \bigsqcup_{b' \in \ebase{d}} \bsem[]{\r_2; \r_3}b' 
  =  \bigsqcup_{b' \in \ebase{d}} \add{\bsem[]{\r_3}}(\bsem[]{\r_2}b') \\
   \bsem[]{(\r_1; \r_2); \r_3}b & = \add{\bsem[]{\r_3}}(\bsem[]{\r_1; \r_2} b') = \add{\bsem[]{\r_3}}(\add{\bsem[]{\r_2} d}) = \add{\bsem[]{\r_3}} \left ( \bigsqcup_{b' \in \ebase{d}} \bsem[]{\r_2} b'\right )
\end{align*}
and we conclude $\bsem[]{\r_1; (\r_2; \r_3)} \sqsubseteq \bsem[]{(\r_1; \r_2); \r_3} $ from the fact that for all monotone functions $f$ and any family $\set{x_i}_{i \in I}$ it 
holds $\bigsqcup_{i \in I} f (x_i) \sqsubseteq f \left ( \bigsqcup_{i \in I} x_i \right )$.

    \bigskip

  \noindent
  \ref{le:kleene-axioms:fix} Let $b \in \base$ and define the function $\psi_{\r, b} = \lambda x. b \sqcup \add{\bsem{\r}}x$ and the sequence 
  $\psi_{\r, b}^0 = \bot$, $\psi_{\r, b}^{i+1} = \psi_{\r, b}(\psi_{\r, b}^{i})$, %
  and $\bigsqcup_{i \in \N} \psi_{\r, b}^i \sqsubseteq \bsem{\fix{\r}}b$.
  Then we can show that for all $b \in \base$ and $i \in \N$ it holds $\bsem{\r^i}b \sqsubseteq \psi_{\r, b}^{i+1}$, by induction on $i$: 
  \begin{itemize}
    \item $i = 0$: $\bsem{\r^0}b = b = \psi_{\r, b}(\psi_{\r, b}^{0}) = b \sqcup \bot$, hence it holds by reflexivity. 
    \item $i \rightarrow i+1$: 
    \begin{align*}
      \bsem{\r^{i+1}}b & = \add{\bsem{\r}}(\bsem{\r^{i}}b) \\
      & \sqsubseteq \add{\bsem{\r}}(\psi_{\r, b}^{i+1}) & [\mbox{inductive hypothesis and monotonicity of $\add{\bsem{\r}}$}] \\
      & \sqsubseteq b \sqcup \add{\bsem{\r}}(\psi_{\r, b}^{i+1}) \\
      & = \psi_{\r, b}(\psi_{\r, b}^{i+1}) \\
      & = \psi_{\r, b}^{i+2}
    \end{align*}
  \end{itemize}
  hence, by monotonicity of $\sqcup$, $\bsem{\r^*}b = \bigsqcup_{i \in \N} \bsem{\r^i}b \sqsubseteq \bigsqcup_{i \in \N} \psi_{\r, b}^i \sqsubseteq \bsem{\fix{\r}}b$. 

\bigskip

\noindent    
    \ref{le:kleene-axioms:left-exp}
  Let $b \in \base$. Then
\begin{align*}
  \bsem{\r^*}b 
  & = \bigsqcup_{i \geq 0} \bsem{\r^i}b \\ 
  & = b \sqcup \bigsqcup_{i \geq 1} \bsem{\r^i}b \\ 
  & =  b \sqcup \bigsqcup_{i \geq 0} \bsem{\r^i; \r}b \\
  & = b \sqcup \bigsqcup_{i \geq 0} \add{\bsem{\r}}(\bsem{\r^i}b)  \\
  & \sqsubseteq b \sqcup \add{\bsem{\r}} \bigsqcup_{i \geq 0} (\bsem{\r^i}b)  \\
& =  \bsem{\1}b \sqcup \add{\bsem{\r}}(\bsem{\r^*}b) \\ 
  & = \bsem{\1 + (\r^*; \r)}b 
  \end{align*} 

  \bigskip

\noindent
  \ref{le:kleene-axioms:left-dist}
    Let $b \in \base$. Then
  \begin{align*}
    \bsem{\r; (\r_1 + \r_2)}b
    & =  \add{\bsem{\r_1 + \r_2}}(\bsem{\r}b)\\          
    & = \bigsqcup_{b' \in \ebase{\bsem{\r}b}} \bsem{\r_1 + \r_2}b'\\
    & = \bigsqcup_{b' \in \ebase{\bsem{\r}b}} (\bsem{\r_1}b' \sqcup \bsem{\r_2}b') \\
    & = \Bigl (\bigsqcup_{b' \in \ebase{\bsem{\r}b}} \bsem{\r_1}b' \Bigr ) \sqcup
      \Bigl ( \bigsqcup_{b' \in \ebase{\bsem{\r}b}} \bsem{\r_2}b' \Bigr )
      & %
    \mbox{[associativity and commutativity of $\sqcup$]}\\
    & =  \add{\bsem{\r_1}}(\bsem{\r}b) \sqcup \add{\bsem{\r_2}}(\bsem{\r}b)\\          
    & =  \bsem{\r;\r_1}b \sqcup \bsem{\r; \r_2}b\\ 
    & =  \bsem{(\r;\r_1) + (\r; \r_2)}b 
  \end{align*}

      \bigskip
  
\noindent
  \ref{le:kleene-axioms:right-dist}
   Let $b \in \base$. Then
   \begin{align*}
    \bsem{\r_1; \r + \r_2; \r}b & = \bsem{\r_1; \r}b \sqcup \bsem{\r_2; \r}b \\ 
    & = \add{\bsem{\r}}(\bsem{\r_1}b) \sqcup \add{\bsem{\r}}(\bsem{\r_2}b) \\ 
    & \sqsubseteq \add{\bsem{\r}}((\bsem{\r_1}b) \sqcup \bsem{\r_2}b ) \\ 
    & = \add{\bsem{\r}}(\bsem{\r_1 + \r_2}b) \\ 
    & = \bsem{(\r_1 + \r_2); \r}b    
    \tag*{\qedhere}
   \end{align*} 
\end{proof}

Inequality~\ref{le:kleene-axioms:seq} states that sequential
composition is \emph{over-approximated} by its inductive counterpart
and this inequality can be strict (cf. Example~\ref{ex:signs-hole});
the same holds for the fixpoint,
by~\ref{le:kleene-axioms:fix-ind}. This is also what prevents
sequential composition from being associative and right-distributing over choice, so
inequalities \ref{le:kleene-axioms:seq-assoc} and~\ref{le:kleene-axioms:right-dist} may be strict too
(Examples~\ref{ex:notassociative} and~\ref{ex:intervals-kleene-fail}). Inequality~\ref{le:kleene-axioms:fix}
states that the Kleene star is over-approximated by the fixpoint,
possibly in a strict sense
(Example~\ref{ex:fix-and-star}). Finally,
  by~\ref{le:kleene-axioms:left-exp}, unfolding $\r^*$ by appending
  one execution of the body, i.e., as $\1 + (\r^*;\r)$,
  over-approximates $\r^*$, possibly strictly
  (Example~\ref{ex:incompleteness}).

\begin{example}[Failure of associativity for sequential composition]\label{ex:notassociative}
On the interval interpretation domain $(\Int, \sqsubseteq, \Ir{\Int})$, consider $\r_1 \triangleq \mathsf{inc}\,x + \mathsf{dec}\,x$,
$\r_2 \triangleq  x{\neq}0?$ and $\r_3 \triangleq  x{=}0?$, where $\mathsf{inc}\,x$ and
$\mathsf{dec}\,x$ abbreviate $x:=x+1$ and $x:=x-1$.
On input $[0,0]$, the two branches of $\r_1$ give $[1,1]$ and $[-1,-1]$, so that
$\bsem{\r_1}[0,0] = [-1,1]$. Then, $\r_2$ keeps $[-1,-1]$ and $[1,1]$, whose join $[-1,1]$
spuriously contains $0$ again. Hence:
\begin{align*}
  \bsem{(\r_1;\r_2);\r_3}[0,0]
  & = \add{\bsem{\r_3}}\bigl(\add{\bsem{\r_2}}(\bsem{\r_1}[0,0])\bigr)
    = \add{\bsem{\r_3}}\bigl([-1,-1] \sqcup \varnothing \sqcup [1,1]\bigr)
    = \add{\bsem{\r_3}}[-1,1] = [0,0],\\
  \bsem{\r_1;(\r_2;\r_3)}[0,0]
  & = \add{\bsem{\r_2;\r_3}}[-1,1]
    = \textstyle\bigsqcup_{n \in \{-1,0,1\}} \bsem{\r_2;\r_3}[n,n] = \varnothing ,
\end{align*}
since each singleton fails one of the two guards.
\hfill $\lozenge$
\end{example}

\begin{example}[Failure of right distributivity]\label{ex:intervals-kleene-fail}
  On $(\Int, \sqsubseteq, \Ir{\Int})$, let us consider the program
  $\mathsf{q} \triangleq (\mathsf{inc}\,x +
  \mathsf{dec}\,x);\,x{=}0?$.  Right distributivity of sequential
  composition over choice fails.  As in
    Example~\ref{ex:notassociative}, on input $[0,0]$ the two branches
    give $[1,1]$ and $[-1,-1]$, whose join $[-1,1]$ spuriously
    contains $0$, so that
  $\sem{\mathsf{q}}[0,0] = \add{\bsem{x{=}0?}}[-1,1] = [0,0]$.  Instead, when each branch is tested separately,
  $\bsem{(\mathsf{inc}\, x); x{=}0?}[0,0] \sqcup  \bsem{(\mathsf{dec}\, x); x{=}0?}[0,0] = \add{\bsem{x{=}0?}}[1,1] \sqcup \add{\bsem{x{=}0?}}[-1,-1] =
  \varnothing$.  \hfill $\lozenge$
\end{example}

\begin{example}[Fixpoint vs.\ Kleene star]\label{ex:fix-and-star}
  On $(\Int, \sqsubseteq, \Ir{\Int})$, consider the program
  $$\ti \triangleq ((x\,\mathsf{mod}\,2 = 1?);\,x:=x-2) + (x:=x+2)\ .$$
  Then $\sem{\ti^*}[0,0] = \bigsqcup_{i \in \N} \sem{\ti^i}[0,0] = \bigsqcup_{i \in \N}[2i,2i] =
  [0,+\infty]$, whereas $\sem{\fix{\ti}}[0,0] = [-\infty,+\infty]$. Indeed,
  $\sem{\fix{\ti}}[0,0] = \lfp(\lambda X.\,[0,0]\sqcup\add{\bsem{\ti}}X)$, whose Kleene
  iterates are $\varnothing,\,[0,0],\,[0,2],\,[-1,4],\dots$. Thus, once $1 \in [0,2]$, the guard
  $(x\,\mathsf{mod}\,2 = 1?)$ holds for $1$, so that $x:=x-2$ produces $-1$, and Kleene
  iteration then extends the interval indefinitely in both directions. \hfill $\lozenge$
\end{example}
The inequalities of Lemma~\ref{le:kleene-axioms} might not be equalities when the
semantics is not additive, as is typical in abstract interpretation. 
When the 
semantics $\sem[\base]{\cdot}$ is additive, they become equalities, and, for the
irreducible basis, $\sem[\base]{\r}$ coincides with $\sem[D]{\r}$, i.e.\ it admits an
inductive definition (Proposition~\ref{pr:inductive}; cf.~\S~\ref{inst:simple-inductive}).

\begin{propositionrep}[Full is inductive]\label{pr:inductive}
  Let $(D, \sqsubseteq, B)$ be an interpretation domain with irreducible basis $B=\Ir{D}$
  whose full semantics $\sem[B]{\cdot}: D \to D$ is additive. Then,
  $\sem[B]{\r} = \sem[D]{\r}$ for every command $\r$, and, moreover,
  $\sem[B]{\r^*} = \sem[B]{\fix{\r}}$.
\end{propositionrep}

\begin{proof}
  The proof shows that $\sem[\base]{\cdot}$ can be defined
  inductively, i.e., also the cases~\ref{le:kleene-axioms:seq} and~\ref{le:kleene-axioms:fix-ind}
  in Lemma~\ref{le:kleene-axioms} become
  equalities. Hence we conclude that
  $\sem[\base]{\cdot} = \sem[D]{\cdot}$.

  For sequential composition, in this setting let $d \in D$, then
  \begin{align*}
    \sem[\base]{\r_1; \r_2}d %
    & = \bigsqcup_{b \in \ebase{d}} \bsem[\base]{\r_1; \r_2}b \\ 
    & =  \bigsqcup_{b \in \ebase{d}} \add{\bsem[\base]{\r_2}} (\bsem[\base]{\r_1}b) \\
    & =  \bigsqcup_{b \in \ebase{d}} \sem[\base]{\r_2} (\bsem[\base]{\r_1}b) \\ 
    & =  \sem[\base]{\r_2} \Bigl(\bigsqcup_{b \in \ebase{d}}  \bsem[\base]    {\r_1}b\Bigr) %
    & %
    [\mbox{additivity of $\sem[\base]{\r_2}$}] \\ 
    & = \sem[\base]{\r_2}(\sem[\base]{\r_1}d) %
  \end{align*}

  For the fixpoint operator, we need to show that under the additivity hypothesis, for $d \in D$, the missing inequality:
\[
    \lfp_{\sqsubseteq} (\lambda x . d \sqcup \sem{\r}x) \sqsubseteq \sem[\base]{\fix{\r}}d.
\]
  Define the function $\psi_{\r, d} = \lambda x. d \sqcup \sem{\r}x$.
  Then
  \begin{align*}
    \psi_{\r, d}(\sem[\base]{\fix{\r}}d)
    & = \psi_{\r, d}\Bigr (\bigsqcup_{b \in \ebase{d}} \sem[\base]{\fix{\r}}b \Bigr)\\
    & = d \sqcup \sem[\base]{\r} \Bigl (\bigsqcup_{b \in \ebase{d}} \sem[\base]{\fix{\r}}b \Bigr ) \\
    & = \Bigl (\bigsqcup_{b \in \ebase{d}} b \Bigr) \sqcup \Bigl (\bigsqcup_{b \in \ebase{d}} \sem[\base]{\r} (\sem[\base]{\fix{\r}}b) \Bigr)
     & \mbox{[additivity of $\sem[\base]{\r}$]} \\
    & = \bigsqcup_{b \in \ebase{d}} \Bigl (b \sqcup  \sem[\base]{\r} (\sem[\base]{\fix{\r}}b) \Bigr )\\
    & = \bigsqcup_{b \in \ebase{d}} \psi_{\r, b}(\sem[\base]{\fix{\r}}b) \\
    & = \bigsqcup_{b \in \ebase{d}} \sem[\base]{\fix{\r}} 
    & \mbox{[$\sem[\base]{\fix{\r}}b$ is a fixpoint]} \\ 
    & = \sem[\base]{\fix{\r}}d
  \end{align*}
  Therefore, since $\sem[\base]{\fix{\r}}d$ is post-fixpoint of $\psi_{\r, d}$, it is larger than its least fixpoint, as desired.

  \medskip
  
  We are left to show that
  $\sem[\base]{\r^*} = \sem[\base]{\fix{\r}}$. We know, by
  Lemma~\ref{le:kleene-axioms:fix} that
  $\sem[\base]{\fix{\r}}d \sqsupseteq \sem[\base]{\r^*}d$. We need to
  prove the converse inequality. We obtain it by showing that
  $\psi_{\r,d}(\sem[\base]{\r^*}d) = \sem[\base]{\r^*}d$ and thus
  $\sem[\base]{\fix{\r}}d = \lfp_\sqsubseteq(\psi_{\r,d}) \sqsubseteq
  \sem[\base]{\r^*}d$, where the first equality has been just
  shown. This follows easily from additivity

  \begin{align*}
    \psi_{\r, d}(\sem[\base]{\r^*}d)
    & = \psi_{\r, d} \Bigl (\bigsqcup_{i \geq 0} \sem[\base]{\r^i}d \Bigr )\\
    & = d \sqcup \sem[\base]{\r} \Bigl (\bigsqcup_{i \geq 0} \sem[\base]{\r^i}d \Bigr ) \\
    & = d \sqcup \bigsqcup_{i \geq 0} \sem[\base]{\r} (\sem[\base]{\r^i}d)
    & \mbox{[additivity of $\sem[\base]{\r}$]} \\
    & = d \sqcup \bigsqcup_{i \geq 0} \sem[\base]{\r^{i+1}}d
    & \mbox{[seq. comp. is inductive, as shown above]} \\    
    & = \bigsqcup_{i \geq 0} \sem[\base]{\r^i}d\\
    & = \sem[\base]{\r^*}d
    \tag*{\qedhere}
  \end{align*}
\end{proof}

In particular, when the semantics on the irreducible basis is additive, the hierarchy induced by different bases collapses. For completely meet-distributive lattices (where binary meets distribute over arbitrary joins) with an irreducible basis, every subset is dense, and additivity on elementary commands extends to the entire semantics. This holds, e.g., for the powerset and the $\Sign$  domain.

\begin{lemmarep}[Condition for additive semantics]\label{le:basic-prop-sem:distrib}
  If $D$ is a completely meet-distributive lattice with basis $\Ir{D}$, then
  $\sem{\r}: D \rightarrow D$ is additive for every $\r\in \Reg$.
\end{lemmarep}

\begin{proof}
  Observe that, in this setup, all $X \subseteq D$ are dense.
  In fact, for all $b\in \Ir{D}$, if $b \sqsubseteq \sqcup X$,
  then
  \( b = b \sqcap \sqcup X = \sqcup \set{b \sqcap x \mid x \in X } \),
  so that there exists $x \in X$ such that $b = b \sqcap x$, and, in
  turn, $b \sqsubseteq x$.  Hence, by Lemma~\ref{le:add-dense}, the
  full semantics $\sem{\r}$ is additive on $X$.
  \hfill
\end{proof}

\subsection{Density} 
\label{ss:density}
Unlike the semantics on the powerset $\pow{\Sigma}$, the semantics on a generic domain
is, in general, not additive, with the consequences discussed above for its properties and,
later, for the design of the logic. The choice of the join basis lets us characterise the fragments on which
the full semantics \emph{is} additive, through density, which will govern the
applicability of the crucial (join) rule.
\begin{definition}[Dense subset]
  \label{def:dense}
  Let $\lat$ be a complete lattice with a basis $\base$.  A
  subset $X \subseteq \lat$ is \emph{dense} (w.r.t.\  $\base$) if
  for every $b \in \base$ such that
  $b \sqsubseteq \sqcup X$, there exists $x \in X$ with
  $b \sqsubseteq x$. A dense $X$ is also called a \emph{dense decomposition} of $\sqcup X$.
  \hfill $\lozenge$
\end{definition}
Intuitively, a dense subset covers the basis without ``holes''. In $\Int$ with basis
$\Ir{\Int}$, for instance, $X = \set{[-1,0],[2,3]}$ is not dense, as
$[1,1] \sqsubseteq \sqcup X = [-1,3]$ but $[1,1]$ is covered by no element of $X$, whereas
the set $\set{[-1,1],[2,3]}$ is. This is made precise for complete lattices
$D \subseteq \pow{Y}$ with a basis of singletons: in this case, $X$ is dense iff $\sqcup X = \cup X$
(Lemma~\ref{lemma:density}). This covers most abstract domains in the paper, notably 
intervals and octagons (and more: congruences, convex polyhedra, etc.). 
A simple but fundamental observation is that the full semantics $\sem{\r}$
is additive on dense subsets. %

\begin{toappendix}
\begin{lemmarep}[Density in lattices of sets]\label{lemma:density}
  Let $D \subseteq \pow{Y}$ for some set $Y$. Assume that
  $(D, \subseteq)$
  is a complete lattice and let  
  $\base \subseteq \set{\set{y} \mid y \in Y}$ be basis for $D$. Then $X \subseteq D$ is dense w.r.t. $\base$ if and only if
  $\sqcup X = \cup X$.
\end{lemmarep}

\begin{proof}
We have $\sqsubseteq{}={}\subseteq$, and $\bigsqcup X=\bigcap\{d\in D\mid \bigcup X\subseteq d\}$ is the
least element of $D$ containing every $x\in X$; in particular $x\subseteq\bigsqcup X$ for each $x\in X$, so
\( \bigcup X\;\subseteq\;\bigsqcup X\) always holds.

\smallskip
\noindent$(\Rightarrow)$ Assume $X$ is dense. Let $y\in\bigsqcup X$. Then $\{y\}\sqsubseteq\bigsqcup X$ and
$\{y\}\in B$, so density gives the existence of $x\in X$ with $\{y\}\sqsubseteq x$, i.e.\ $y\in x\subseteq\bigcup X$.
Hence $\bigsqcup X\subseteq\bigcup X$, which entails $\bigsqcup X=\bigcup X$.

\smallskip
\noindent$(\Leftarrow)$ Assume $\bigsqcup X=\bigcup X$. Let $b=\{y\}\in \base$ with $b\sqsubseteq\bigsqcup X$.
Then $y\in\bigsqcup X=\bigcup X$, so $y\in x$ for some $x\in X$, i.e.\ $b\sqsubseteq x$. Hence $X$ is dense.
 \hfill
\end{proof}
\end{toappendix}

\begin{lemmarep}[Additivity on dense subsets]
  \label{le:add-dense}
  If $X \subseteq D$ is dense, then $\sem{\r}(\sqcup X) = \bigsqcup_{d \in X} \sem{\r}d$.
\end{lemmarep}

\begin{proof}
  Just observe that, as an immediate consequence of density of $X$, we
  have $\ebase{\sqcup X} = \bigcup_{d \in X} \ebase{d}$ and thus
  \[ \sem{\r} (\sqcup X)
     = \bigsqcup_{b \in \ebase{\sqcup X}} \bsem{\r}b
     = \bigsqcup_{b \in \bigcup_{d \in X} \ebase{d}} \bsem{\r}b
     = \bigsqcup_{d \in X} \bigsqcup_{b \in \ebase{d}} \bsem{\r}b
     = \bigsqcup_{d \in X} \sem{\r}d \tag*{\qedhere}
    \]
\end{proof}

Deciding whether an arbitrary set is dense can be costly. For
instance, consider the interval domain $\Int$. For a
program over $n$ variables, its elements are products of
intervals, called \emph{boxes}.
By Lemma~\ref{lemma:density}, checking the density of $k$ boxes amounts to the
\emph{coverage problem}: does the union of the boxes cover their
hull? This problem is closely related to Klee's measure problem: 
it is coNP-complete when $n$ is part of the input~\cite{HoffmannK05}, and 
solvable in $O(k^{n/2})$ time for
fixed $n \ge 3$~\cite{Chan13}.

However, as illustrated by the example in
\S~\ref{subsec:example}, dense
decompositions are \emph{generated}, rather than tested, during the construction of a derivation.
We next show how to generate them by operations that preserve
density by construction.

\begin{definition}[Splitting set]
  \label{de:splitting}
  Given $h \in D$, a subset $K \subseteq D$
  \emph{splits} $h$ if $\ebase{h} \subseteq \down{K}$. If $K$ splits
  every $h \in D$, i.e.\
  $\base \subseteq \down{K}$, then $K$ is a \emph{universal} splitting set.
  \hfill$\lozenge$
\end{definition}

In words, $K$ splits $h$ when the elements of $K$ cover every basis
element below $h$, so that $K$ can be used to decompose $h$ into
smaller pieces. It turns out that meeting $h$ with a splitting set
provides a dense decomposition of $h$, and refining each element of a
dense decomposition by a dense decomposition yields again a dense
decomposition.

\begin{propositionrep}[Generating dense decompositions]
  \label{pr:dense-gen}
  Let $(D, \sqsubseteq, \base)$ be an interpretation domain and $h \in D$.
  \begin{enumerate}[{\rm (1)}]
  \item
    \label{pr:dense-gen:split}
    \emph{(Splitting)} If $K \subseteq D$ splits $h$, then
    $\{h \sqcap k\}_{k \in K}$ is a dense decomposition of $h$.
    
  \item
    \label{pr:dense-gen:refine}
    \emph{(Refinement)} If $\{h_i\}_{i \in I}$ is a dense
    decomposition of $h$ and, for every $i \in I$,
    $\{h_{i,j}\}_{j \in J_i}$ is a dense decomposition of $h_i$, then
    $\{h_{i,j}\}_{i \in I,\, j \in J_i}$ is a dense decomposition of
    $h$.
  \end{enumerate}
\end{propositionrep}

\begin{proof}
  \ref{pr:dense-gen:split}
  Clearly
  $\bigsqcup_{k \in K} h \sqcap k \sqsubseteq h$. Let $b \in
  \ebase{h}$. Since $K$ splits $h$, there is $k \in K$ with
  $b \sqsubseteq k$, hence $b \sqsubseteq h \sqcap k$. Therefore,
    $h = \bigsqcup \ebase{h} \sqsubseteq \bigsqcup_{k \in K} h \sqcap k$, 
    so that
  $\bigsqcup_{k \in K} (h \sqcap k) = h$ holds. The same arguments shows that 
  $\{h \sqcap k\}_{k \in K}$ is a dense set.
   
  \smallskip
  \noindent
  \ref{pr:dense-gen:refine}
  We have $\bigsqcup_{i \in I} \bigsqcup_{j \in J_i} h_{i,j} = \bigsqcup_{i \in I} h_i = h$.
  Let $b \in \base$ with $b \sqsubseteq h$. By density of $\{h_i\}_{i \in I}$, there is
  $i \in I$ with $b \sqsubseteq h_i = \bigsqcup_{j \in J_i} h_{i,j}$, and by density of
  $\{h_{i,j}\}_{j \in J_i}$, there is $j \in J_i$ with $b \sqsubseteq h_{i,j}$, thus showing
  that  $\{h_{i,j}\}_{i \in I,\, j \in J_i}$ is a dense decomposition of
    $h$.
  \hfill
\end{proof}

\emph{Template domains}~\cite{SankaranarayananSM05}, of which boxes
and octagons are instances, naturally provide universal
splitting sets. Let $\Var$ be finite and $\Lambda$ a finite set of
linear forms $\ell$ on $\Var$ to be instantiated to states $\sigma\in \Sigma$, i.e., 
$\ell(\sigma) = \sum_{v \in \Var} a_v\, \sigma(v)$ with
$a_v \in \Z$, closed under negation and including $\pm v$ for all
$v \in \Var$. The template domain $\Tmpl{\Lambda}$ consists of the
sets
$\{\sigma \in \Sigma \mid \forall \ell \in \Lambda.\ \ell(\sigma) \le
c_\ell\}$ with $c_\ell \in \Z \cup \{+\infty\}$, ordered by
inclusion. For $\ell \in \Lambda$ and
$c \in \Z$, we write
$\{\ell \le c\} \triangleq \{\sigma \in \Sigma \mid \ell(\sigma) \le c\}$,
i.e., the element of $\Tmpl{\Lambda}$ with $c_\ell = c$ and $c_{\ell'} = +\infty$ for all
$\ell' \neq \ell$, and $\{\ell \ge c\} \triangleq \{-\ell \le -c\}$, which belongs to
$\Tmpl{\Lambda}$ as $\Lambda$ is closed under negation.
$(\Tmpl{\Lambda}, \subseteq, \Ir{\Tmpl{\Lambda}})$ is an interpretation domain, with
$\Ir{\Tmpl{\Lambda}} = \{\{\sigma\} \mid \sigma \in \Sigma\}$ and $\sqcap = \cap$ (Proposition~\ref{pr:template-split}).
Boxes, i.e.\ $\Int$ over several variables, and octagons
$\Oct$~\cite{Mine06} are $\Tmpl{\Lambda}$ for $\Lambda = \{\pm v\}_{v \in \Var}$ and
$\Lambda = \{\pm v\}_{v \in \Var} \cup \{\pm v \pm w\}_{v \neq w \in \Var}$, respectively.

\begin{propositionrep}[Splitting in template domains]
  \label{pr:template-split}
  Let $\Lambda$ be as above. For all $\ell \in \Lambda$ and $c, c' \in \Z$ with
  $c' \le c + 1$, the pair $\{\{\ell \le c\}, \{\ell \ge c'\}\}$ is a universal splitting
  set of\/ $\Tmpl{\Lambda}$. Hence, for every $h \in \Tmpl{\Lambda}$, the pair
  $\{h \cap \{\ell \le c\},\ h \cap \{\ell \ge c'\}\}$ is a dense decomposition of $h$.
  \end{propositionrep}

\begin{proof}
  First recall that $\Tmpl{\Lambda}$ is a Moore family and thus
  $\sqcap = \cap$. In fact let
  $\{h_j\}_{j \in J} \subseteq \Tmpl{\Lambda}$, with
  $h_j = \{\sigma \mid \forall \ell \in \Lambda.\ \ell(\sigma) \le
  c^j_\ell\}$. Then
  $h = \bigcap_{j \in J} h_j = \{\sigma \mid \forall \ell \in
  \Lambda.\ \ell(\sigma) \le \inf_{j \in J} c^j_\ell\}$. If
  $\inf_{j} c^j_\ell \in \Z \cup \{+\infty\}$ for every
  $\ell \in \Lambda$, then clearly $h \in \Tmpl{\Lambda}$. Otherwise,
  if for some $\ell$ it holds $\inf_{j} c^j_\ell = - \infty$, then
  $h=\varnothing$, which also belongs to $\Tmpl{\Lambda}$.

  Then it suffices to note that for every $\sigma \in \Sigma$,
  $\ell(\sigma) \in \Z$, hence either $\ell(\sigma) \le c$ or
  $\ell(\sigma) \ge c + 1 \ge c'$: every basis element $\{\sigma\}$
  lies below one of the two elements, so the pair is a universal
  splitting family. \hfill
\end{proof}

In particular, splitting a box along the half-spaces $v \le c$ and $v \ge c$ of a variable
$v$, or an octagon along those of an octagonal constraint, always yields a dense
decomposition. Hence, the splits at $x = 0$ in \S~\ref{subsec:example} and
\S~\ref{sub:octagons} are of this form.

\subsection{Abstract interpretation for interpretation domains}
\label{ss:ai-interpretation}

Interpretation domains fit naturally within abstract interpretation. We already argued that smaller bases yield finer semantics for the same lattice. 
Since basis and density are preserved through abstraction, given an interpretation domain $(C, \sqsubseteq, \base)$ and a Galois insertion $\langle \alpha,\gamma \rangle: C \rightleftarrows A$ we can construct an
abstract domain whose semantics
soundly over-approximates the concrete one.

\begin{lemmarep}[Basis and Density Preservation]\label{le:base-transf}
  Let $(C, \sqsubseteq, \base)$ be an interpretation domain and
 $\langle \alpha,\gamma \rangle: C \rightleftarrows A$ a Galois insertion.
  \begin{enumerate}[\rm (P$_1$)]
    \item \label{le:base-transf:base} \emph{(Basis)}
      The set $\alpha(\base) = \{ \alpha(p) \mid p \in \base \}$ is a basis for $A$.
    \item \label{le:base-transf:density} \emph{(Density)}
      If $X \subseteq A$ is dense in $A$ w.r.t.\ $\alpha(\base)$, then
      $\gamma(X) = \{\gamma(a) \mid a \in X\}$ is dense in $C$ w.r.t.\ $\base$.
  \end{enumerate}
\end{lemmarep}

\begin{proof}
  Let us denote by $\base_A = \alpha(\base)$.
  
  \ref{le:base-transf:base}
  Immediate consequence of the fact that $\alpha$ is additive and
  surjective. In fact, for all $a \in A$, by surjectivity of $\alpha$
  there is $c \in C$ such that $a = \alpha(c)$. Now observe
  $\alpha(\ebase{c}) \subseteq \base_A(a)$: in fact, if $p \in \ebase{c}$,
  by monotonicity, $\alpha(p) \sqsubseteq_A \alpha(c)$ and thus
  $\alpha(p) \in \base_A(a)$.
  
  Thus
  \begin{align*}
    a & = \alpha(c)\\
      & = \alpha\left (\sqcup\ebase{c}\right ) & \mbox{[$\base$ is a basis for $C$]}\\
      & = \sqcup\alpha(\ebase{c}) & \mbox{[$\alpha$ additive]}\\
      & \sqsubseteq_A {\sqcup}_A \base_A(a) & \mbox{[observation above]}\\
      & = a                    
  \end{align*}
  and thus we conclude $a =\bigsqcup_A \base_A(a)$.

  \bigskip
  \noindent
  \ref{le:base-transf:density} 
  Assume $X \subseteq A$ dense in $A$
  (w.r.t.  $\base_A$). Let $p \in \base$. If
  $p \sqsubseteq \sqcup\gamma(X) \sqsubseteq \gamma(\bigsqcup_A
  X)$ then, by adjointness, $\alpha(p) \sqsubseteq_A \bigsqcup_A
  X$. Hence by density of $X$, there is $a \in X$ such that
  $\alpha(p) \sqsubseteq_A a$ and thus, again by adjointness,
  $p \sqsubseteq \gamma(a)$, as desired.
  \hfill
\end{proof}

By the result above, $(A, \sqsubseteq_A, \alpha(\base))$ is an interpretation domain, and its semantics
soundly approximates the concrete one.

\begin{propositionrep}[Abstract interpretation domain, soundness]\label{pr:ai-sound}  
  Let $(C, \sqsubseteq, \base)$ be an interpretation domain and
  $\langle \alpha,\gamma \rangle: C \rightleftarrows A$ a Galois insertion. If the abstract semantics of
  elementary commands is sound, i.e.,
  $\alpha \circ \bsem{\e}_{\base}^C \sqsubseteq_A \bsem{\e}_{\alpha(\base)}^A \circ \alpha$,
  then $\alpha \circ \sem{\r}_{\base}^C \sqsubseteq_A \sem{\r}_{\alpha(\base)}^A \circ \alpha$
  for all commands $\r\in \Reg$.
\end{propositionrep}

\begin{proof}
  For the sake of readability, we omit the indication of the bases
  from the semantics and let us denote by $\base^A =\alpha(\base)$.

  We first show that the inequality holds for the
  basis semantics, i.e., for all $b \in B$
  \begin{center}
    $\alpha(\bsem{\r}^C b) \sqsubseteq_A \bsem{\r}^A \alpha(b)$
  \end{center}

  The proof proceeds by a routine induction, with the base case
  $\r=\e$ which holds by hypothesis. 
  
  For $\r = \r_1 + \r_2$ the results follows from additivity of $\alpha$: 
  by inductive hypothesis $\alpha(\bsem{\r_i}^C b) \sqsubseteq_A \bsem{\r_i}^A \alpha(b)$ 
  for $i = 1, 2$, then 
  \begin{align*}
     \alpha( \bsem{\r_1 + \r_2}^C b) & = \alpha ( \bsem{\r_1}^C b \sqcup \bsem{\r_2}^C b ) \\ 
     & = \alpha ( \bsem{\r_1}^C b ) \sqcup \alpha ( \bsem{\r_2}^C b ) 
     & \mbox{[additivity of $\alpha$]} \\ 
     & \sqsubseteq   \bsem{\r_1}^A \alpha(b) \sqcup \bsem{\r_2}^A \alpha(b) 
     & \mbox{[inductive hypothesis and monotonicity of $\sqcup$]} \\ 
     & = \bsem{\r_1 + \r_2}^A \alpha(b)
  \end{align*}

  The delicate case is
  sequential composition $\r = \r_1 ; \r_2$. By inductive hypothesis $\alpha(\bsem{\r_i}^C b) \sqsubseteq_A \bsem{\r_i}^A \alpha(b)$ 
  for $i = 1, 2$. We have that
  \begin{align*}
    \alpha( \bsem{\r_1 ; \r_2}^C b)
    & = \alpha \Bigl ( \bigsqcup_{b' \in \base(\bsem{\r_1}^C b)} \bsem{\r_2}^C b' \Bigr )\\
    &= \bigsqcup_{b' \in \base(\bsem{\r_1}^C b)} \alpha(\bsem{\r_2}^C b')
    & \mbox{[additivity of $\alpha$]}\\
    & \sqsubseteq_A \bigsqcup_{b' \in \base(\bsem{\r_1}^C b)} \bsem{\r_2}^A \alpha(b')
    & \mbox{[by ind. hyp. $\alpha(\bsem{\r_2}^C b') \sqsubseteq_A \bsem{\r_2}^A \alpha(b')$]}\\
    & \sqsubseteq_A \bigsqcup_{\alpha(b') \in \base^A(\bsem{\r_1}^A \alpha(b))} \bsem{\r_2}^A \alpha(b')
        & \mbox{[monotonicity of $\alpha$ and ind. hyp.]}\\
    &  = \bsem{\r_1 ; \r_2}^A \alpha(b)
  \end{align*}

  For $\r = \r_1^*$, by inductive hypothesis $\alpha(\bsem{\r_1}^C b) \sqsubseteq_A \bsem{\r_1}^A \alpha(b)$. 
  By induction then for all $i \in \N$ it holds $\alpha(\bsem{\r_1^i}^C b) \sqsubseteq_A \bsem{\r_1^i}^A \alpha(b)$: \begin{itemize}
    \item $i = 0$: $\alpha(\bsem{\r_1^0}^C b) = \alpha(b) = \bsem{\r_1^0}^A \alpha(b)$ 
    \item $i \rightarrow i+1$: $\alpha(\bsem{\r_1^{i+1}}^C b) = \alpha (\bsem{\r_1^{i+1}}^C b)$ is analogous to the case of sequential composition: 
   \begin{align*}
    \alpha( \bsem{\r_1^{i+1}}^C b) & = \alpha(\add{\bsem{\r_1}^C}(\bsem{\r_1^i}^C b)) \\
    & = \alpha \Bigl ( \bigsqcup_{b' \in \base(\bsem{\r_1^i}^C b)} \bsem{\r_1}^C b' \Bigr )\\
    &= \bigsqcup_{b' \in \base(\bsem{\r_1^i}^C b)} \alpha(\bsem{\r_1}^C b')
    & \mbox{[additivity of $\alpha$]}\\
    & \sqsubseteq_A \bigsqcup_{b' \in \base(\bsem{\r_1}^C b)} \bsem{\r_1}^A \alpha(b')
    & \mbox{[by ind. hyp. $\alpha(\bsem{\r_1}^C b') \sqsubseteq_A \bsem{\r_1}^A \alpha(b')$]}\\
    & \sqsubseteq_A \bigsqcup_{\alpha(b') \in \base^A(\bsem{\r_1^i}^A \alpha(b))} \bsem{\r_2}^A \alpha(b')
        & \mbox{[monot. of $\alpha$ and ind. hyp. $\alpha(\bsem{\r_1^i}^C b') \sqsubseteq_A \bsem{\r_1^i}^A \alpha(b')$]}\\
        & = \add{\bsem{\r_1}^A} \bsem{\r_1^i}\alpha(b)\\ 
    &  = \bsem{\r_1^{i+1} }^A \alpha(b)
  \end{align*}

  \end{itemize}

  For $\r = \fix{\r_1}$ define the functions
  $\psi_{\r_1, b} = \lambda x. b \sqcup \add{\bsem{\r_1}^C}x$ and
  $\psi_{\r_1, \alpha(b)} = \lambda x. b \sqcup \add{\bsem{\r_1}^A}x$.

  \noindent Let $a = \sem{\fix{\r_1}}^A \alpha(b) = \lfp(\psi_{\r_1, \alpha(b)})$.
  Observe that
  \begin{align*}
    \alpha(\psi_{\r_1, b}(\gamma(a)))
    & = \alpha(b \sqcup \add{\bsem{\r_1}^C}\gamma(a))
    \\
    & = \alpha(b) \sqcup \alpha\Bigl (\bigsqcup_{b \sqsubseteq \gamma(a)} \bsem{\r_1}^C b \Bigr )
    & \mbox{[$\alpha$ additive]}\\
    & = \alpha(b) \sqcup \bigsqcup_{\alpha(b) \sqsubseteq a} \alpha(\bsem{\r_1}^C b)
    & \mbox{[$\alpha$ additive and adjointness]}\\
    &  \sqsubseteq \alpha(b) \sqcup \bigsqcup_{\alpha(b) \sqsubseteq a} \bsem{\r_1}^A \alpha(b)
    & \mbox{[inductive hypothesis]}\\
    &  = \alpha(b) \sqcup \add{\bsem{\r_1}^A} a
    & \mbox{[inductive hypothesis]}\\
    &  = \psi_{\r_1, \alpha(b)}(a)\\
    & = a
  \end{align*}
  Therefore, by adjointness, $\psi_{\r_1, b}(\gamma(a)) \sqsubseteq \gamma(a)$, i.e., $\gamma(a)$ is a post-fixpont of $\psi_{\r_1, b}$ and thus
  \begin{center}
    $\sem{\fix{\r_1}}^C b = \lfp(\psi_{\r_1, b}) \sqsubseteq \gamma(a)$
  \end{center}
  
  \noindent Again, by adjointness, $\alpha(\sem{\fix{\r_1}}^C b) \sqsubseteq a = \sem{\fix{\r_1}}^A \alpha(b)$, as desired.
  
 Then, the result can be extended to the full semantics. In fact, for all commands $\r$ and $c \in C$
  \begin{align*}
    \alpha( \sem{\r}^C c)
    & = \alpha \Bigl ( \bigsqcup_{b \in \base(c)} \sem{\r}^C b \Bigr )\\ 
    &= \bigsqcup_{b \in \base(c)} \alpha(\sem{\r}^C b)
    & \mbox{[additivity of $\alpha$]}\\
    & \sqsubseteq_A \bigsqcup_{b \in \base(c)} \sem{\r}^A \alpha(b)
    & \mbox{[by ind. hyp. $\alpha(\sem{\r}^C b) \sqsubseteq_A \sem{\r}^A \alpha(b)$]}\\
    & = \sem{\r}^A \alpha(c)     
    & \mbox{[monotonicity of $\alpha$ and ind. hyp.]}\\
  \end{align*}
  where the second last passage is motivated by the observation that
  $\alpha(c) = \alpha(\bigsqcup_{b \in \base(c)} b) = \bigsqcup_{b \in
    \base(c)} \alpha(b)$, hence
  $\base^A(\alpha(c)) = \alpha(\base(c))$.
  \hfill 
\end{proof}

It turns out that the basis shapes the abstraction: different bases for the same concrete domain typically
induce different abstract semantics, even when they induce the same concrete one. The
trivial basis $\base = C$ gives the coarsest abstraction, namely the usual inductive
abstract semantics. The irreducible basis gives the finest, still in general coarser than
the BCA. For example, the two powerset interpretation domains
$(\pow{\Sigma}, \subseteq, \pow{\Sigma})$ and $(\pow{\Sigma}, \subseteq, \Ir{\pow{\Sigma}})$
have the same concrete semantics (\S~\ref{inst:irr-collecting},~\S~\ref{inst:simple-inductive}),
yet abstracting them along intervals $\Int$ 
(Example~\ref{ex:intervals}) yields the
trivial and irreducible basis domains over $\Int$, whose abstract semantics significantly differ.

\begin{example}[Semantics with irreducible basis coarser than BCA]\label{ex:interval-BCA}
  Consider the program $\mathsf{p} \triangleq (x{\neq}0?);\,(x{=}0?)$ from Example~\ref{ex:signs-hole} on
  $(\Int, \sqsubseteq, \Ir{\Int})$. On input $[-1,1]$ the semantics matches the exact interval
  abstraction, i.e., $\sem{\mathsf{p}}[-1,1] = \varnothing =
  \alpha_{\Int}(\sem{\mathsf{p}}\set{-1,0,1})$: the irreducible basis splits $[-1,1]$ into
  atoms $[-1,-1],[0,0],[1,1]$, avoiding the imprecision of testing $x{\neq}0?$ on the convex
  $[-1,1]$ in the inductive semantics on $(\Int, \sqsubseteq, \Int)$. In general, though, it
  can be coarser than the BCA, e.g.\ when a join merges states: for example, with
  $\mathsf{q} \triangleq (\mathsf{inc}\,x + \mathsf{dec}\,x);\,(x{=}0?)$
  of Example~\ref{ex:intervals-kleene-fail}, we have that $\sem{\mathsf{q}}\set{0} = \varnothing$
  concretely, whereas $\sem{\mathsf{q}}[0,0] = [0,0]$ under the irreducible basis. \hfill$\lozenge$
\end{example}

\section{The Program Logic \texorpdfstring{$\APPL$}{APPL}}\label{sec:logic}

We introduce $\APPL$, an Abstract Program Property Logic: a general
framework for reasoning about regular commands over a domain $D$, with
pre- and postconditions drawn from a language $\L$.
Density plays a central role. The disjunction rule of Hoare logic, read with joins
in $D$, is unsound in general. Since, by  Lemma~\ref{le:add-dense}, the semantics is additive on dense sets, 
guarding the rule with density yields the sound rule (join). The
choice of the basis governs the precision of the semantics
(Proposition~\ref{pr:hierarchy})  and the applicability of (join), and 
thus directly influences the deductive power of the system. 
The theory is
fully general, but three scenarios are of primary interest: (i) $D$ is
a semantic domain and $\L$ the fragment it denotes in a logical
language; (ii) $D$ is a concrete domain and $\L = A$ an abstraction of
$D$; and (iii) $D$ is itself an abstract domain and $\L = D$.

\subsection{Logics} 
\label{ss:logics} 

We assume the full semantics $\sem{\cdot}$ (Definition~\ref{de:semantics}) is defined
over a reference interpretation domain $(D, \sqsubseteq, \base)$, with assertions ranging
over a fixed language $\L \subseteq D$.
$\APPL$ logical judgements are triples $\as{h} \ \r \ \as{k}$, with $h, k \in \L$ and $\r\in \Reg$.

\begin{definition}[Validity]\label{de:validity}
  An $\APPL$ \emph{triple} $\as{h} \ \r \ \as{k}$
  is \emph{valid} if $\sem{\r}{h} \sqsubseteq k$. \hfill$\lozenge$
\end{definition}

\begin{figure*}[t]
\hspace{-20pt}
\resizebox{1.05\textwidth}{!}{
    \begin{tabular}{c}\label{tab:hyperHoare}
          \vspace{15pt}
          {\color{purple!80!black}$\inference[(basic)]{\sem{\e}{h} \sqsubseteq k}{\vdash \as{h} \ \e \ \as{k}}$
          \qquad
          $\inference[(seq)]{\vdash \as{h} \ \r_1 \ \as{k'} & \vdash \as{k'} \ \r_2 \ \as{k}}
          {\vdash \as{h} \ \r_1; \r_2 \ \as{k}}$ 
          \qquad
          $\inference[(cons)]{h \sqsubseteq h' & \vdash \as{h'} \ \r \ \as{k'} & k' \sqsubseteq k}
          {\vdash \as{h} \ \r \ \as{k}}$ 
          }
          \\
          \vspace{15pt}
          {\color{purple!80!black}$\inference[(choice)]{\vdash \as{h} \ \r_1 \ \as{k_1} & \vdash \as{h} \ \r_2 \ \as{k_2} & k_1\sqcup k_2 \sqsubseteq k} %
          {\vdash \as{h} \ \r_1 + \r_2 \ \as{k}}$
          \qquad
          $\inference[(iter)]{\forall i \in \N.\: \vdash \as{h_i} \ \r \ \as{h_{i+1}} & %
          \bigsqcup_{i \in \N} h_i \sqsubseteq k }{\vdash \as{h_0} \ \r^* \ \as{k}}$
          }
          \\
         \vspace{15pt} 
          {\color{purple!80!black}$\inference[(inv-fix)]{\vdash \as{h} \ \r \ {\as{h}} %
          }{\vdash \as{h} \ \fix{\r} \ \as{h}}$
          }
          \qquad 
          {\color{purple!80!black}$\inference[(join)]{\{h_i \mid i \in I\} \text{{\small ~dense}} & h \sqsubseteq \bigsqcup_{i \in I} h_i & \forall i \in I.\: \vdash \as{h_i} \ \r \ \as{k_i}& \bigsqcup_{i \in I} k_i \sqsubseteq k}{\vdash \as{h}\ \r \ \as{k}}$
          }
          \\
          \vspace{15pt}
           $\inference[(rec-fix)]{\vdash \as{h} \ \1 + \r \ \as{h'} & \vdash \as{h'} \ \fix{\r} \ \as{k} 
          }{\vdash \as{h} \ \fix{\r} \ \as{k}}$
          \qquad 
           $\inference[(widen)]{\vdash \as{h} \ \1 + \r \ \as{h'} %
           & \vdash \as{h \ \nabla\ h'} \ \fix{\r} \ \as {k} }{\vdash \as{h} \ \fix{\r} \ \as {k}}$
         \\
                \vspace{15pt}
          $\inference[(rec-star)]{\vdash \as{h} \ \r^* \ \as{k'} & \vdash \as{k'} \ \r \ \as{l} & %
          h\sqcup l\sqsubseteq k}{\vdash \as{h} \ \r^* \ \as{k}}$
           \qquad 
          $\inference[(narrow)]{\vdash \as{h} \ \fix{\r} \ \as{k} & \vdash \as{k} \ \r \ \as {k'} 
           & h \sqcup k' \sqsubseteq h' }{\vdash \as{h} \ \fix{\r} \ \as {k \ \Delta \ h'}}$
          \\ 
          \vspace{15pt}
          $\inference[(fix-star)]{\vdash \as{h} \ \fix{\r} \ \as{k}}{\vdash \as{h} \ \r^* \ \as{k}}$
           \qquad
           $\inference[(meet)]{h \sqsubseteq \bigsqcap_{i \in I} h_i & \forall i \in I.\: \vdash \as{h_i} \ \r \ \as{k_i} & \bigsqcap_{i \in I} k_i \sqsubseteq k}{\vdash \as{h}
          \ \r \ \as{k}}$
          \vspace{-12.5pt}
    \end{tabular}
    }
\caption{The program logic $\APPL$ over the interpretation domain $(D, \sqsubseteq, \base)$. 
In {\color{purple!80!black} purple} the rules that are necessary for relative completeness.}
\label{fig:rules}
\end{figure*}

The rules of $\APPL$ are in Figure~\ref{fig:rules}. Rule (basic) covers
elementary commands: its premise holds for any over-approximation of $\sem{\e}{h}$, which
is needed when $\sem{\e}{h}$ is not expressible in $\L$. This rule supports both forward axioms
(strongest post-condition) and backward ones: 
by the duality between the strongest
post-condition and weakest liberal pre-condition $\mathit{wlp}$,
$\sem{\e}\mathit{wlp}(\e,k) \sqsubseteq k$, so that $\vdash \!\as{\mathit{wlp}(\e,k)}\,\e\,\as{k}$
is an instance of (basic). Rules (seq) and (cons) are standard, with (cons) weakening
preconditions and strengthening postconditions along $\sqsubseteq$. Rule (choice) handles
nondeterminism. Its side condition $k_1 \sqcup k_2 \sqsubseteq k$, with $k\in \L$, absorbs the
approximation of the branch join. Rule (iter) reasons about iteration via a sequence of assertions
$\{h_i\}_{i\in\N}$ whose lub is approximated by $k\in \L$.

Iteration can also be handled by the \emph{finitary rules} 
(fix-star), 
(rec-star), (rec-fix), (inv-fix). Rule (fix-star) encodes that the fixpoint over-approximates the Kleene star
(Lemma~\ref{le:kleene-axioms}~\ref{le:kleene-axioms:fix}). The recursive rule (rec-star) unfolds the
star as $\1 + \r^*;\r$ (Lemma~\ref{le:kleene-axioms}~\ref{le:kleene-axioms:left-exp}), while
the recursive rule (rec-fix) unfolds the fixpoint by one Kleene
iteration. Repeated application strengthens the precondition until a loop invariant is
reached, then (inv-fix) discharges the fixpoint (or the derived (inv-star), obtained by
composing (inv-fix) with (fix-star), discharges the Kleene star). Rules (widen) and (narrow) assume widening
$\nabla$ and narrowing $\Delta$ are defined over the language $\L$, that is, $\nabla,\Delta : \L \times \L \to \L$. %
The rule (widen) is
(rec-fix) with $h'$ %
replaced by the accelerated $h\,\nabla\,h'$, and
(narrow) refines the invariant $k$ of $\vdash \as{h}\,\fix{\r}\,\as{k}$ to
$k\,\Delta\,h'$, as standard in abstract interpretation. 
The logic neither constructs $\nabla$ and $\Delta$ nor decides where to apply them: these are 
external inputs, and an iteration strategy becomes the \emph{shape} of a derivation, i.e., at which $\mathsf{fix}$ nodes (widen)
replaces (rec-fix), after how many (rec-fix) steps (delayed widening) and how often (narrow) is iterated. 
As for invariants in Hoare logic, the choice of shape is left outside the logic.

Finally, rules (join) and (meet) govern the lattice structure of $D$. Rule (join) crucially relies
on density (Definition~\ref{def:dense}): it lets $h$ be over-approximated by the
join of a dense decomposition $\{h_i\}_{i\in I}$ over which $\r$ is analysed, and is sound because
$\sem{\cdot}$ is additive on dense sets (Lemma~\ref{le:add-dense}).

It is worth stressing that an attempt at building a derivation for $\fix{\r}$ by systematically applying (widen) is guaranteed to
succeed, provided it uses no (cons) weakenings, since (widen) performs a widened Kleene
iteration. Concretely, to prove $\vdash \as{h_0}\,\fix{\r}\,\as{k}$ we derive a sound postcondition
$\vdash \as{h_0}\,\1+\r\,\as{h_0'}$, set $k_1' = h_0 \nabla h_0'$, and reduce to
$\vdash \as{k_1'}\,\fix{\r}\,\as{k}$. Iterating yields a sequence with $k_0' = h_0$ and
$k_{n+1}' = k_n' \,\nabla\, h_n'$ where $\sem{\r}k_n' \sqsubseteq h_n'$, which is guaranteed
to stabilise.

Since (widen) is (rec-fix) with $h\sqcup h'$ replaced by $h \, \nabla \, h'$, %
one might, for precision, widen only \emph{intermittently}: apply (widen) at an infinite set
$J\subseteq\N$ of unfoldings and the more precise (rec-fix) elsewhere. This is sound, but
one must be warned that such an ``infinitely often'' scheme no longer forces convergence
unless $J$ is cofinite %
(a counterexample is given in the Appendix as Remark~\ref{rem:no-alternation}).

\begin{toappendix}
\begin{remark}[Intermittent widening does not ensure convergence]
  \label{rem:no-alternation}
  Let us consider a widening on integer intervals defined as
  $\bot \nabla [a,b] = [a,b]$ and $[a,b] \nabla [a',b'] = [c,d]$
  where, as usual, the lower bound $c = -\infty$ if $a' \sqsubset a$ and
  $c = a$, otherwise.  Instead, the upper bound is
  \begin{center}
    $d =
    \begin{cases}
      b                   & \mbox{if  $b' \sqsubseteq b$}\\
      +\infty             & \mbox{if  $b' \sqsupset b$ and $b$ even}\\
      2 \lceil b'/2 \rceil & \mbox{if  $b' \sqsupset b$ and $b$ odd}\\
    \end{cases}
    $
  \end{center}
  where $2 \lceil b'/2 \rceil$ simply computes the smallest even integer greater than or equal to $b'$.
  One can easily show that $\nabla$ is indeed a widening. 
  
  Now, consider a simple loop that increments a variable:
\[
x = 0;\ ~\textnormal{\textbf{while}}~ \mathsf{true}~ \textnormal{\textbf{do}}~ (x = x + 1)
\]
  
  Unfolding the loop via the rule (rec-fix) only, without using (cons), corresponds to a standard Kleene iteration which produces the infinite increasing sequence of intervals $[0,i]$, $i \in \N$. By alternating (rec-fix) and (widen), i.e., using the widening on odd indices  $J=\{1,2,\ldots\}$ only, we obtain again a sequence which does not stabilise. 
  In fact, starting from $[0,0]$, we first use (rec-fix) and perform a join step, thus obtaining $[0,0] \sqcup [1,1] = [0,1]$. Next, the using (widen) we produce $[0,1] \ \nabla \ ([0, 1] \sqcup [1,2]) = [0,1] \ \nabla\  [0, 2] = [0,2]$, since the upperbound of $[0,1]$ is odd. The next (rec-fix) step would produce $[0,3]$ and so on. In general, the interval at step $2i$ will be $[0,2i+1]$ and at step $2i+1$ it will be $[0,2(i+1)]$, hence the sequence continues to increase.
\end{remark}
\end{toappendix}

$\APPL$'s proof system is always sound and, when the assertion language $\L$ is expressive enough,
also relatively complete. Soundness states that every derivable triple is valid.

\begin{theoremrep}[Soundness]\label{prop:sound}
  Let $h, k \in \L$ and $\r\in \Reg$. Then
  $\vdash \as{h} \ \r \ \as{k} \Rightarrow \sem{\r}{h} \sqsubseteq k$.
\end{theoremrep}

\begin{proof}
  By structural induction on the derivation. We consider the last rule applied in the derivation.

  \bigskip

  \noindent
  (basic) $\inference{\sem{\e}{h} \sqsubseteq k}{\vdash \as{h} \ \e \ \as{k}}$:
  we conclude by the side condition $\sem{\e}{h} \sqsubseteq k$.

  \bigskip\bigskip

  \noindent
  (seq) $\inference{\vdash \as{h} \ \r_1 \ \as{k'} & \vdash \as{k'} \ \r_2 \ \as{k}}
          {\vdash \as{h} \ \r_1; \r_2 \ \as{k}}$:
          \bigskip

          \noindent
          by inductive hypothesis we have $\sem{\r_1}h \sqsubseteq k'$ and $\sem{\r_2}k' \sqsubseteq k$, then:
        \begin{equation*}
          \sem{\r_1; \r_2}h \sqsubseteq \sem{\r_2}\left ( \sem{\r_1}h \right)
            \sqsubseteq \sem{\r_2}k'  \sqsubseteq k
        \end{equation*}
        where we used Lemma~\ref{le:kleene-axioms}~\ref{le:kleene-axioms:seq}. %

  \bigskip

    \noindent
  (cons) $\inference{ h \sqsubseteq h' & \vdash \as{h'} \ \r \ \as{k'} & k' \sqsubseteq k}
          {\vdash \as{h} \ \r \ \as{k}}$:

        \bigskip

        \noindent
        by inductive hypothesis we have $\sem{\r}h' \sqsubseteq k'$ and by side condition
        $h \sqsubseteq h'$, which implies $\sem{\r}h \sqsubseteq \sem{\r}h'$, and by side condition $k' \sqsubseteq k$,
        then
        \begin{equation*}
          \sem{\r}h \sqsubseteq \sem{\r}h' \sqsubseteq k' \sqsubseteq k\, .
        \end{equation*}

  \bigskip

  \noindent
  (choice) $\inference{
            \vdash \as{h} \ \r_1 \ \as{k_1} & \vdash \as{h} \ \r_2 \ \as{k_2} & k_1\sqcup k_2 \sqsubseteq k
            }
          {\vdash \as{h} \ \r_1 + \r_2 \ \as{k}}$:
          \bigskip

        \noindent
        by inductive hypothesis we have $\sem{\r_i}h \sqsubseteq k_i$ for $i = 1, 2$, and 
        by the side condition $k_1 \sqcup k_2 \sqsubseteq k$, hence
        \begin{equation*}
            \sem{\r_1 + \r_2}h = \sem{\r_1} h \sqcup \sem{\r_2} h
            \sqsubseteq k_1 \sqcup k_2 \sqsubseteq k 
        \end{equation*}
       where we used Lemma~\ref{le:kleene-axioms}~\ref{le:kleene-axioms:sum}. 

    \bigskip
  \noindent
  (iter) $\inference{\forall i \in \N.\: \vdash \as{h_i} \ \r \ \as{h_{i+1}} & \bigsqcup_{i \in \N} h_i \sqsubseteq k}{\vdash \as{h_0} \ \r^* \ \as{k}}$:
        \bigskip

        \noindent
         by inductive hypothesis we have $\sem{\r}h_{i} \sqsubseteq h_{i+1}$ for all $i \in \N$.
        For all $ i \geq 1$ it holds that
        $\sem{\r^i} h_{0} \sqsubseteq \sem{\r}h_{i-1}$. In fact, by induction on $i$:
        \begin{itemize}
            \item $\sem{\r^1} h_{0} = \sem{\r}h_{0}$, hence $\sem{\r} h_{0} \sqsubseteq \sem{\r}h_{0}$ by reflexivity of $\sqsubseteq$,
            \item $\sem{\r^{i+1}} h_{0} = \sem[]{\r^i; \r}h_0 \sqsubseteq \sem{\r}(\sem{\r^i} h_{0}) \sqsubseteq \sem{\r}(\sem{\r}h_{i-1}) \sqsubseteq \sem{\r}h_{i}$
            by Lemma~\ref{le:kleene-axioms}~\ref{le:kleene-axioms:seq} the two inductive hypotheses and the monotonicity of $\sem{\r}$.
        \end{itemize}
        Hence in particular by the inductive hypotheses we have that for all $i \geq 1$,
        $\sem{\r^i} h_{0} \sqsubseteq \sem{\r}h_{i-1} \sqsubseteq h_{i}$,
        hence by monotonicity of $\sqcup$ we have that
        \begin{equation}
            \bigsqcup_{i \geq 1} \sem{\r^i} h_0 \sqsubseteq \bigsqcup_{i \geq 1} h_i
            \tag{$\dagger$}
            \label{dagger}
        \end{equation}
        and therefore
        \begin{align*}
            \sem{\r^*}h_{0} %
            & = \bigsqcup_{i \in \N} \sem{\r^i}h_0 \\ 
            & \sqsubseteq \bigsqcup_{i \in \N} h_i & \mbox{[\eqref{dagger}, monotonicity of $\sqcup$ and $\sem{\r^0}h_0 = h_0$]}\\
            & \sqsubseteq k & \mbox{[side condition]}
        \end{align*}

  \bigskip

  \noindent
  (inv-fix) 
  $\inference{\vdash \as{h} \ \r \ \as{h} %
  }{\vdash \as{h} \ \fix{\r} \ \as{h}}$:

  \bigskip

  \noindent
  by inductive hypothesis $\sem{\r}h \sqsubseteq h$. We observe that for all $B \ni b \sqsubseteq h$ it holds $\bsem{\fix{\r}}b \sqsubseteq h$.
  In fact,
  $\psi_{\r, b}(h) = b \sqcup \sem{\r}h \sqsubseteq h \sqcup h = h$,
  i.e., $h$ is a post-fixpoint of $\psi_{\r, b}(h)$ and thus
  $\bsem{\fix{\r}}b = \lfp(\psi_{\r, b}) \sqsubseteq h$.
  Then we conclude recalling that
  $\sem{\fix{\r}}h = \bigsqcup_{b \in \ebase{h}} \bsem{\fix{\r}}b$.

  \bigskip

  \noindent
  (join) $\inference{\set{h_i}_{i \in I} \text{dense} & h \sqsubseteq \displaystyle\bigsqcup_{i \in I} h_i & \forall i \in I. \vdash \as{h_i} \ \r \ \as{k_i}& \displaystyle\bigsqcup_{i \in I} k_i \sqsubseteq k}{\vdash \as{h}
          \ \r \ \as{k}}$:
        \bigskip

        \noindent
         by inductive hypothesis we have $\sem{\r}(h_i) \sqsubseteq k_i$ for all $i \in I$, hence
        \begin{equation}
            \bigsqcup_{i \in I} \sem{\r}h_i \sqsubseteq \bigsqcup_{i \in I} k_i
            \tag{$\clubsuit$}
            \label{club}
         \end{equation}
         then
         \begin{align*}
            \sem{\r}h %
            & \sqsubseteq \sem{\r}\Bigl ( \bigsqcup_{i \in I} h_i \Bigr ) %
            & %
            [\mbox{side condition and monotonicity of } \sem{\r}] \\
            & = \bigsqcup_{i \in I} \sem{\r}(h_i) %
            & %
            [\mbox{additivity of } \sem{\r} \mbox{ on } \set{h_i \mid i \in I} \mbox{dense} ] \\
            & \sqsubseteq \bigsqcup_{i \in I} k_i %
            & %
            [\mbox{by~\eqref{club}}] \\
            & \sqsubseteq k %
            & %
            [\mbox{side condition}]
         \end{align*}

           \bigskip

  \noindent (rec-fix)
  $\inference{\vdash \as{h} \ \1 + \r \ \as{h'} & \vdash \as{h'} \ \fix{\r} \ \as{k} 
          }{\vdash \as{h} \ \fix{\r} \ \as{k}}$
          
\bigskip 
\noindent
For the left branch, the only applicable rule is (choice), whence we get 
\begin{equation*}
  \sem{\1 + \r} h = \sem{\1} h \sqcup \sem{\r}h = h \sqcup \sem{\r}h \sqsubseteq h'
\end{equation*}
By inductive hypothesis $\sem{\fix{\r}}h' \sqsubseteq k$, and from the previous step $h \sqsubseteq h'$, hence by monotonicity of $\sem{\fix{\r}}$: 
\begin{center}
  $\sem{\fix{\r}}h \sqsubseteq \sem{\fix{\r}}h' \sqsubseteq k \, .$
\end{center}

  \bigskip

  \noindent (widen)
  $\inference{\vdash \as{h} \ \1 + \r \ \as{h'} %
           & \vdash \as{h \ \nabla\ h'} \ \fix{\r} \ \as {k} }{\vdash \as{h} \ \fix{\r} \ \as {k}}$
    
           \bigskip 
   \noindent by inductive hypothesis, monotonicity of $\sem{\fix{\r}}$, and properties of widening: 
  \begin{center}
     $\sem{\fix{\r}} h \sqsubseteq \sem{\fix{\r}}(h \ \nabla \ h') \sqsubseteq k \, . $
   \end{center}

          \bigskip

  \noindent
  (rec-star) $\inference{\vdash \as{h} \ \r^* \ \as{k'} & \vdash \as{k'} \ \r \ \as{l} & h\sqcup l\sqsubseteq k}{\vdash \as{h} \ \r^* \ \as{k}}$:
        \bigskip

        \noindent
       applying the inductive hypothesis to the premise of the rule we obtain 
       $\sem[]{\r^*} h \sqsubseteq k'$ and $\sem[]{\r}k' \sqsubseteq l$.
 Therefore
\begin{align*}
  \sem[]{\r^*}h & \sqsubseteq \sem[]{\1 + (\r^*; \r)}h & \mbox{[by Lemma~\ref{le:kleene-axioms}~\ref{le:kleene-axioms:left-exp}]} \\
  & = h \sqcup \sem[]{\r^*; \r} \\ 
  & \sqsubseteq h \sqcup \sem[]{\r}(\sem[]{\r^*} h) & \mbox{[by Lemma~\ref{le:kleene-axioms}~\ref{le:kleene-axioms:seq}]} \\
  & \sqsubseteq h \sqcup \sem[]{\r}k' & \mbox{[by inductive hypothesis and monotonicity of $\sem[]{\r}, \sqcup$]} \\ 
  & \sqsubseteq h \sqcup l & \mbox{[by inductive hypothesis and monotonicity of $\sqcup$]} \\
  & \sqsubseteq k & \mbox{[by the side condition]} 
\end{align*}

  \bigskip

  \noindent  (narrow)
  $\inference{\vdash \as{h} \ \fix{\r} \ \as{k} & \vdash \as{k} \ \r \ \as {k'} 
           & h \sqcup k' \sqsubseteq h' }{\vdash \as{h} \ \fix{\r} \ \as {k \ \Delta \ h'}}$ 

  \bigskip
  \noindent by inductive hypotheses $\sem{\fix{\r}}h \sqsubseteq k$ and $\sem{\r}k \sqsubseteq k'$. 
  We need to prove $\sem{\fix{\r}}h \sqsubseteq k \ \Delta \ h$. To do so we exploit the properties of the narrowing and the side condition, 
 to prove $\sem{\fix{\r}}h \sqsubseteq k \sqcap (h \sqcup k') \sqsubseteq k \sqcap h' \sqsubseteq k \ \Delta \ h'$. 
  So we need to prove both $\sem{\fix{\r}}h \sqsubseteq k$, which is true by inductive hypothesis, and  $\sem{\fix{\r}}h \sqsubseteq  (h \sqcup k')$. 
  We prove the second inequality by proving that $\forall \base \ni b \sqsubseteq h$ it holds $\bsem{\fix{\r}}b \sqsubseteq (h \sqcup k')$, and conclude by 
  monotonicity of $\bigsqcup$, as  $\sem{\fix{\r}}h = \bigsqcup_{b \in \ebase{h}} \bsem{\fix{\r}}b$. 

  Let $ \base \ni b \sqsubseteq h$, define $a = \bsem{\fix{\r}}b = \lfp (\lambda x . b \sqcup \add{\bsem[]{\r}}x)$. Then by transititivity $a \sqsubseteq \sem{\fix{\r}}h \sqsubseteq k$, 
  and by monotonicity and the inductive hypothesis we get 
  \begin{center}
    $a = b \sqcup \sem{\r} a \sqsubseteq  b \sqcup \sem{\r} k \sqsubseteq h \sqcup k'$.
  \end{center}

    \bigskip

  \noindent (fix-star)
  $\inference{\vdash \as{h} \ \fix{\r} \ \as{k}}{\vdash \as{h} \ \r^* \ \as{k}}$ \bigskip

  \noindent by Lemma~\ref{le:kleene-axioms}~\ref{le:kleene-axioms:fix}, inductive hypothesis and transitivity of $\sqsubseteq$: 
  \begin{center}
    $\sem{\r^*} h \sqsubseteq \sem{\fix{\r}}h \sqsubseteq k \, .$
  \end{center}

  \bigskip

  \noindent
  (meet) $\inference{h \sqsubseteq \bigsqcap_{i \in I} h_i & \forall i \in I.\: \vdash \as{h_i} \ \r \ \as{k_i} & \bigsqcap_{i \in I} k_i \sqsubseteq k}{\vdash \as{h}
          \ \r \ \as{k}}$:
        \bigskip

        \noindent
         by inductive hypothesis we have $\sem{\r}h_i \sqsubseteq k_i$ for all $i \in I$, hence
         \begin{equation}
         \bigsqcap_{i \in I} \sem{\r}h_i \sqsubseteq \bigsqcap_{i \in I} k_i
         \tag{$\diamondsuit$}
            \label{diamond}
         \end{equation}
         and for all $i \in I$,
         $\sem{\r}\left ( \bigsqcap_{i \in I} h_i \right ) \sqsubseteq \sem{\r_i}h_i $ by monotonicity of $\sem{\r}$,
         hence \begin{equation}
            \sem{\r}\Bigl ( \bigsqcap_{i \in I} h_i \Bigr ) \sqsubseteq \bigsqcap_{i \in I} \sem{\r} h_i
            \tag{$\heartsuit$}
            \label{heart}
         \end{equation}
         then \begin{align*}
            \sem{\r}h %
            & \sqsubseteq \sem{\r}\Bigl ( \bigsqcap_{i \in I} h_i \Bigr ) %
            & %
            [\mbox{monotonicity of } \sem{\r}] \\
            & \sqsubseteq \bigsqcap_{i \in I} \sem{\r} (h_i) %
            & %
            [(\mbox{by \ref{heart}})] \\
            & \sqsubseteq \bigsqcap_{i \in I} k_i  %
            & %
            [\mbox{by~(\ref{diamond})}] \\
            & \sqsubseteq k %
            & %
            [\mbox{side condition}] \tag*{\qedhere}
         \end{align*}
\end{proof}

When $\L$ is \emph{expressive}, 
in the sense of Definition~\ref{de:expressive-language} below,
$\APPL$ is also
relatively complete. This relative completeness proof rests on strongest postconditions being always derivable in $\APPL$.

\begin{definition}[Expressive language]\label{de:expressive-language}
  Let $(D, \sqsubseteq, \base)$ be an interpretation domain, $\L$ a language over it, and
  $\L_{\base} = \L \cap \base$. Then $\L$ is \emph{expressive} if:
  \begin{enumerate*}[{\rm (C$_1$)}]
   \item\label{de:expressive-language:sem-close} $\bsem{\r}(\L_{\base}) \subseteq \L$ for all $\r \in \Reg$;\hfill
   \item\label{de:expressive-language:quant} $\ebase{h} \subseteq \down \L_{\base}(h)$ for all $h \in \L$;\hfill
   \item\label{de:expressive-language:join} $\L$ is closed under joins.\hfill $\lozenge$
  \end{enumerate*}
\end{definition}
In the typical setting where $\L$ represents all basis elements,
  i.e., $\L_{\base} = \base$,
  condition~\ref{de:expressive-language:quant} becomes vacuous,
  \ref{de:expressive-language:join} is closure under disjunction, and,
  given these two, \ref{de:expressive-language:sem-close} reduces to
  closure under elementary commands, namely,
  $\bsem{\e}(\L_{\base}) \subseteq \L$
  (Lemma~\ref{le:compl-closure-semantics-oplus}).  When $\L = D$, as
  for an abstract domain used as its own assertion language, the three
  conditions hold trivially, giving unconditional relative completeness for static analysis.
  When $D = \pow{\Sigma}$
  and $\L$ is a first-order specification language,
  \ref{de:expressive-language:join} may fail, as $\L$
  is generally not closed under arbitrary disjunctions. However, the proof of
  Lemma~\ref{prop:compl-condit} uses~\ref{de:expressive-language:join}
  only to obtain $\sem{\r}h \in \L$ for $h \in \L$,
  and weakening it
  to this requirement yields the classical expressiveness
  of~\citet{Cook78}, which first-order arithmetic satisfies.
  Conversely, when $D$ is a concrete domain and $\L = A$ is a proper
  abstraction (e.g., in the setting of~\cite{CCLB:AIFR}),
  expressiveness fails in general. This is the expected incompleteness
  of the abstraction, and Proposition~\ref{prop:equivalence} shows
  that nothing further is lost: the abstract logic $\APPL_A$ has the
  same deductive power as $\APPL$ with $\L = A$ whenever $D$ is
  completely meet-distributive, as for powersets, or $\alpha$ reflects
  density.

\begin{lemmarep}[Derivability of strongest postcondition]\label{prop:compl-condit} 
  For an interpretation domain $(D, \sqsubseteq, \base)$ and an expressive language $\L$,
  the triple $\vdash \as{h} \ \r \ \as{\sem{\r}h}$ is derivable in {\rm $\APPL$} for all
  $\r \in \Reg$ and $h \in \L$.
\end{lemmarep}

\begin{proof}
  We prove the result holds on every representable basis element $b \in \L_{\base} \subseteq \base$.  %
   The result on $\sem{\r}$ then follows by application of the (join) rule:
   let $h \in \L$, then by hypothesis~\ref{de:expressive-language:quant} of Definition~\ref{de:expressive-language} $h = \bigsqcup_{b \in \ebase{h} \cap \L_{\base}} b $, and %
   $\ebase{h} \cap \L_{\base}$ is dense since for every $b \in \base$ such that $b \sqsubseteq \sqcup \ebase{h} \cap \L_{\base} = h$ there exists $b' \in \ebase{h} \cap \L_{\base}$ such that $b \sqsubseteq b'$. 
   Then, assume $\vdash \as{b} \ \r \ \as{\bsem{\r}b}$ is derivable for all $b \in \ebase{h} \cap \L_{\base}$, so we can apply
   \begin{equation*}
    \inference[(join)]{\ebase{h} \cap \L_{\base} \text{ dense} %
    & \forall b \in \ebase{h} .\: \vdash \as{b} \ \r \ \as{\bsem{\r}b}& \bigsqcup_{b \in \ebase{h} \cap \L_{\base}} \bsem{\r}b = \sem{\r}h}{\vdash \as{h}
          \ \r \ \as{\sem{\r}h}}
      \end{equation*}
      Note that $\bigsqcup_{b \in \ebase{h} \cap \L_{\base}} \bsem{\r}b = \sem{\r}h \in \L$ by closure under joins of $\L$~\ref{de:expressive-language:join}.
      
   \noindent Let $b \in \L_{\base}$. We prove that $\vdash \as{b} \ \r \ \as{\sem{\r}b}$ by structural induction on the command syntax. 
 
   \begin{itemize}
    \item If $\r = \e$ then since by~\ref{de:expressive-language:sem-close} it holds $\bsem{\e}b = k \in \L$, just apply $\inference[(basic)]{\bsem{\e}{b} = k}{\vdash \as{b} \ \e \ \as{k}}$.
    \bigskip
     \item If $\r = \r_1 + \r_2$, by definition $\bsem{\r_1 + \r_2}b=\bsem{\r_1} b \sqcup \bsem{\r_2} b$, and by~\ref{de:expressive-language:sem-close} $\bsem{\r_1 + \r_2}b \in \L$.
       We use the inductive hypothesis on $b$ to know that $\vdash \as{b} \ \r_i \ \as{\bsem{\r_i} b}$ are derivable for $i = 1, 2$. Then
       \begin{equation*}
           \inference[(choice)]{\vdash \as{b} \ \r_1 \ \as{\bsem{\r_1} b} & \vdash \as{b} \ \r_2 \ \as{\bsem{\r_2} b} & \bsem{\r_1}b \sqcup \bsem{\r_2}b = k_b}{\vdash \as{b} \ \r_1 + \r_2 \ \as{k_b}}
      \end{equation*}
     \item If $\r = \r_1; \r_2$, by definition $\bsem{\r_1; \r_2}b= \add{\bsem{\r_2}}(\sem{\r_1}b)$, by~\ref{de:expressive-language:sem-close} $\bsem{\r_1; \r_2}b \in \L$, 
     and by inductive hypothesis we have $\vdash \as{b} \ \r_1 \ \as{\bsem{\r_1}b}$
     and that all  $\vdash \as{b'} \ \r_2 \ \as{\bsem{\r_2}b'}$ for $b' \in \ebase{\bsem{\r}b} \cap \L_{\base}$ are derivable.
     Then we can apply the rule (join), since $\ebase{\bsem{\r_1}b} \cap \L_{\base}$ is dense, 
     to get \begin{equation*}
      \inference[(join)]{\forall b' \in \ebase{\bsem{\r_1}b} .\: \vdash 
        \as{b'} \ \r_2 \ \as{\bsem{\r_2}b' } & \bigsqcup_{b' \in \ebase{\bsem{\r_1}b}} \bsem{\r_2}b' = \add{\bsem{\r_2}}(\bsem{\r_1}b)}{\vdash \as{\bsem{\r_1}b} \ \r_2 \ \as{\add{\bsem{\r_2}}(\bsem{\r_1}b)}}
     \end{equation*}
     then we can derive
     \begin{equation*}
         \inference[(seq)]{\vdash \as{b} \ \r_1 \ \as{\bsem{\r_1}b} & \vdash \as{\sem{\r_1}b} \ \r_2 \ \as{\add{\bsem{\r_2}}(\bsem{\r_1}b)}}
       {\vdash \as{b} \ \r_1; \r_2 \ \as{\add{\bsem{\r_2}}(\bsem{\r_1}b)}}
     \end{equation*}
     \item If $\r = \r_1^*$,
     by definition
     $\bsem{\r_1^*}b = \bigsqcup_{i \geq 0} \bsem{\r_1^i} b$, and by~\ref{de:expressive-language:sem-close} $\bsem{\r_1^*}b \in \L$. 
    Define $k_0 = b$ and $k_{i+1} = \add{\bsem{\r_1}}k_i$ for all $i \geq 0$.
        Now we show that $k_i = \bsem{\r_1^i} b$ for all $i \geq 0$ by induction on $i$.
    \begin{itemize}
        \item $k_0 = b = \bsem{\r_1^0}b$.
         \item $k_{i+1} = \add{\bsem{\r_1}}k_i = \add{\bsem{\r_1}}(\bsem{\r_1^i} b) = \bsem{\r_1^{i+1}} b$ by the inductive hypothesis.
    \end{itemize}
    Hence $k_i \in \L$ for all $i$ by~\ref{de:expressive-language:sem-close}.
    Then by inductive hypothesis we have that for all $i \geq 0$ each $\vdash \as{b'} \ \r_1 \as{\bsem{\r_1}b'}$ is derivable for $b' \in \ebase{k_i} \cap \L_{\base}$, 
    hence again by applying the rule (join) for $b' \in \ebase{k_i} \cap \L_{\base}$ (which again by~\ref{de:expressive-language:quant} is dense), since $k_{i+1} = \add{\bsem{\r_1}}k_i = \bigsqcup_{b' \in \ebase{k_i} \cap \L_{\base}} \bsem{\r_1}b'$,
    we have that 
    $\vdash \as{k_i} \ \r_1 \ \as{k_{i+1}}$ is derivable for all $i \geq 0$. 
    Then we can derive
    \begin{equation*}
             \inference[(iter)]{
                \forall i \in \N &  \vdash \as{k_i} \ \r_1 \ \as{k_{i+1}} & \bigsqcup_{i \in \N} k_i = k_b 
              }{
                \vdash \as{k_0} \ \r_1^* \ \as{k_b}
                }
            \end{equation*}
         which is, as $k_0 = b$ and $\bigsqcup_{i \in \N} k_i = \bigsqcup_{i \in \N} \bsem{\r_1^i} b$,
         \begin{equation*}
             \inference[(iter)]{
                \forall i \in \N & \vdash \as{k_i} \ \r_1 \ \as{k_{i+1}} & \bigsqcup_{i \in \N}\bsem{\r_1^i} b = k_b
              }{
                \vdash \as{b} \ \r_1^* \ \as{k_b}
                }
       \end{equation*}
    \item If $\r = \fix{\r_1}$, by definition $\bsem{\fix{\r_1}}b = \lfp_{\sqsubseteq} \Bigl ( \lambda x . b \sqcup \add{\bsem{\r_1}}x \Bigr)$. Let $k = \lfp_{\sqsubseteq} \Bigl ( \lambda x . b \sqcup \add{\bsem{\r_1}}x \Bigr)$, 
    i.e. $k = \bsem{\fix{\r_1}}b \in \L$ by~\ref{de:expressive-language:sem-close}. Then $k = b \sqcup \add{\bsem{\r_1}}k$, hence in particular $b \sqsubseteq k$ and $\sem{\r_1}k \sqsubseteq k$. Then we can derive 
    \begin{equation*}
      \inference[(cons)]{b \sqsubseteq k & \inference[(inv-fix)]{
        \inference[(cons)]{
\vdash \as{k} \ \r_1 \ \as{\sem{\r_1}k}
        }{\vdash \as{k} \ \r_1 \ \as{k}}
      }
        { \vdash \as{k} \ \fix{\r_1} \ \as{k}} & k = \sem{\fix{\r_1}}b}{
          \vdash \as{b} \ \fix{\r_1} \ \as{\sem{\fix{\r_1}}b}}
    \end{equation*}
    where $\vdash \as{k} \ \r_1 \ \as{\sem{\r_1}k}$ is derivable by inductive hypothesis.
    \hfill \qedhere
    \end{itemize}
\end{proof}

\begin{theoremrep}[Relative completeness]
  \label{prop:weak-rel-compl}
  Let $(D, \sqsubseteq, %
  \base)$ be an interpretation domain and let $\L$ be an expressive language over it. 
  Let $h,k \in \L$ and $\r\in \Reg$. Then %
 $\sem{\r}h \sqsubseteq k  
 \ \Rightarrow \;\; %
 \vdash \as{h} \  \r \ \as{k}$. %
\end{theoremrep}

\begin{proof}
  Since $\L$ is expressive we apply Lemma~\ref{prop:compl-condit} to derive $\vdash \as{h} \ \r \ \as{\sem{\r}h}$, 
  and conclude using rule (cons).
  \hfill
\end{proof}

Completeness for languages forming an abstract domain is studied in
\S~\ref{sec:abs-sublogic}. Relative completeness uses only (basic), (seq), (cons),
(choice), (iter), (inv-fix), and (join), i.e., the {\color{purple!80!black} purple} rules in Figure~\ref{fig:rules}. 
The rules (rec-star), (rec-fix), (fix-star),
(widen), (narrow), and (meet) are auxiliary.

Inspecting the completeness proof shows that, if the Kleene star is
removed from the language, relative completeness holds without the infinitary rule (iter).

\begin{corollary}[Finitary rules for $\fix{\cdot}$]
  \label{co:remove-iter}
  For the restricted language $\r ::= \e \mid \r;\r \mid \r + \r \mid \fix{\r}$, the deduction system
  consisting of (basic), (seq), (cons), (choice), (inv-fix), and (join) is relatively
  complete.
\end{corollary}

The next example shows that, for the full language, the deduction system without (iter)
is no longer relatively complete, even with the auxiliary rules. This provides an
\emph{a posteriori} justification for the use of fixpoints instead of the Kleene star in
abstract interpretation frameworks: while the Kleene star is more precise, it requires
infinitary rules for completeness.

\begin{example}\label{ex:incompleteness}
  Consider the program
  $\ti \triangleq ((x\,\mathsf{mod}\,2\,{=}\,1?);x:=x-2) + (x:=x+2)$
  on $(\Int, \sqsubseteq, \Ir{\Int})$ from Example~\ref{ex:fix-and-star}, where we
  observed that $\sem{\ti^*}[0,0] = [0,+\infty]$. The triple
  $\as{[0,0]}\,\ti^*\,\as{[0,+\infty]}$ is derivable by the infinitary rule (iter):
  \begin{center}
    \small
    $\inference[(iter)]{\forall i \in \N & \vdash \as{[2i,2i]}\ \ti\ \as{[2(i+1),2(i+1)]}}
    {\vdash \as{[0,0]}\ \ti^*\ \as{[0,+\infty] = \bigsqcup_{i\in\N}[2i,2i]}}$
  \end{center}
  \smallskip
  However, $\as{[0,0]}\,\ti^*\,\as{[0,+\infty]}$ is not derivable without (iter):
  the finitary rules only yield the trivial
  triple $\as{[0,0]}\,\ti^*\,\as{[-\infty,+\infty]}$. In fact, by an easy induction on
  derivation trees, if $\vdash \as{h}\,\ti^*\,\as{k}$ with $0 \in h$ is derivable without
  (iter), then $k = [-\infty,+\infty]$.
  If the last rule is (fix-star), then, by soundness,
  $k \sqsupseteq \sem{\fix{\ti}}h \sqsupseteq \sem{\fix{\ti}}[0,0] = [-\infty,+\infty]$.
  If the last rule is (rec-star), with premises $\vdash \as{h}\,\ti^*\,\as{k'}$ and
  $\vdash \as{k'}\,\ti\,\as{l}$, then $k' = [-\infty,+\infty]$ by the inductive
  hypothesis, and thus
  $k \sqsupseteq l \sqsupseteq \sem{\ti}k' = [-\infty,+\infty]$ by soundness.
  If the last rule is (cons) or (join), then $0$ belongs to the precondition of some
  premise (for (join), by density), whose postcondition is $[-\infty,+\infty]$ by the
  inductive hypothesis, so that $k = [-\infty,+\infty]$. If it is (meet), then $0$
  belongs to the precondition of every premise, so all their postconditions are
  $[-\infty,+\infty]$, and therefore $k = [-\infty,+\infty]$ (also when there are no premises, because
  $k \sqsupseteq \bigsqcap \varnothing = [-\infty,+\infty]$).
  Also observe that $\sem{\ti^*}[0,0] = [0,+\infty] \sqsubset [-1,+\infty] =
  \sem{\1 + (\ti^*;\ti)}[0,0]$, showing that the inequality
  of Lemma~\ref{le:kleene-axioms}~\ref{le:kleene-axioms:left-exp} can be strict.
  \hfill$\lozenge$
\end{example}

Varying the basis in $(D, \sqsubseteq, \base)$ induces not only a hierarchy of semantics
but also one of proof systems: enlarging the basis weakens the system, as density becomes
more demanding and the (join) rule less applicable. It turns out that with the trivial basis $\base = D$, a family
$\set{h_i}$ is dense only if it contains $h = \bigsqcup_{i \in I} h_i$, so (join) is
subsumed by the (cons) rule and becomes redundant.

Over the concrete domain $(\pow{\Sigma}, \subseteq, \base)$ one easily recovers Hoare logic:
with $\base = \Ir{\pow{\Sigma}}$ every set is dense, so the (join) rule applies unconstrained,
whereas with $\base = \pow{\Sigma}$ it becomes redundant. Since both bases induce the same
collecting semantics, we recover the familiar fact that in Hoare logic the join rule is
sound but redundant for completeness. This redundancy is an instance of a more general
result tying the necessity of the (join) rule to properties of the interpretation domain 
(properties that, e.g., the $\mathsf{Sign}$ abstract domain satisfies).

\begin{propositionrep}[Completeness without (join)]\label{pr:nojoin}
  For an interpretation domain $(D, \sqsubseteq, \base)$ with irreducible basis and an
  expressive language $\L$, if the full semantics $\sem{\cdot}$ is additive, then
  {\rm (basic)}, {\rm (seq)}, {\rm (choice)}, {\rm (fix-star)}, {\rm (inv-fix)}, and
  {\rm (cons)} form a relatively complete proof system.
\end{propositionrep}

\begin{proof}
   We show that we can deduce the post %
  and conclude as in Theorem~\ref{prop:weak-rel-compl}. 
  Just note that, by the expressiveness hypotheses~(\ref{de:expressive-language:sem-close}--\ref{de:expressive-language:join}) in Lemma~\ref{prop:compl-condit}, we deduce that $\sem{\r} h \in \L$ for all $h \in \L$. Thus, using
   Proposition~\ref{pr:inductive}, the result can be easily proved by induction on the command syntax. Let $h \in \L$.
\begin{itemize}
   \item If $\r = \e$ then just apply rule (basic)
   \begin{center}
     $\inference[(basic)]{\sem{\e}{h} = k}{\vdash \as{h} \ \e \ \as{k}}$. 
   \end{center}
   \smallskip
      \item If $\r = \r_1 + \r_2$, $\sem{\r_1 + \r_2}h=\sem{\r_1} h \sqcup \sem{\r_2} h$.
     By inductive hypothesis $\vdash \as{h} \ \r_i \ \as{\sem{\r_i} h}$ are derivable for $i = 1, 2$. Then
        \begin{equation*}
         \inference[(choice)]{\vdash \as{h} \ \r_1 \ \as{\sem{\r_1} h} & \vdash \as{h} \ \r_2 \ \as{\sem{\r_2} h} & \sem{\r_1}d \sqcup \sem{\r_2}d = k}{\vdash \as{h} \ \r_1 + \r_2 \ \as{k}}
       \end{equation*}
     \item If $\r = \r_1; \r_2$, by Proposition~\ref{pr:inductive}, 
     $\sem{\r_1; \r_2}h= \sem{\r_2}(\sem{\r_1}h)$, 
     and by inductive hypothesis we have $\vdash \as{h} \ \r_1 \ \as{\sem{\r_1}h}$
     and that all  $\vdash \as{\sem{\r_1}h} \ \r_2 \ \as{\sem{\r_2}(\sem{\r_1}h)}$ are derivable.
     Then we can derive
     \begin{equation*}
         \inference[(seq)]{\vdash \as{h} \ \r_1 \ \as{\sem{\r_1}h} & \vdash \as{\sem{\r_1}h} \ \r_2 \ \as{\sem{\r_2}(\sem{\r_1}h)}}
       {\vdash \as{h} \ \r_1; \r_2 \ \as{\sem{\r_2}(\sem{\r_1}h)}}
     \end{equation*}
   \item If $\r = \r_1^*$, %
   by Proposition~\ref{pr:inductive},
     $\sem{\r_1^*}h = \sem{\fix{\r_1}}h$, hence we can conclude using the rule (inv-fix) as in the $\r = \fix{\r_1}$ case, reducing to it by means of rule (fix-star):
     \begin{equation*}
     \inference[(fix-star)]{\vdash \as{h} \ \fix{\r} \ \as{k}}{\vdash \as{h} \ \r^* \ \as{k}}
     \end{equation*} %
  \item If $\r = \fix{\r_1}$ by Proposition~\ref{pr:inductive}, $\sem{\fix{\r_1}} h = \lfp_{\sqsubseteq} (\lambda x . h \sqcup \sem{\r_1} x )$ (as in~\S~\ref{inst:simple-inductive}).
    Let $k = \lfp_{\sqsubseteq} \Bigl ( \lambda x . h \sqcup \sem{\r_1}x \Bigr )$, 
    i.e. $k = \sem{\fix{\r_1}}h$, then $k = h \sqcup \sem{\r_1}k$, hence in particular $h \sqsubseteq k$ and $\sem{\r_1}k \sqsubseteq k$. Then we can derive 
    \begin{equation*}
      \inference[(cons)]{h \sqsubseteq k & \inference[(inv-fix)]{
        \inference[(cons)]{
\vdash \as{k} \ \r_1 \ \as{\sem{\r_1}k}
        }{\vdash \as{k} \ \r_1 \ \as{k}}
      }
        { \vdash \as{k} \ \fix{\r_1} \ \as{k}} & k = \sem{\fix{\r_1}}h}{
          \vdash \as{b} \ \fix{\r} \ \as{\sem{\fix{\r_1}}h}}
    \end{equation*}
    where $\vdash \as{k} \ \r_1 \ \as{\sem{\r_1}k}$ is derivable by inductive hypothesis. 
    \hfill \qedhere
    \end{itemize}
 \end{proof}

\subsection{A Relational Example: Octagons}\label{sub:octagons}

We illustrate the use of $\APPL$ over the octagon abstract domain $\Oct$ \cite{Mine06}, a well-known and widely used relational abstract domain that refines intervals by representing constraints between pairs of variables of the form $\pm x \pm y \leq k$.
Consider the program  %
\begin{center}
  $\po \triangleq \textbf{while~} x<4 \textbf{~do~} (z:=-x*y;\ x:=x+1;\ y:=y-1)$
\end{center}
over integer stores, where $x,y$ are anti-correlated counters (a relation exactly represented in
$\Oct$), and $z$
stores their sign-reversed product, from input $O_0 \triangleq \{\,x={-}4,\ y=4,\ z=0\,\}$
(via the initialisation $x:={-}4;\ y:=4;\ z:=0$). Let $\bo \triangleq (x<4?);\ z:=-x*y;\ x:=x+1;\ y:=y-1$ be
the guarded body of $\po$. We display octagons as constraint sets, neglecting their difference bound matrix (DBM) closure~\cite{Mine06}. 
In particular, we perform the analysis on the interpretation domain $(\Oct, \sqsubseteq, \Ir{\Oct})$, where %
$\Ir{\Oct}$ is the set of singletons $\set{p}$, where $p$ is a point. 
The
transfer functions of $\Oct$ are exact on octagonal tests and assignments ($x<4?$,
$x:=x+1$, $y:=y-1$), but the nonlinear assignment $z:=-x*y$ is not octagonal: its transfer function
falls back on the interval evaluation $-^\sharp([x]*^{\sharp}[y])$ of its right-hand side expression $-x*y$,
dropping every constraint relating $z$ to the other variables~\cite{Mine06}.

The strongest concrete loop invariant of $\po$ consists of the following stores $(x,y,z)$:
\begin{center}
  \vspace{-10pt}
  \resizebox{\textwidth}{!}{
$I\triangleq \{({-}4,4,0),\,({-}3,3,16),\,({-}2,2,9),\,({-}1,1,4),\,(0,0,1), %
(1,{-}1,0),\,(2,{-}2,1),\,(3,{-}3,4),\,(4,{-}4,9)\}$. 
}
\end{center}
Since $x+y=0$ holds throughout and octagons propagate it exactly, $z$ stores
$-x\cdot y = -x\cdot({-}x)=x^{2}\ge 0$, i.e., $z\ge0$ always holds. Standard octagon analysis with widening and
narrowing infers $\Ostd \triangleq \{\,{-}4\le x\le 4,\ x+y=0,\ {-}12\le z\le 16\,\}$,
which is optimal on $x,y$ and on $x+y=0$ but with spurious negative $z$. On the other hand, a single application of the 
(join) rule allows us to infer
$\Onew \triangleq \{\,{-}4\le x\le 4,\ x+y=0,\ 0\le z\le 16\,\}$, thus recovering the
sign of $z$, with optimal bounds on every variable and on $x+y$.

\begin{figure}[!t]
\begin{subfigure}{0.48\textwidth}
  \resizebox{\textwidth}{!}{
\begin{tikzpicture}[
  jd/.style ={anchor=west, inner sep=1.6pt},
  lbl/.style={anchor=east, inner sep=1.6pt},
  >=Stealth,
]
\def\ind{.75}\def\rh{0.65}\def\base{-0.2}
\pgfmathsetmacro\arl{0.30}
\pgfmathsetmacro\txtgap{0.42}

\newcommand{\trunk}[5]{%
  \pgfmathsetmacro\cx{\base + #2*\ind - 0.16}%
  \pgfmathsetmacro\ty{-#3*\rh + 0.32}%
  \pgfmathsetmacro\by{-#4*\rh}%
  \draw[#1, line width=0.8pt] (\cx,\ty) -- (\cx,\by);%
  \foreach \rrow in {#5}{%
    \pgfmathsetmacro\ry{-\rrow*\rh}%
    \draw[#1, line width=0.8pt]
      (\cx,\ry) -- (\cx+\arl,\ry);%
  }%
}
\newcommand{\row}[4]{%
 \node[lbl] at (-2,-#1*\rh) {\scriptsize #1};%
  \node[lbl] at (0,-#1*\rh) {\rm #3};%
  \node[jd]  at (\base+#2*\ind+\txtgap,-#1*\rh) {$#4$};%
}

\trunk{lvlA}{1}{2}{10}{2,9,10}   %
\trunk{lvlB}{2}{3}{4}{3,4}       %
\trunk{lvlC}{3}{5}{6}{5,6}       %
\trunk{lvlD}{4}{7}{7}{7}         %
\trunk{lvlE}{5}{8}{8}{8}         %

\row{1}{0}{(narrow)}{\mathbf{s}_0 \triangleq \vdash \as{O_0}\ \colcode{\fix{\bo}}\ \as{\Ostd \triangleq W'~\Delta ~U'}}
\row{2}{1}{(widen)}{{\color{lvlA} \mathbf{s}_1  \triangleq \ } \vdash \as{O_0}\ \colcode{\fix{\bo}}\ \as{W'}}
\row{3}{2}{(choice*)}{{\color{lvlB} \mathbf{s}_2 \triangleq }\ \vdash \as{O_0}\ \colcode{\1 + \bo}\ \as{O_0 \sqcup O_1}}
\row{4}{2}{(widen)}{{\color{lvlB} \mathbf{s}_3  \triangleq \ } \vdash \as{W\triangleq O_0\ \nabla\ (O_0\sqcup O_1)}\ \colcode{\fix{\bo}}\ \as{W'}}
\row{5}{3}{(choice*)}{{\color{lvlC} \mathbf{s}_4 \triangleq }\ \vdash \as{W}\ \colcode{\1 + \bo}\ \as{W\sqcup U}}
\row{6}{3}{(inv-fix)}{{\color{lvlC} \mathbf{s}_5 \triangleq }\vdash \as{W'\triangleq W\ \nabla\ (W\sqcup U)} \ \colcode{\fix{\bo}} \ \as{W'}}
\row{7}{4}{(cons)}{{\color{lvlD} \mathbf{s}_6 \triangleq \ }\vdash \as{W'}\ \colcode{\bo}\ \as{W'}}
\row{8}{5}{(seq)+(basic)}{{\color{lvlE} \mathbf{s}_7 \triangleq \ } \vdash \as{W'}\ \colcode{\bo}\ \as{U}}
\row{9}{1}{(seq)+(basic)}{{\color{lvlA} \mathbf{s}_7 \triangleq \ } \vdash \as{W'}\ \colcode{\bo}\ \as{U}}
\row{10}{1}{}{U' \triangleq O_0\sqcup U}
\end{tikzpicture}%
  }
\caption{Standard octagon analysis in $\APPL$.}\label{fig:octa}
\end{subfigure}
\hfill 
\begin{subfigure}{0.48\textwidth}
  \resizebox{\textwidth}{!}{
\begin{tikzpicture}[
  jd/.style ={anchor=west, inner sep=1.6pt},
  lbl/.style={anchor=east, inner sep=1.6pt},
  >=Stealth,
]
\def\ind{.75}\def\rh{0.65}\def\base{-0.2}
\pgfmathsetmacro\arl{0.30}
\pgfmathsetmacro\txtgap{0.42}

\newcommand{\trunk}[5]{%
  \pgfmathsetmacro\cx{\base + #2*\ind - 0.16}%
  \pgfmathsetmacro\ty{-#3*\rh + 0.32}%
  \pgfmathsetmacro\by{-#4*\rh}%
  \draw[#1, line width=0.8pt] (\cx,\ty) -- (\cx,\by);%
  \foreach \rrow in {#5}{%
    \pgfmathsetmacro\ry{-\rrow*\rh}%
    \draw[#1, line width=0.8pt]
      (\cx,\ry) -- (\cx+\arl,\ry);%
  }%
}
\newcommand{\row}[4]{%
 \node[lbl] at (-2,-#1*\rh) {\scriptsize #1};%
  \node[lbl] at (0,-#1*\rh) {\rm #3};%
  \node[jd]  at (\base+#2*\ind+\txtgap,-#1*\rh) {$#4$};%
}

\trunk{lvlA}{1}{2}{12}{2,11,12}  %
\trunk{lvlB}{2}{3}{4}{3,4}       %
\trunk{lvlC}{3}{5}{6}{5,6}       %
\trunk{lvlD}{4}{7}{7}{7}         %
\trunk{lvlE}{5}{8}{8}{8}         %
\trunk{lvlF}{6}{9}{10}{9,10}     %

\row{1}{0}{(narrow)}{\mathbf{t}_0 \triangleq \vdash \as{O_0}\ \colcode{\fix{\bo}}\ \as{\Onew \triangleq V \ \Delta \ U'}}
\row{2}{1}{(widen)}{{\color{lvlA} \mathbf{t}_1  \triangleq \ } \vdash \as{O_0}\ \colcode{\fix{\bo}}\ \as{V}}
\row{3}{2}{(choice*)}{{\color{lvlB} \mathbf{t}_2 \triangleq }\ \vdash \as{O_0}\ \colcode{\1 + \bo}\ \as{O_0 \sqcup O_1}}
\row{4}{2}{(widen)}{{\color{lvlB} \mathbf{t}_3  \triangleq \ } \vdash \as{W\triangleq O_0\ \nabla\ (O_0\sqcup O_1)}\ \colcode{\fix{\bo}}\ \as{V}}
\row{5}{3}{(choice*)}{{\color{lvlC} \mathbf{t}_4 \triangleq }\ \vdash \as{W}\ \colcode{\1 + \bo}\ \as{W\sqcup W^+}}
\row{6}{3}{(inv-fix)}{{\color{lvlC} \mathbf{t}_5 \triangleq }\vdash \as{V = W\ \nabla\ (W\sqcup W^+)} \ \colcode{\fix{\bo}} \ \as{V}}
\row{7}{4}{(cons)}{{\color{lvlD} \mathbf{t}_6 \triangleq \ }\vdash \as{V}\ \colcode{\bo}\ \as{V}}
\row{8}{5}{(join)}{{\color{lvlE} \mathbf{t}_7 \triangleq \ } \vdash \as{V} \ \colcode{\bo}\ \as{W^+}}
\row{9}{6}{(seq)+(basic)}{{\color{lvlF} \mathbf{t}_8 \triangleq }\ \vdash \as{C_- \triangleq V \sqcap \{ x \leq 0\}} \ \colcode{\bo}\ \as{C'_-}}
\row{10}{6}{(seq)+(basic)}{{\color{lvlF} \mathbf{t}_9 \triangleq }\ \vdash \as{C_+ \triangleq V \sqcap \{ x \geq 0\}} \ \colcode{\bo}\ \as{C'_+}}
\row{11}{1}{(join*)}{{\color{lvlA} \mathbf{t}_7 \triangleq \ } \vdash \as{V} \ \colcode{\bo}\ \as{W^+}}
\row{12}{1}{}{U' \triangleq O_0 \sqcup W^+}
\end{tikzpicture}%
  }
\caption{More precise octagon analysis in $\APPL$.}\label{fig:octb}
\end{subfigure}
\caption{$\APPL$ derivations in the abstract domain $\Oct$.}
\label{fig:linear-derivation-oct1}
\end{figure}

\paragraph{Standard octagon analysis in {\rm $\APPL$}.}
The derivation in Fig.~\ref{fig:octa} has the shape of the interval case
in Fig.~\ref{fig:linear-derivation1}: two nested (widen) steps closed by (inv-fix),
then a (narrow). Thus, we emphasize only the step where relationality matters. From
$O_0$, (choice), (basic), (seq) analyse the body exactly, giving
${\color{lvlB} \mathbf{s}_2 \triangleq }\ \vdash \as{O_0}\ \colcode{\1 + \bo}\ \as{O_0 \sqcup O_1}$ with
$O_1 = \{\,x={-}3,\ y=3,\ z=16\,\}$ (here, $-^\sharp([{-}4]*^{\sharp}[4])=[16]$ is exact). Widening
against the input gives the first limit
$W \triangleq O_0\ \nabla\ (O_0\sqcup O_1) = \{\,x\ge{-}4,\ x+y=0,\ z\ge x+4\,\}$, where the closure
propagates  $y\le 4$ from $x\ge{-}4$ through $x+y=0$, and $z\ge 0$ from $x\ge{-}4$ through $z\ge x+4$; note that the standard widening~\cite{Mine06} keeps $z\ge x+4$, as it holds in both $O_0$ and $O_1$.
As the octagon $W$ is not an invariant, a second widening re-analyses the body. The
guard yields the octagon $W_g = \{\,{-}4\le x\le 3,\ x+y=0,\ z\ge x+4\,\}$, therefore $x\in[{-}4,3]$,
$y\in[{-}3,4]$. This is the only imprecise step: the interval fallback gives
$-^\sharp([{-}4,3]*^{\sharp}[{-}3,4])=[{-}12,16]$. Hence:
\begin{center}
${\color{lvlC} \mathbf{s}_4 \triangleq }\ \vdash \as{W}\ \colcode{\1 + \bo}\ \as{W\sqcup U}, \quad\text{where}\;
U \triangleq \sem{\bo} W = \{\,{-}3\le x\le 4,\ x+y=0,\ {-}12\le z\le 16\,\},$
\end{center}
with $U\not\sqsubseteq W$. The spurious $z={-}12={-}({-}4)\cdot({-}3)$ pairs $x={-}4$ with
$y={-}3$, excluded by the octagonal relation $x+y=0$: the fallback considers $x$ and $y$ as independent, discarding exactly
the relation the octagon carries. From here the derivation matches the interval case. The
second widening drops $z\ge x+4$ and $z\ge0$, yielding $W'\triangleq\{\,x\ge{-}4,\ x+y=0\,\}$ (closed by
(cons) and (inv-fix)), and (narrow) gives
$\mathbf{s}_0 \triangleq \: \vdash \as{O_0}\ \colcode{\fix{\bo}}\ \as{\Ostd \triangleq W'\,\Delta\,U'}$,
with $\Ostd = \{\,{-}4\le x\le 4,\ x+y=0,\ {-}12\le z\le 16\,\}$.
Narrowing cannot recover the sign: from any $x\in[{-}4,3]$ 
the fallback re-evaluates
the product  $-^\sharp([x]*^{\sharp}[y])$ to $[{-}12,16]$, so every descending step re-generates $z\ge -12$. The bounds of $x$ and $y$ and the
relation $x+y=0$ are optimal, while $z$ carries the spurious negative values in $[-12,-1]$.

\paragraph{The {\rm (join)} rule exploits the relational information.}
We fix the imprecise step in {\color{lvlC} $\mathbf{s}_4$} through the (join) rule, splitting at $x=0$ the octagon $V \triangleq \{\,x\ge{-}4,\ x+y=0,\ z\ge 0\,\}\sqsupseteq W$ as follows:
\begin{center}
$C_{-} \triangleq V \sqcap \{x\le 0\} = \{\,{-}4\le x\le 0, x+y=0, z\ge 0\}, \;\;\;\;\:
C_{+} \triangleq V \sqcap \{x\ge 0\} = \{\,x \geq 0, x+y=0, z\ge 0\,\}.$
\end{center}
The cover $\{C_{-},C_{+}\}$ is dense, since
$\gamma_{\Oct}(C_{-})\cup\gamma_{\Oct}(C_{+})=\gamma_{\Oct}(C_{-}\sqcup C_{+})=\gamma_{\Oct}(V)$.
Let us remark that splitting an octagon along the half-spaces of a constraint always yields a dense cover (Proposition~\ref{pr:template-split}), so  %
dense covers abound in $\Oct$. 
Within each component $C_{-}$ and $C_{+}$, the sign of $x$ is propagated to $y$  through the constraint
$x+y=0$, so that we get, after the filtering through the guard, $y\in[0,4]$ in $C_{-}$ and $y\in[{-}3,0]$ in $C_{+} \cap \{x \leq 3\}$. Both factors are
thus sign-definite, and the product is sign-exact:
\begin{center}
$-^\sharp([{-}4,0]*^{\sharp}[0,4])=[0,16]\ \text{ on }C_{-}\, ,
\qquad -^\sharp([0,3]*^{\sharp}[{-}3,0])=[0,9]\ \text{ on }C_{+}\cap\{x\le3\}\, .$
\end{center}
In the derivation in Fig.~\ref{fig:octb}, the (join) rule combines the sign-exact branches into
$W^+ \triangleq \sem{\bo}C_+ \sqcup \sem{\bo}C_- = \{\,{-}3\le x\le 4,\ x+y=0,\ 0\le z\le 16\,\}\sqsubseteq V$.
Since $W\sqsubseteq V$, rule (cons) yields $\vdash\as{W}\,\bo\,\as{W^+}$, which replaces the imprecise $U$ of $\mathbf{s}_4$: the second widening now keeps $z\ge 0$ and yields $W\,\nabla\,(W\sqcup W^+) = V$, which (cons) and (inv-fix) prove invariant.
An application of the (narrow) rule provides the invariant
$\Onew \triangleq V\,\Delta\,(O_0\sqcup W^+) = \{\,{-}4\le x\le4,\ x+y=0,\ 0\le z\le16\,\}$.

Thus, a single (join) certifies the sign of $z$ and refines the octagon $\Ostd$ to $\Onew$. At loop
exit we have that $\Onew\sqcap\{x\ge4\}=\{\,x=4,\ y={-}4,\ 0\le z\le16\,\}$, whereas $\Ostd$ still
allows $z\in[{-}12,16]$. Let us stress that this precision gain is inherently relational: the same split is useless in the
interval domain, where without $x+y=0$ the components $x\in[{-}4,0]$ and $x\in[0,3]$ leave
$y\in[{-}4,4]$ and the products evaluate to $[{-}16,16]$ and $[{-}12,12]$. It arises from
the interplay between the relation inferred by the octagon domain and the (join) rule that exploits
it locally. 
Let us observe that $\Onew$ is precise on all variable bounds and on the relation $x+y=0$,
though not the best octagonal
invariant, which also satisfies, e.g., $z+x\le13$: finer dense covers would recover such bounds.

\subsection{Sublogic as an Abstraction}
\label{sec:abs-sublogic}

As noted in \S~\ref{ss:logics}, limiting the assertion language to a proper subset
$\L \subsetneq D$ may cost relative completeness, but never soundness, as every derivable
triple is still a valid over-approximation. Thus, in our approach, restricting $\L$ is itself a form of
abstraction. We consider the case where $\L$ is an abstract domain, i.e.\ a Galois
insertion $\langle \alpha,\gamma \rangle: D=C \rightleftarrows A$
with $\L = A$. By \S~\ref{ss:ai-interpretation},
$(A, \sqsubseteq_A, \alpha(\base))$ is an interpretation domain, giving an instance of
$\APPL$ over $A$ that we denote by $\APPL_A$.

The natural question is how the deductive power of the concrete system $\APPL$ with
$\L = A$ compares to that of the abstract system $\APPL_A$. This is of practical interest,
since $\APPL_A$ works with sound and typically effective over-approximations of the
operators of $C$. As discussed in the introduction, this is precisely the setting of~\citet{CCLB:AIFR},
with $D$ the store domain and $A$ an abstraction. Note that the comparison does \emph{not}
follow from soundness and completeness of $\APPL$
(Theorems~\ref{prop:sound},~\ref{prop:weak-rel-compl}) together with
Proposition~\ref{pr:ai-sound}, because $\L = A$ is, in general, not expressive, according to
Definition~\ref{de:expressive-language}, unless $A$ is disjunctive. 

One direction is immediate: every abstract derivation lifts to a concrete one.

\begin{propositionrep}[$\APPL_A$ soundness]\label{prop:sound-abstract}
  Let $(C, \sqsubseteq, \base)$ be an interpretation domain and
  $\langle \alpha,\gamma \rangle: C \rightleftarrows A$ a Galois insertion, with $a, b \in A$, $\r\in \Reg$. If\/
  $\vdash_A \as{a} \ \r \ \as{b}$ is derivable in {\rm $\APPL_A$}, then
  $\vdash \as{\gamma(a)} \ \r \ \as{\gamma(b)}$ is derivable in {\rm $\APPL$}.
\end{propositionrep}

\begin{proof}
  Immediate consequence of the fact that all abstract rules, once
  abstract properties are replaced by their concretisation, are
  instances of the corresponding concrete rules. For instance, the
  abstract rule (choice) requires in the premises
  $b_1 \sqcup^A b_2 \sqsubseteq_A b$, i.e.,
  $\alpha(\gamma(b_1) \sqcup \gamma(b_2)) \sqsubseteq_A b$. Thus, by
  adjointness
  \begin{center}
    $\gamma(b_1) \sqcup \gamma(b_2) \sqsubseteq  \gamma(b)$.
  \end{center}
  which is a
  valid premise for (choice) in the concrete.

  In the case of rule (join), we exploit Lemma~\ref{le:base-transf}~\ref{le:base-transf:density} for deducing that the premise
  $\set{a_i \mid i \in I}$ being dense in $A$ implies that
  $\set{\gamma(a_i) \mid i \in I}$ is dense in $C$.
  \hfill
\end{proof}

For the converse, assume the abstract semantics of elementary commands
is given by their BCAs. Then, the abstract rules match the concrete
ones on abstract properties, with one possible exception: the (join)
rule. Its abstract premise requires $\{a_i \mid i \in I\}$ to be dense
\emph{in $A$}. By
Lemma~\ref{le:base-transf}~\ref{le:base-transf:density}, this implies
density of $\{\gamma(a_i)\}$ in $C$, but the converse, in general, is
not true.  Therefore, in general, the abstract premise is  stronger, and $\APPL_A$ could a
priori be weaker. The two coincide when the abstraction $\alpha$ \emph{reflects density}, defined as, for
every $X \subseteq A$, density of $\gamma(X)$ in $C$ implies density
of $X$ in $A$. Interestingly, this additional requirement is unnecessary when $C$ is completely meet-distributive.

\begin{propositionrep}[Equivalence of $\APPL$ and $\APPL_A$]\label{prop:equivalence}
  Let $(C, \sqsubseteq, \base)$ be an interpretation domain and
  $\langle \alpha,\gamma \rangle: C \rightleftarrows A$ a Galois
  insertion, and let the semantics of
  the elementary commands be given by their BCAs.
  Assume that $\alpha$
  reflects density or that $C$ is completely meet-distributive.
  For all $a, b \in A$, $\r\in \Reg$,
  if
  $\vdash \as{\gamma(a)} \ \r \ \as{\gamma(b)}$ is derivable in {\rm
    $\APPL$} then $\vdash_A \as{a} \ \r \ \as{b}$ is derivable in {\rm
    $\APPL_A$}.
\end{propositionrep}

\begin{proof}
  Under the hypotheses, all concrete rules, if used on (concretisations of) abstract elements can be seen as instances of the abstract ones. For instance, the
  concrete rule (choice) would be
  \begin{center}
      $\inference[(choice)]{\vdash \as{\gamma(a)} \ \r_1 \ \as{\gamma(b_1)} & \vdash_A \as{\gamma(a)} \ \r_2 \ \as{\gamma(b_2)} & \gamma(b_1)\sqcup \gamma(b_2) \sqsubseteq \gamma(b)}
      {\vdash \as{\gamma(a)} \ \r_1 + \r_2 \ \as{\gamma(b)}}$
  \end{center}
  The request in
  $\gamma(b_1) \sqcup \gamma(b_2) \sqsubseteq  \gamma(b)$, by adjointness, implies
  \begin{center}
    $b_1 \sqcup^A b_2  = \alpha(\gamma(b_1) \sqcup \gamma(b_2)) \sqsubseteq b$.
  \end{center}
   which is a
   valid premise for (choice) in the abstract system.

   The only delicate case, where preservation of density or complete meet-distributivity of $C$ come into play is rule (join)
  \begin{center}
   $ \inference[(join)]{\set{\gamma(a_i)}_{i \in I} \text{dense} & \gamma(a) \sqsubseteq \bigsqcup_{i \in I} \gamma(a_i) & \forall i \in I.\: \vdash \as{\gamma(a_i)} \ \r \ \as{\gamma(b_i)}& \bigsqcup_{i \in I} \gamma(b_i) \sqsubseteq \gamma(b)}{\vdash \as{\gamma(a)}
     \ \r \ \as{\gamma(b)}}$
 \end{center}

 First assume that $\alpha$ reflects density. In this case
 $\set{a_i}_{i \in I}$ is dense in $A$. From
 $\gamma(a) \sqsubseteq \bigsqcup_{i \in I} \gamma(a_i)$, exploiting
 the additivity of $\alpha$, we obtain
 \begin{center}
   $a = \alpha(\gamma(a)) \sqsubseteq_A \alpha(\bigsqcup_{i \in I}
   \gamma(a_i)) = \bigsqcup^A_{i \in I} \alpha(\gamma(a_i)) =
   \bigsqcup^A_{i \in I} a_i$.
 \end{center}
 Similarly, from
 $\bigsqcup_{i \in I} \gamma(b_i) \sqsubseteq \gamma(b)$ one derives
 $\bigsqcup^A_{i \in I} b_i \sqsubseteq b$. Hence the rule (join) can
 be applied in the abstract system
   \begin{center}
   $ \inference[(join)]{\set{a_i}_{i \in I} \text{dense} & a \sqsubseteq \bigsqcup_{i \in I}^A a_i & \forall i \in I.\: \vdash \as{a_i} \ \r \ \as{b_i}& \bigsqcup_{i \in I}^A b_i \sqsubseteq b}{\vdash \as{a}
     \ \r \ \as{b}}$
 \end{center}

 \bigskip

 Assume, instead, that $C$ is completely meet distributive. Define
 $a_i' = a \sqcap a_i$ for $i \in I$.

 \noindent Observe that $\set{\gamma(a_i')}_{i \in I}$ is dense in $C$.  In fact, let
 $p \in \base$ and assume that
 $p \sqsubseteq \bigsqcup_{i \in I} \gamma(a_i')$. Then
   \begin{align*}
     p & \sqsubseteq \bigsqcup_{i \in I} \gamma(a_i')\\
       & = \bigsqcup_{i \in I} \gamma(a \sqcap a_i)\\
       & = \bigsqcup_{i \in I}(\gamma(a) \sqcap \gamma(a_i))
       & \mbox{[$\gamma$ right-adjoint]}\\
       &= \gamma(a) \sqcap \bigsqcup_{i \in I} \gamma(a_i)
       & \mbox{[$C$ completely meet-distributive]}
   \end{align*}
   Then $p \sqsubseteq \gamma(a)$ and
   $p \sqsubseteq \bigsqcup_{i \in I} \gamma(a_i)$. From the latter,
   recalling that $\set{\gamma(a_i)}_{i \in I}$ is dense, we deduce
   the existence of $i \in I$ such that $p \sqsubseteq
   \gamma(a_i)$. Therefore
   $p \sqsubseteq \gamma(a) \sqcap \gamma(a_i) = \gamma(a \sqcap a_i)
   = \gamma(a_i')$, as desired.

   Moreover, since $a_i' \sqsubseteq a_i$ and thus
   $\gamma(a_i') \sqsubseteq \gamma(a_i)$, from
   $\vdash \as{\gamma(a_i)} \ \r \ \as{\gamma(b_i)}$, using rule
   (cons), one can deduce
   $\vdash \as{\gamma(a_i')} \ \r \ \as{\gamma(b_i)}$.
   
   Therefore we can transform the application of the (join) rule above as

   \begin{center}
   $ \inference[(join)]{\set{\gamma(a_i')}_{i \in I} \text{dense} & \gamma(a) = \bigsqcup_{i \in I} \gamma(a_i') & \forall i \in I.\: \vdash \as{\gamma(a_i')} \ \r \ \as{\gamma(b_i)}& \bigsqcup_{i \in I} \gamma(b_i) \sqsubseteq \gamma(b)}{\vdash \as{\gamma(a)}
     \ \r \ \as{\gamma(b)}}$
 \end{center}

 We finally observe that $\set{a_i'}_{i \in I}$ is dense in $A$ with
 respect to the abstract basis $\alpha(\base)$. In fact, preliminarily
 notice that the premise
 $\gamma(a) = \bigsqcup_{i \in I} \gamma(a_i')$ implies
 \begin{equation}
   \label{eq:sub}
   \textstyle
   a = \alpha(\gamma(a)) = \alpha \Bigr (\bigsqcup_{i \in I} \gamma(a_i') \Bigl ) =
   \bigsqcup_{i \in I}^A\alpha(\gamma(a_i')) = \bigsqcup_{i \in I}^A
   a_i'
\end{equation}

Now, let $p \in \base$ and assume
$\alpha(p) \sqsubseteq \bigsqcup_{i \in I}^A a_i'$. Then
\begin{align*}
  p
  & \sqsubseteq \gamma\Bigl (\bigsqcup_{i \in I} a_i' \Bigr )
  & \mbox{[by adjointness]}\\
  & = \gamma(a)
  & \mbox{[by \eqref{eq:sub} above]}\\
  & = \bigsqcup_{i \in I} \gamma(a_i') 
  & \mbox{[premise of the rule (join)]}\\
\end{align*}
Since $\set{\gamma(a_i')}_{i \in I}$ is dense, we deduce the existence of $i \in I$ such that $p \sqsubseteq \gamma(a_i')$. Then by adjointness $\alpha(p) \sqsubseteq a_i'$, as desired. Thus, again, rule (join) can be applied in the abstract system.
 \hfill
\end{proof}

Observe that when abstracting the concrete collecting semantics,
the lattice $(C, \leq)$ is the powerset of states $(\pow{\Sigma}, \subseteq)$,
which is completely meet-distributive. Hence the above result applies, ensuring that $\APPL_A$ and $\APPL$ with
$\L = A$ have the same deductive power: this is essentially the
setting of algebraic Hoare logic by~\citet{CCLB:AIFR}.

\section{Non-determinism: Beyond Lattice Join}
\label{sec:nondet}

We generalise our framework by interpreting non-deterministic
composition with a monoidal operator $\oplus$, which may differ from
the lattice join and need not be idempotent. This uniformly captures
Hoare logic, incorrectness logic, and hyperproperty logics, each
possibly combined with abstraction, in one foundational
setting. Notably, the (join) rule, which is essential for soundly
handling non-additive \emph{abstract} semantics, plays the same
central role in hyperlogics, for an analogous reason: there too the
monoidal operation for non-determinism does not preserve
additivity. In this generalisation the order need not be inflationary
for $\oplus$, meaning that $x \sqsubseteq x \oplus y$ may fail, which
blocks the direct use of fixpoint semantics for iteration. %
We therefore restrict to the treatment of Kleene star iteration.
Hence, rules (widen) and (narrow) are unavailable over interpretation monoids.
This issue can be addressed by employing an approximation preorder weaker than the logical order. 
We defer the development to future work and elaborate further in \S~\ref{sec:conclusion}.

\subsection{Interpretation Monoids and $\APPL$}
\label{sec:intmonsem}

We introduce interpretation monoids, which generalise interpretation domains %
by equipping them with a monoidal operator, required to be infinitary  for interpreting the Kleene star.

A \emph{complete monoid}~\cite{Krob87,Kar04} is a structure
$(M, \oplus)$, where $\oplus$ is an
infinitary sum operation that maps each indexed family $\{m_i\}_{i \in I}$ of elements of $M$,
to 
$\bigoplus_{i \in I} m_i \in M$ such that: 
\begin{enumerate}[{\rm (i)}]
\item
$\bigoplus_{i \in \set{j}} m_i = m_j$ and
\item  \label{def:ass-mon}
for any partition
$(I_j)_{j\in J}$ of the index set $I$, it holds
  $   \textstyle \bigoplus_{i \in I} m_i = \bigoplus_{j \in J} (\bigoplus_{i \in I_j} m_i)\, $.
\end{enumerate}
Property (ii) is an infinitary form of associativity,  
which entails infinitary commutativity 
and the existence of a neutral element $0_{\oplus} \triangleq \bigoplus_{i \in \varnothing} m_i \in M$. 
We will %
write $m_1 \oplus \cdots \oplus m_n$ instead of $\bigoplus_{j \in [1,n]} m_j$.
Given $S\subseteq M$, we denote by $\mon{S}$ the least complete
submonoid including $S$. %
Every complete lattice is a complete monoid with respect to the join, with neutral element $\bot$, and, dually, also with respect to the meet, with  neutral element $\top$.

\begin{example}[Hypermonoid]
  \label{ex:hyper-monoid}
  Let $(C, \leq)$ be a complete lattice. Then $(\pow{C}, \otimes)$ is
  a complete monoid, where $\otimes$ is the hyper-assertion operation
  of~\citet[Def.~6--7]{DM:HHL} defined as
  \( \bigotimes_{i \in I} h_i \triangleq \set{ \bigvee_{i \in I} x_i
    \mid \forall i \in I.\, x_i \in h_i } \). Intuitively, if $h_1$
  and $h_2$ are hyperproperties (sets of properties) satisfied by two commands
  $\r_1$ and $\r_2$ from the same precondition, then $h_1 \otimes h_2$ is a hyperproperty
  satisfied by $\r_1 + \r_2$,  whose elements pair every outcome of $\r_1$ with every outcome of $\r_2$.
    \hfill$\lozenge$
\end{example}

Generalising Definition~\ref{de:semantics}, we interpret commands first on the basis
$\base$ by monotone maps $\bsem{\r} : \base \to \mon{\base}$, and then extend them to the
whole domain $D$ by join extension. In particular, evaluating a sequential composition
$\r_1;\r_2$ at a basis element $b$ requires applying the join extension of $\bsem{\r_2}$ to
$\bsem{\r_1}b \in \mon{\base}$. Since $\bsem{\r_2}$ is monotone, this join can be computed
over any set of basis elements covering $\bsem{\r_1}b$. To guarantee that the result remains
in $\mon{\base}$, the latter must be closed under joins of sets whose size is bounded by the
number of basis elements needed to cover an element of $\mon{\base}$. This motivates the
notion of \emph{weight} of an element, which roughly represents the number of basis elements
needed to cover it. Formally, the weight of $d \in D$ w.r.t.\ $\base$ is
$\w{d} \triangleq \min \{ |Y| \mid Y \subseteq \ebase{d},\, \ebase{d} \subseteq \down{Y} \}$.
The notation extends to sets $X \subseteq D$ by letting
$\w{X} \triangleq \sup\set{ \w{d} \mid d \in X}$.

\begin{definition}[Interpretation monoid]
  \label{de:interpretation-monoid}
  An \emph{interpretation monoid} is a tuple $(D, \sqsubseteq, \oplus, \base)$ such that:
  (i)~$(D, \sqsubseteq, \base)$ is an interpretation domain;
  (ii)~$(D, \oplus)$ is a complete monoid with $\oplus$ monotone; and
  (iii)~$\mon{\base}$ is closed under non-empty $\kappa_D$-bounded joins, where
  $\kappa_D \triangleq \w{\mon{\base}}$.
  \hfill$\lozenge$
\end{definition}

In the instances of this paper, condition~(iii) is immediate: for $\oplus = \sqcup$ and in
the incorrectness instance, $\mon{\base} = D$ is closed under all joins, and in the
hypermonoid we have $\kappa_D = 1$.

\begin{example}[Hypermonoid, continued]\label{ex:hyper-interpretation}
  Given a complete lattice $(C, \leq)$, the \emph{hypermonoid}
  $(\pow{C}, \subseteq, \otimes, \Ir{\pow{C}})$, with the irreducible basis
  $\Ir{\pow{C}} = \set{\{c\} \mid c \in C}$, is an interpretation monoid. In particular,
  $\kappa_{\pow{C}} = 1$, since ${\Ir{\pow{C}}}^{\otimes} = \Ir{\pow{C}}$ and every
  $h \in \Ir{\pow{C}}$ is a singleton with $\w{h} = 1$, so that condition~(iii) of
  Definition~\ref{de:interpretation-monoid} is trivial. Observe that $\varnothing$ and
  $\{\bot_C\}$ play two distinct roles: $\varnothing$ is the bottom element of the
  lattice, whereas $\{\bot_C\}$ is the neutral element of $\otimes$.
  \hfill$\lozenge$
\end{example}

\begin{example}[Interpretation monoid for incorrectness logic]\label{ex:incorrectness-interpretation}
  Consider $(\pow{\Sigma}, \supseteq, \cup, \pow{\Sigma})$, i.e.\ the powerset ordered by
  \emph{reverse} inclusion, with $\oplus = \cup$ and basis $\pow{\Sigma}$. Since
  ${\sqsubseteq} = {\supseteq}$, the lattice bottom is $\Sigma$, the top is
  $\varnothing$, the join is $\cap$ and the meet is $\cup$. Thus $\oplus = \cup$, which is
  monotone w.r.t.\ $\supseteq$ and has neutral element $\varnothing$, is the lattice
  \emph{meet}, not the join. Condition~(iii) of Definition~\ref{de:interpretation-monoid}
  holds trivially, since $\mon{\base} = \base = \pow{\Sigma}$ is closed under arbitrary
  intersections and $d \in \ebase{d}$ for every $d$.
  \hfill$\lozenge$
\end{example}

\begin{definition}[Basis and Full Semantics]\label{de:semantics-monoid}
  Let \((D, \sqsubseteq, \oplus, \base)\) be an interpretation monoid, with a monotone
  semantics $\bsem{\e} : \base \to \mon{\base}$ for elementary commands, where
  $\bsem[\base]{\1} \triangleq \lambda x . x$ and divergence is the neutral element
  $\0_{\oplus}$. For $\r \in \Reg$ and $b \in \base$, the \emph{basis semantics} $\bsem{\r}: \base \to \mon{\base}$  is
  \begin{align*}
    \bsem{\r_1;\r_2} b \triangleq \add{\bsem{\r_2}}(\bsem{\r_1}b)\, , \qquad
    \bsem{\r_1 + \r_2} b \triangleq \bsem{\r_1} b \oplus \bsem{\r_2} b\, , \qquad
   \textstyle \bsem{\r^*} b \triangleq \bigoplus_{i \geq 0} \bsem{\r^i}b\, .
  \end{align*}
  The \emph{full semantics} $\sem{\r} : D \to D$ is defined as the join extension $\sem{\r} \triangleq \add{\bsem{\r}}$ of the basis semantics. \hfill $\lozenge$
\end{definition}

Every interpretation domain $(D, \sqsubseteq, \base)$ is an interpretation monoid with
$\oplus = \sqcup$. In particular $\mon{\base} = D$, hence Definition~\ref{de:interpretation-monoid}~(iii)
holds, and the two semantics (Definitions~\ref{de:semantics-monoid}
and~\ref{de:semantics}) coincide.

\begin{toappendix}

\begin{lemmarep}[Basic properties of $\sem{\cdot}$]\label{le:basic-prop-sem-mon}
  For all $\r\in \Reg$:
  \begin{enumerate}[{\rm (B$_1$)}]
  \item\label{le:basic-prop-sem-mon:well-def} $\bsem{\r}$ is well-defined, i.e., 
    for all $b \in \base$, $\bsem{\r}b \in \mon{\base}$.
  \item\label{le:basic-prop-sem-mon:monot}
    The full semantics $\sem{\r}$ is monotone w.r.t. \/ $\sqsubseteq$.
  \item\label{le:basic-prop-sem-mon:distrib} If $D$ is a completely
meet-distributive lattice and the basis is $\Ir{D}$ then $\sem{\r}$ is additive. 
  \end{enumerate}
\end{lemmarep}

\begin{proof}
  \ref{le:basic-prop-sem-mon:well-def} By structural induction on the command syntax. Let
  $b \in \base$.
  \begin{itemize}
  \item
    If $\r = \e$ we have that $\bsem{\e}b \in \mon{\base}$, since by
    assumption $\bsem{\e}: \base \to \mon{\base}$.

  \item
    If $\r = \r_1 + \r_2$ we have
    $\bsem{\r_1 + \r_2}b = \bsem{\r_1}b \oplus \bsem{\r_2}b \in
    \mon{\base}$ because by inductive hypothesis
    $\bsem{\r_i}b \in \mon{\base}$ for $i = 1, 2$ and
    $\mon{\base}$ is a submonoid.

  \item
    If $\r = \r_1; \r_2$ we have that
    \begin{equation}\label{eq:composition1}
      \bsem{\r_1 ; \r_2}b
      = \add{\bsem{\r_2}}(\bsem{\r_1}b) 
      = \bigsqcup_{b' \in \ebase{\bsem{\r_1}b}} \bsem{\r_2}b'
    \end{equation}
    Now, by inductive hypothesis, $\bsem{\r_1}b \in \mon{\base}$
    and thus, by the definition of interpretation monoid, there is
    $Y \subseteq \base$ such that
    $\down{Y} = \ebase{\bsem{\r_1}b}$ and $|Y| \leq \kappa_D$.

    If $Y = \varnothing$, then $\bsem{\r_1}b = \bot$, hence
    $\bot \in \mon{\base}$. Therefore, by \eqref{eq:composition1}
    $\bsem{\r_1 ; \r_2}b = \bot \in \mon{\base}$, as desired.

    Let us focus on the case $|Y|>0$.
    By monotonicity of $\bsem{\r_2}$ we deduce that
    \begin{equation}\label{eq:composition2}
      \down{\set{\sem{\r_2}b' \mid b' \in Y}} = \down{\set{\sem{\r_2}b'
          \mid b' \in\ebase{\bsem{\r_1}b}}}.
    \end{equation}
    Moreover, again by
    inductive hypothesis, $\bsem{\r_2}b' \in \mon{\base}$ for all
    $b' \in Y \subseteq \ebase{\bsem{\r_1}b}$. Finally
    $|\set{\sem{\r_2}b' \mid b' \in Y}| \leq |Y| \leq \kappa_D$, hence
    for the fact that $\mon{\base}$ is
    closed by non-empty $\kappa_D$ joins, we deduce
    \begin{equation}\label{eq:composition3}
      \sqcup \set{\sem{\r_2}b' \mid b' \in Y} \in \mon{\base}
    \end{equation}
    Putting things together we obtain
    \begin{align*}
      &
      \bsem{\r_1 ; \r_2}b =
      & \mbox{[by \eqref{eq:composition1}]}\\
      & = \sqcup{\set{\sem{\r_2}b' \mid b' \in\ebase{\bsem{\r_1}b}}}
      & \mbox{[by \eqref{eq:composition2}]}\\
      &= \sqcup{\set{\sem{\r_2}b' \mid b' \in Y}} \in \mon{\base}
      & \mbox{[by \eqref{eq:composition3}]}
    \end{align*}

  \item
    If $\r = \r^*$ we have
    $\bsem{\r^*}b = \bigoplus_{i \geq 0} \bsem{\r^i}b$.  By inductive
    hypothesis $\bsem{\r}b \in \mon{\base}$.  Moreover, for all
    $i \geq 0$, $\bsem{\r^i}b \in \mon{\base}$. In fact, by straightwforward induction of $i$, we have: 
    \begin{itemize}

    \item
      $\bsem{\r^0}b = b \in \base \subseteq \mon{\base}$ by definition of $\mon{(\cdot)}$.
      \item
        $\bsem{\r^{i+1}}b = \bsem{\r^i; \r}b =
        \add{\bsem{\r}}(\bsem{\r^i}b) = \bigsqcup_{b'\in
          \ebase{\bsem{\r^i}b}} \bsem{\r}b'$ and we conclude as in the
        case $\r = \r_1; \r_2$ case since by inductive hypothesis
        $\bsem{\r^i}b \in \mon{\base}$ and
        $\bsem{\r}b' \in \mon{\base}$.
        Then, $\bigoplus_{i \geq 0} \bsem{\r^i}b \in \mon{\base}$ by
        definition of $\mon{(\cdot)}$.
    \end{itemize}
  \end{itemize}

\bigskip

\noindent
\ref{le:basic-prop-sem-mon:monot} We show that if $\bsem{\e}: \base \to \mon{\base}$ 
is monotone w.r.t.  $\sqsubseteq$ then 
$\bsem{\r}$ is monotone w.r.t.  $\sqsubseteq$. The result for $\sem{\r}$ follows 
from Lemma~\ref{le:additive-sq}
recalling that $\sem{\r} = \add{\bsem{\r}}$.

We proceed by structural induction on the command syntax and show that for all $b, b' \in \base$,
if $b \sqsubseteq b'$, then $\bsem{\r}b \sqsubseteq \bsem{\r}b'$. 
\begin{itemize}
  \item If $\r = \e$, monotonicity is by assumption. %

  \item If $\r = \r_1 + \r_2$, then by inductive hypothesis we have
    $\bsem{\r_1}$ and $\bsem{\r_2}$ are monotone. Hence
    \begin{align*}
      \bsem{\r_1 + \r_2}b %
      & = \bsem{\r_1}b \oplus \bsem{\r_2}b %
      & %
      [\mbox{definition of } \bsem{\r_1 + \r_2}] \\
      & \sqsubseteq \bsem{\r_1}b' \oplus \bsem{\r_2}b' %
      & %
      [\mbox{inductive hypothesis and monotonicity of } \oplus] \\
      & = \bsem{\r_1 + \r_2}b' %
      & %
      [\mbox{definition of } \bsem{\r_1 + \r_2}]
    \end{align*}

  \item If $\r = \r_1; \r_2$, then by inductive hypothesis
    $\bsem{\r_1}$ and $\bsem{\r_2}$ are monotone. Hence, by
    Lemma~\ref{le:additive-sq}, $\add{\bsem{\r_2}}$ is monotone and thus
    $\bsem{\r_1; \r_2} = \add{\bsem{\r_2}} \circ \bsem{\r_1}$ is monotone because it is composition of monotone functions.
    
    \item If $\r = \r_1^*$
     by inductive hypothesis we have that $\bsem{\r_1}$ is monotone,
     hence $\bsem{\r_1^i}$ is monotone for all $i \in \N$. In fact, by straightforward induction on $i$ we have:
     \begin{itemize}
     \item $\bsem{\r_1^0} = \bsem{\1}$ is monotone,
     \item  $\bsem{\r_1^{i+1}} = \add{\bsem{\r_1}} \circ \bsem{\r_1^i}$
     is monotone because it is composition of monotone functions:
       $\bsem{\r_1^i}$ is monotone by the inductive hypothesis and
       $\add{\bsem{\r_1}}$ is monotone by the inductive hypothesis and
       Lemma~\ref{le:additive-sq}.
     \end{itemize}
     Since $\bigoplus$ is monotone we conclude.
   \end{itemize}
   
\bigskip

\noindent
\ref{le:basic-prop-sem-mon:distrib}
Observe that, in this setup, each
$X\subseteq D$
is dense. In fact, for
all $b\in \Ir{D}$, if $b \sqsubseteq \sqcup X$, then
\( b = b \sqcap \sqcup X = \sqcup \set{b \sqcap x \mid x \in X } \),
so that there exists $x \in X$ such that $b = b \sqcap x$, and, in
turn, $b \sqsubseteq x$.  Hence, by Lemma~\ref{le:add-dense}, the full semantics $\sem{\r}$ is
additive on $X$.
\hfill
\end{proof}

\begin{definition}[Additive combination of a family of functions]
  Given a family of functions $f_i : \lat \to \lat$, indexed by
$i \in I$, let us define their \emph{additive combination}  $\add{\bigoplus}_{i\in I} f_i : \lat \to \lat$ as follows:
\begin{equation*}
  \Big(\add{\bigoplus}_{i\in I} f_i \Big)(d) \triangleq \bigsqcup_{b \in \ebase{d}}
\Big(\bigoplus_{i\in I} f_i(b) \Big)
\tag*{$\lozenge$}
\end{equation*}
 \end{definition}

\begin{lemmarep}\label{le:additive-comb}
  For all $\r,\r_1,\r_2\in \Reg$ and $d \in \lat$, the following equalities hold:
  \begin{enumerate}[label=\rm(A\arabic*)] 
    \item\label{le:additive-comb:sum}
    $\sem{\r_1 + \r_2}d = (\sem{\r_1} \add{\oplus} \sem{\r_2})d$
    \item\label{le:additive-comb:star}
    $\sem{\r^*}d = (\add{\bigoplus}_{i \geq 0} \sem{\r^i})(d)$
  \end{enumerate}
\end{lemmarep}

\begin{proof}
\bigskip

\ref{le:additive-comb:sum}
For all  $d \in D$:
    \begin{align*}
      \sem{\r_1 + \r_2}d
      = \add{\bsem{\r_1 + \r_2}}d %
      = \bigsqcup_{b \in \ebase{d}} \left(\bsem{\r_1}b \oplus \bsem{\r_2}b\right) %
      = \bigsqcup_{b \in \ebase{d}} \left( \sem{\r_1}b \oplus \sem{\r_2}b \right) %
      = (\sem{\r_1} \add{\oplus} \sem{\r_2})d
    \end{align*}
\ref{le:additive-comb:star}
For all $d \in D$
    \begin{align*}
      \sem{\r^*}d 
      = \add{\bsem{\r^*}}d %
      = \bigsqcup_{b \in \ebase{d}}\left(\bigoplus_{i\geq 0} \bsem{\r^i}b\right) %
      = \left(\add{\bigoplus}_{i\geq 0} \sem{\r^i}\right)d
      \tag*{\qedhere}
    \end{align*}
\end{proof}
\end{toappendix}

Unfolding the Kleene star requires $\sem{\r}$ to satisfy the following condition: 
\begin{equation}\label{eq:subadd-oplus}
  \tag{$\mathrm{S}_\oplus$}
\text{for all } \{h_i\}_{i \in \N} \subseteq D,\;  \textstyle\bigoplus_{i \in \N} \sem{\r} h_i \sqsubseteq
  \sem{\r} \bigl( \bigoplus_{i \in \N} h_i \bigr).
\end{equation}
For interpretation domains, where $\oplus = \sqcup$ this is just monotonicity of the semantics. In the hyperproperty and
incorrectness instances,~\eqref{eq:subadd-oplus} holds with equality, but it may fail in general. 

\begin{lemmarep}[Properties of $\sem{\cdot}$]\label{le:kleene-axioms-monoid}
  For all $\r,\r_1,\r_2\in \Reg$ and $d \in \lat$, the properties below hold:
  \vspace*{-10pt}
  \begin{enumerate}[{\rm (1)}]
  \begin{multicols}{2}
  \item\label{le:kleene-axioms-monoid:seq}
    $\sem{\r_1;\r_2}d \sqsubseteq \sem{\r_2}(\sem{\r_1}d)$
  \item\label{le:kleene-axioms-monoid:sum}
    $\sem{\r_1 + \r_2}d \sqsubseteq (\sem{\r_1}d) \oplus (\sem{\r_2}d)$
  \item\label{le:kleene-axioms-monoid:star}
    $\sem{\r^*}d \sqsubseteq \bigoplus_{i \in \N}(\sem{\r^i} d)$
  \item\label{le:kleene-axioms-monoid:left-dist}
    $\sem{\r; (\r_1 + \r_2)} \sqsubseteq \sem{(\r; \r_1) + (\r; \r_2)}$
  \item\label{le:kleene-axioms-monoid:left-exp2}
    $\sem{\r^*} \sqsubseteq \sem{\1 + (\r^*; \r)}$ if $\sem{\r}$
    satisfies~\eqref{eq:subadd-oplus}
  \end{multicols}
  \end{enumerate}
\end{lemmarep}

\begin{proof}%

  \ref{le:kleene-axioms-monoid:seq}
  For all $d \in D$
  \begin{align*}
    \sem{\r_1; \r_2}d 
    &= \add{\bsem{\r_1; \r_2}}d \\
    & 
    = \bigsqcup_{b \in \ebase{d}} \bsem{\r_1; \r_2}b \\
    & 
    = \bigsqcup_{b \in \ebase{d}} \add{\bsem{\r_2}} (\bsem{\r_1} b) \\
    & 
    \sqsubseteq \add{\bsem{\r_2}} \Bigl(\bigsqcup_{b \in \ebase{d}}  \bsem{\r_1} b\Bigl) \\
    & 
    =  \add{\bsem{\r_2}} (\add{\bsem{\r_1}} d) \\
    &
    = \sem{\r_2} (\sem{\r_1} d) 
  \end{align*}

  \bigskip

\noindent
  \ref{le:kleene-axioms-monoid:sum} For all  $d \in D$:
    \begin{align*}
      \sem{\r_1 + \r_2}d
      & = \add{\bsem{\r_1 + \r_2}}d\\
      & = \bigsqcup_{b \in \ebase{d}} \Bigl(\bsem{\r_1}b \oplus \bsem{\r_2}b \Bigr)\\
      & \sqsubseteq   \Bigl ( \bigsqcup_{b \in \ebase{d}} \bsem{\r_1}b \Bigr ) 
      \oplus \Bigl (\bigsqcup_{b \in \ebase{d}} \bsem{\r_2}b \Bigr ) \\ 
      & = \sem{\r_1}d \oplus \sem{\r_2}d
    \end{align*}
    as
    $\bsem{\r_i}b \sqsubseteq \bigsqcup_{b \in \ebase{d}}
    \bsem{\r_i}b$ by monotonicity of $\bsem{\r_i}$ for $i = 1, 2$, 
    hence by monotonicity of $\oplus$, for all $b \in \ebase{d}$
    \begin{center}
      $\bsem{\r_1}b \oplus \bsem{\r_2}b \sqsubseteq \left ( 
        \bigsqcup_{b \in \ebase{d}} \bsem{\r_1}b \right ) \oplus 
        \left (\bigsqcup_{b \in \ebase{d}} \bsem{\r_2}b \right )$.
    \end{center}

    Since the inequality holds for all
      $b \in \ebase{d}$ it also holds for
    $\bigsqcup_{b \in \ebase{d}} (\bsem{\r_1}b \oplus \bsem{\r_2}b)$.
    
    \bigskip

\noindent
    \ref{le:kleene-axioms-monoid:star} For all $d \in D$
    \begin{align*}
       \sem{\r^*}d 
       = \add{\bsem{\r^*}}d %
      = \bigsqcup_{b \in \ebase{d}}\Biggl(\bigoplus_{i\geq 0} \bsem{\r^i}b\Biggr) %
      \sqsubseteq \bigoplus_{i \geq 0} \bigsqcup_{b \in \ebase{d}} \bsem{\r^i}b %
      = \bigoplus_{i \geq 0} \sem{\r^i}d
    \end{align*}
    and
    $\bsem{\r^i}b \sqsubseteq \bigsqcup_{b \in \ebase{d}} \bsem{\r^i}b
    = \sem{\r^i}d$ by monotonicity of $\bsem{\r^i}$ for all $i$ and
    for all $b \in \ebase{d}$. Then by monotonicity of
    $\bigoplus_{i \geq 0}$ we have
    $\bigoplus_{i\geq 0} \bsem{\r^i}b \sqsubseteq \bigoplus_{i \geq 0}
    \sem{\r^i}d$, and since it holds for all $b \in \ebase{d}$ it also
    holds for
    $\bigsqcup_{b \in \ebase{d}}(\bigoplus_{i\geq 0} \bsem{\r^i}b)$.
    
    \bigskip
  
\noindent
\ref{le:kleene-axioms-monoid:left-dist}
Let $b \in \base$ and $e \triangleq \bsem{\r}b$. Then
\begin{align*}
  \bsem{\r; (\r_1 + \r_2)}b
  & = \add{\bsem{\r_1 + \r_2}}(e)
    = \bigsqcup_{b' \in \ebase{e}} (\bsem{\r_1}b' \oplus \bsem{\r_2}b')\\
  & \sqsubseteq \Bigl(\bigsqcup_{b' \in \ebase{e}} \bsem{\r_1}b'\Bigr) \oplus
    \Bigl(\bigsqcup_{b' \in \ebase{e}} \bsem{\r_2}b'\Bigr)
  & \mbox{[monotonicity of $\oplus$]}\\
  & = \bsem{\r;\r_1}b \oplus \bsem{\r; \r_2}b = \bsem{(\r;\r_1) + (\r; \r_2)}b .
\end{align*}
If $e \in \base$, then $e$ is the maximum of $\ebase{e}$ and, by monotonicity of
$\bsem{\r_1}$, $\bsem{\r_2}$ and $\oplus$, both sides are equal to
$\bsem{\r_1}e \oplus \bsem{\r_2}e$.

\bigskip

\noindent
\ref{le:kleene-axioms-monoid:left-exp2}
Let $b \in \base$. Then
\begin{align*}
  \bsem{\r^*}b
  & = \bigoplus_{i \geq 0} \bsem{\r^i}b
    = b \oplus \bigoplus_{i \geq 0} \bsem{\r^{i}; \r}b
    = b \oplus \bigoplus_{i \geq 0} \sem{\r}(\bsem{\r^i}b)\\
  & \sqsubseteq b \oplus \sem{\r}\Bigl(\bigoplus_{i \geq 0} \bsem{\r^i}b\Bigr)
  & \mbox{[by~\eqref{eq:subadd-oplus} and monotonicity of $\oplus$]}\\
  & = \bsem{\1}b \oplus \bsem{\r^*;\r}b = \bsem{\1 + (\r^*; \r)}b .
  \tag*{\qedhere}
\end{align*}    
\end{proof}

Note that the inductive definitions now over-approximate not only sequential composition
but also choice and Kleene iteration. If the full semantics is additive,
then~\ref{le:kleene-axioms-monoid:seq} becomes an equality (this concerns additivity of
$\sem{\r_2}$, and is independent of $\oplus$). If, moreover, the pointwise extension of
$\oplus$ preserves additivity\footnote{Let $\{f_i\}_{i \in I}$ be a family of functions
$f_i : D \to D$. Define their pointwise combination $\bigoplus_{i \in I} f_i$ by
$(\bigoplus_{i \in I} f_i)(d) = \bigoplus_{i \in I} f_i(d)$. We say that $\oplus$
\emph{preserves additivity} if, whenever every $f_i$ is additive, the function
$\bigoplus_{i \in I} f_i$ is also additive.}, then
\ref{le:kleene-axioms-monoid:sum}--\ref{le:kleene-axioms-monoid:left-dist} become
equalities too, and the full semantics coincides with the inductive one
(cf.\ Proposition~\ref{pr:inductive}).

\paragraph{The logic {\rm $\APPL$} over interpretation monoids.}
$\APPL$ generalises to interpretation monoids by replacing $\sqcup$ with $\oplus$ in the
rules governing  non-deterministic choice.
The affected rules are those shown in Figure~\ref{fig:rules-oplus}. Since $\mathsf{fix}$ is no longer available, (inv-star) ceases to be derivable and therefore becomes a primitive rule; moreover, it requires an explicit side condition because $\oplus$ need not be idempotent. 
Rule (rec-star) requires $\sem[]{\r}$ to satisfy the condition~\eqref{eq:subadd-oplus}.
Over interpretation monoids, $\APPL$ is sound and, when $\L$ is expressive, relatively
complete (Theorem~\ref{prop:sound-oplus} and Proposition~\ref{de:expressive-language-oplus}).
 Validity (Definition~\ref{de:validity}) and expressiveness
(Definition~\ref{de:expressive-language}) are unchanged.

\begin{figure*}[t]
\hspace{-20pt}
\resizebox{1.02\textwidth}{!}{
    \begin{tabular}{c}\label{tab:hyperHoare}
          \vspace{10pt}
          $\inference[(choice)]{\vdash \as{h} \ \r_1 \ \as{k_1} & \vdash \as{h} \ \r_2 \ \as{k_2} & k_1\oplus k_2 \sqsubseteq k}
          {\vdash \as{h} \ \r_1 + \r_2 \ \as{k}}$
          \qquad
          $\inference[(iter)]{\forall i \in \N.\: \vdash \as{h_i} \ \r \ \as{h_{i+1}} & \bigoplus_{i \in \N} h_i \sqsubseteq k}{\vdash \as{h_0} \ \r^* \ \as{k}}$
          \\
            $\inference[(rec-star)]{\vdash \as{h} \ \r^* \ \as{k'} & \vdash \as{k'} \ \r \ \as{l} & h\oplus l\sqsubseteq k}{\vdash \as{h} \ \r^* \ \as{k}}$
          \qquad 
          $\inference[(inv-star)]{\vdash \as{h} \ \r \ {\as{h}} %
          & \bigoplus_{\N} h \sqsubseteq k}{\vdash \as{h} \ \r^* \ \as{k}}$ \\
    \end{tabular}
    }
    \caption{The modified rules of the logic $\APPL$ over the
      interpretation monoid $(D, \sqsubseteq, \oplus, \base)$. 
      By $\bigoplus_{\N} m$ we denote the sum of countably many copies of  $m$.}
\label{fig:rules-oplus}
\end{figure*}

\begin{toappendix}
\begin{theoremrep}[Soundness]\label{prop:sound-oplus}
    Let $h, k \in \L$ and $\r\in \Reg$. Then %
        $\vdash \as{h} \ \r \ \as{k} 
        \ \Rightarrow \ %
        \sem{\r}{h} \sqsubseteq k$. %
\end{theoremrep}

\begin{proof}
  By structural induction on the derivation. We consider the last rule applied in the derivation.

  \bigskip

  \noindent
  (basic) $\inference{\sem{\e}{h} \sqsubseteq k}{\vdash \as{h} \ \e \ \as{k}}$:
  we conclude by the side condition $\sem{\e}{h} \sqsubseteq k$.

  \bigskip \bigskip

  \noindent
  (seq) $\inference{\vdash \as{h} \ \r_1 \ \as{k'} & \vdash \as{k'} \ \r_2 \ \as{k}}
          {\vdash \as{h} \ \r_1; \r_2 \ \as{k}}$:
          \bigskip

          \noindent
          by inductive hypothesis we have $\sem{\r_1}h \sqsubseteq k'$ and $\sem{\r_2}k' \sqsubseteq k$, then:
        \begin{equation*}
          \sem{\r_1; \r_2}h \sqsubseteq \sem{\r_2}\left ( \sem{\r_1}h \right )
            \sqsubseteq \sem{\r_2}k'  \sqsubseteq k
        \end{equation*}
        where we used Lemma~\ref{le:kleene-axioms-monoid}~\ref{le:kleene-axioms-monoid:seq} and Lemma~\ref{le:additive-comb}~\ref{le:additive-comb:star}. 

    \bigskip

  \noindent
  (cons) $\inference{ h \sqsubseteq h' & \vdash \as{h'} \ \r \ \as{k'} & k' \sqsubseteq k}
          {\vdash \as{h} \ \r \ \as{k}}$:

        \bigskip

        \noindent
        by inductive hypothesis we have $\sem{\r}h' \sqsubseteq k'$ and by side condition
        $h \sqsubseteq h'$, which implies $\sem{\r}h \sqsubseteq \sem{\r}h'$, and by side condition $k' \sqsubseteq k$,
        then
        \begin{equation*}
          \sem{\r}h \sqsubseteq \sem{\r}h' \sqsubseteq k' \sqsubseteq k\, .
        \end{equation*}      

  \bigskip

  \noindent
  (choice) $\inference{
            \vdash \as{h} \ \r_1 \ \as{k_1} & \vdash \as{h} \ \r_2 \ \as{k_2} & k_1\oplus k_2 \sqsubseteq k
            }
          {\vdash \as{h} \ \r_1 + \r_2 \ \as{k}}$:
          \bigskip

        \noindent
        by inductive hypothesis we have $\sem{\r_i}h \sqsubseteq k_i$ for $i = 1, 2$, and 
        by the side condition $k_1 \oplus k_2 \sqsubseteq k$, hence
        \begin{equation*}
            \sem{\r_1 + \r_2}h \sqsubseteq \sem{\r_1} h \oplus \sem{\r_2} h
            \sqsubseteq k_1 \oplus k_2 \sqsubseteq k
        \end{equation*}
        where we used Lemma~\ref{le:kleene-axioms-monoid}~\ref{le:kleene-axioms-monoid:sum}. 

    \bigskip
  \noindent
  (iter) $\inference{\forall i \in \N.\: \vdash \as{h_i} \ \r \ \as{h_{i+1}} & \bigoplus_i h_i \sqsubseteq k}{\vdash \as{h_0} \ \r^* \ \as{k}}$:
        \bigskip

        \noindent
         by inductive hypothesis we have $\sem{\r}h_{i} \sqsubseteq h_{i+1}$ for all $i \in \N$.
        For all $ i \geq 1$ it holds that
        $\sem{\r^i} h_{0} \sqsubseteq \sem{\r}h_{i-1}$. In fact, by induction on $i$:
        \begin{itemize}
            \item $\sem{\r^1} h_{0} = \sem{\r}h_{0}$, hence $\sem{\r} h_{0} \sqsubseteq \sem{\r}h_{0}$ by reflexivity of $\sqsubseteq$,
            \item $\sem{\r^{i+1}} h_{0} = \sem{\r^i; \r}h_{0} \sqsubseteq \sem{\r}(\sem{\r^i} h_{0}) \sqsubseteq \sem{\r}(\sem{\r}h_{i-1}) \sqsubseteq \sem{\r}h_{i}$
            by Lemma~\ref{le:kleene-axioms-monoid}~\ref{le:kleene-axioms-monoid:seq}, the two inductive hypotheses and the monotonicity of $\sem{\r}$.
        \end{itemize}
        Hence in particular by the inductive hypotheses we have that for all $i \geq 1$,
        $\sem{\r^i} h_{0} \sqsubseteq \sem{\r}h_{i-1} \sqsubseteq h_{i}$,
        hence by monotonicity of $\oplus$ we have that
        \begin{equation}
            \bigoplus_{i \geq 1} \sem{\r^i} h_0 \sqsubseteq \bigoplus_{i \geq 1} h_i
            \tag{$\dagger$}
            \label{dagger}
        \end{equation}
        and therefore
        \begin{align*}
            \sem{\r^*}h_{0} %
            & = \add{\bigoplus_{i \in \N}} \sem{\r^i} h_{0} \\
            & \sqsubseteq \bigoplus_{i \in \N} \sem{\r^i} h_{0} %
            & %
            [\mbox{Lemma~\ref{le:kleene-axioms-monoid}~\ref{le:kleene-axioms-monoid:star}}]\\
            & = h_{0} \oplus \Bigl ( \bigoplus_{i \geq 1} \sem{\r^i}h_{0} \Bigr )  %
            & %
            [\mbox{by~\eqref{dagger}}]\\
            & \sqsubseteq h_{0} \oplus \Bigl (\bigoplus_{i \geq 1} h_{i} \Bigr ) \\
            & = \bigoplus_{i \in \N} h_{i} \sqsubseteq k %
            & %
            [\mbox{side condition}]
        \end{align*}

  \bigskip

  \noindent
  (rec-star) $\inference{\vdash \as{h} \ \r^* \ \as{k'} & \vdash \as{k'} \ \r \ \as{l} & h\oplus l\sqsubseteq k}{\vdash \as{h} \ \r^* \ \as{k}}$:
        \bigskip

        \noindent
       applying the inductive hypothesis to the premise of the rule we obtain
       $\sem{\r^*}h \sqsubseteq k'$ and $\sem[]{\r}k' \sqsubseteq l$.
 Therefore, for every $b \in \ebase{h}$ it holds: 
 \begin{align*}
  \bsem[]{\r^*}b & \sqsubseteq \bsem[]{\1 + (\r^*; \r)}b & \mbox{[by Lemma~\ref{le:kleene-axioms-monoid}~\ref{le:kleene-axioms-monoid:left-exp2} under condition~\eqref{eq:subadd-oplus}]}\\
  & = b \oplus \add{\bsem[]{\r}}(\bsem[]{\r^*} b ) \\ 
  & \sqsubseteq h \oplus \sem[]{\r} (\sem[]{\r^*} h) & \mbox{[by $b \sqsubseteq h$, $\bsem[]{\r^*}b \sqsubseteq \sem[]{\r^*}h$ and monotonicity of $\oplus, \sem[]{\r}$]}\\
  & \sqsubseteq h \oplus \sem[]{\r} k' & \mbox{[by inductive hypothesis and monotonicity of $\sem[]{\r}$]} \\
  & \sqsubseteq h \oplus l \sqsubseteq k & \mbox{[by inductive hypothesis and side condition]}
\end{align*}
Hence $\sem[]{\r^*} h = \bigsqcup_{b \in \ebase{h}}  \bsem[]{\r^*}b \sqsubseteq k$.

  \bigskip

  \noindent
  (inv-star) 
  $\inference{\vdash \as{h} \ \r \ \as{h} %
  & \bigoplus_{\N} h \sqsubseteq k}{\vdash \as{h} \ \r^* \ \as{k}}$:
        \bigskip

        \noindent
        by inductive hypothesis $\sem{\r}h \sqsubseteq h$ %
        with side condition $\bigoplus_{\N} h \sqsubseteq k$.
Then for all $i \in \N$ it holds $\sem{\r^i}h \sqsubseteq h$. In fact, by induction on $i$:
\begin{itemize}
    \item $\sem{\r^0}h = h \sqsubseteq h$ by reflexivity of $\sqsubseteq$,
    \item $\sem{\r^{i+1}}h = \sem{\r^i; \r}h \sqsubseteq \sem[]{\r}(\sem{\r^i}h) \sqsubseteq \sem[]{\r} h %
    \sqsubseteq h$
    by Lemma~\ref{le:kleene-axioms-monoid}~\ref{le:kleene-axioms-monoid:seq}, the inductive hypotheses and monotonicity of $\sem[]{\r}$, and we conclude by transitivity of $\sqsubseteq$.
  \end{itemize}
 Therefore
  \begin{equation*}
    \sem{\r^*}h = \add{\bigoplus_{i \in \N}}(\sem{\r^i} h) \sqsubseteq \bigoplus_{i \in \N}(\sem{\r^i} h) \sqsubseteq \bigoplus_{i \in \N} h \sqsubseteq k
  \end{equation*}
   where we used Lemma~\ref{le:kleene-axioms-monoid}~\ref{le:kleene-axioms-monoid:star}. 

  \bigskip

  \noindent
  (join) $\inference{\set{h_i}_{i \in I} \text{dense} & h \sqsubseteq \displaystyle\bigsqcup_{i \in I} h_i & \forall i \in I. \vdash \as{h_i} \ \r \ \as{k_i}& \displaystyle\bigsqcup_{i \in I} k_i \sqsubseteq k}{\vdash \as{h}
          \ \r \ \as{k}}$:
        \bigskip

        \noindent
         by inductive hypothesis we have $\sem{\r}(h_i) \sqsubseteq k_i$ for all $i \in I$, hence
        \begin{equation}
            \bigsqcup_{i \in I} \sem{\r}h_i \sqsubseteq \bigsqcup_{i \in I} k_i
            \tag{$\clubsuit$}
            \label{club}
         \end{equation}
         then
         \begin{align*}
            \sem{\r}h %
            & \sqsubseteq \sem{\r}\Bigl ( \bigsqcup_{i \in I} h_i \Bigr ) %
            & %
            [\mbox{side condition and monotonicity of } \sem{\r}] \\
            & = \bigsqcup_{i \in I} \sem{\r}(h_i) %
            & %
            [\mbox{additivity of } \sem{\r} \mbox{ on } \set{h_i \mid i \in I} \mbox{dense} ] \\
            & \sqsubseteq \bigsqcup_{i \in I} k_i %
            & %
            [\mbox{by~\eqref{club}}] \\
            & \sqsubseteq k %
            & %
            [\mbox{side condition}]
         \end{align*}

  \bigskip

  \noindent
  (meet) $\inference{h \sqsubseteq \bigsqcap_{i \in I} h_i & \forall i \in I.\: \vdash \as{h_i} \ \r \ \as{k_i} & \bigsqcap_{i \in I} k_i \sqsubseteq k}{\vdash \as{h}
          \ \r \ \as{k}}$:
        \bigskip

        \noindent
         by inductive hypothesis we have $\sem{\r}h_i \sqsubseteq k_i$ for all $i \in I$, hence
         \begin{equation}
         \bigsqcap_{i \in I} \sem{\r}h_i \sqsubseteq \bigsqcap_{i \in I} k_i
         \tag{$\diamondsuit$}
            \label{diamond}
         \end{equation}
         and for all $i \in I$,
         $\sem{\r}\left ( \bigsqcap_{i \in I} h_i \right ) \sqsubseteq \sem{\r_i}h_i $ by monotonicity of $\sem{\r}$,
         hence \begin{equation}
            \sem{\r}\Bigl ( \bigsqcap_{i \in I} h_i \Bigr ) \sqsubseteq \bigsqcap_{i \in I} \sem{\r} h_i
            \tag{$\heartsuit$}
            \label{heart}
         \end{equation}
         then \begin{align*}
            \sem{\r}h %
            & \sqsubseteq \sem{\r}\Bigl ( \bigsqcap_{i \in I} h_i \Bigr ) %
            & %
            [\mbox{monotonicity of } \sem{\r}] \\
            & \sqsubseteq \bigsqcap_{i \in I} \sem{\r} (h_i) %
            & %
            [(\mbox{by \ref{heart}})] \\
            & \sqsubseteq \bigsqcap_{i \in I} k_i  %
            & %
            [\mbox{by~(\ref{diamond})}] \\
            & \sqsubseteq k %
            & %
            [\mbox{side condition}] \tag*{\qedhere}
         \end{align*}
\end{proof}

The completeness proof relies again on the observation that strongest
postconditions are derivable in $\vdash_{\APPL}$.

\begin{propositionrep}[Derivability of strongest postcondition]\label{de:expressive-language-oplus} 
  Let $(D, \sqsubseteq, \oplus, \base)$ be an interpretation monoid and let $\L$ be an expressive language over it (cf. Definition~\ref{de:expressive-language}). 
  Then for all $\r \in \Reg$ and $h \in \L$ the triple $\vdash \as{h} \ \r \ \as{\sem{\r}h}$ can be derived in $\APPL$.
\end{propositionrep}

\begin{proof}
  We prove the result holds on every representable basis element $b \in \L_{\base} \subseteq \base$.  %
   The result on $\sem{\r}$ then follows by application of the (join) rule:
   let $h \in \L$, then by hypothesis~\ref{de:expressive-language:quant} of Definition~\ref{de:expressive-language} $h = \bigsqcup_{b \in \ebase{h} \cap \L_{\base}} b $, and 
   $\ebase{h} \cap \L_{\base}$ is dense since for every $b \in \base$ such that $b \sqsubseteq \sqcup \ebase{h} \cap \L_{\base} = h$ there exists $b' \in \ebase{h} \cap \L_{\base}$ such that $b \sqsubseteq b'$. 
   Then, assume $\vdash \as{b} \ \r \ \as{\bsem{\r}b}$ is derivable for all $b \in \ebase{h} \cap \L_{\base}$, so we can apply
   \begin{equation*}
    \inference[(join)]{\ebase{h} \cap \L_{\base} \text{ dense} %
    & \forall b \in \ebase{h} .\: \vdash \as{b} \ \r \ \as{\bsem{\r}b}& \bigsqcup_{b \in \ebase{h} \cap \L_{\base}} \bsem{\r}b = \sem{\r}h}{\vdash \as{h}
          \ \r \ \as{\sem{\r}h}}
      \end{equation*}
      Note that $\bigsqcup_{b \in \ebase{h} \cap \L_{\base}} \bsem{\r}b = \sem{\r}h \in \L$ by closure under joins of $\L$ (Definition~\ref{de:expressive-language}~\ref{de:expressive-language:join}).
      
   \noindent Let $b \in \L_{\base}$. We prove that $\vdash \as{b} \ \r \ \as{\sem{\r}b}$ by structural induction on the command syntax. 
 
   \begin{itemize}
    \item If $\r = \e$ then since by Definition~\ref{de:expressive-language}~\ref{de:expressive-language:sem-close} it holds $\bsem{\e}b = k \in \L$, just apply $\inference[(basic)]{\bsem{\e}{b} = k}{\vdash \as{b} \ \e \ \as{k}}$.
    \bigskip
     \item If $\r = \r_1 + \r_2$, by definition $\bsem{\r_1 + \r_2}b=\bsem{\r_1} b \oplus \bsem{\r_2} b$, and by Definition~\ref{de:expressive-language}~\ref{de:expressive-language:sem-close} $\bsem{\r_1 + \r_2}b \in \L$.
       We use the inductive hypothesis on $b$ to know that $\vdash \as{b} \ \r_i \ \as{\bsem{\r_i} b}$ are derivable for $i = 1, 2$. Then
       \begin{equation*}
           \inference[(choice)]{\vdash \as{b} \ \r_1 \ \as{\bsem{\r_1} b} & \vdash \as{b} \ \r_2 \ \as{\bsem{\r_2} b} & \bsem{\r_1}b \oplus \bsem{\r_2}b = k_b}{\vdash \as{b} \ \r_1 + \r_2 \ \as{k_b}}
      \end{equation*}
     \item If $\r = \r_1; \r_2$, by definition $\bsem{\r_1; \r_2}b= \add{\bsem{\r_2}}(\sem{\r_1}b)$, by Definition~\ref{de:expressive-language}~\ref{de:expressive-language:sem-close} $\bsem{\r_1; \r_2}b \in \L$, 
     and by inductive hypothesis we have $\vdash \as{b} \ \r_1 \ \as{\bsem{\r_1}b}$
     and that all  $\vdash \as{b'} \ \r_2 \ \as{\bsem{\r_2}b'}$ for $b' \in \ebase{\bsem{\r}b} \cap \L_{\base}$ are derivable.
     Then we can apply the rule (join), since $\ebase{\bsem{\r_1}b} \cap \L_{\base}$ is dense, 
     to get \begin{equation*}
      \inference[(join)]{\forall b' \in \ebase{\bsem{\r_1}b} .\: \vdash 
        \as{b'} \ \r_2 \ \as{\bsem{\r_2}b' } & \bigsqcup_{b' \in \ebase{\bsem{\r_1}b}} \bsem{\r_2}b' = \add{\bsem{\r_2}}(\bsem{\r_1}b)}{\vdash \as{\bsem{\r_1}b} \ \r_2 \ \as{\add{\bsem{\r_2}}(\bsem{\r_1}b)}}
     \end{equation*}
     then we can derive
     \begin{equation*}
         \inference[(seq)]{\vdash \as{b} \ \r_1 \ \as{\bsem{\r_1}b} & \vdash \as{\sem{\r_1}b} \ \r_2 \ \as{\add{\bsem{\r_2}}(\bsem{\r_1}b)}}
       {\vdash \as{b} \ \r_1; \r_2 \ \as{\add{\bsem{\r_2}}(\bsem{\r_1}b)}}
     \end{equation*}
     \item If $\r = \r_1^*$,
     by definition
     $\bsem{\r_1^*}b = \bigoplus_{i \geq 0} \bsem{\r_1^i} b$, and by Definition~\ref{de:expressive-language}~\ref{de:expressive-language:sem-close} $\bsem{\r_1^*}b \in \L$. 
    Define $k_0 = b$ and $k_{i+1} = \add{\bsem{\r_1}}k_i$ for all $i \geq 1$.
        Now we show that $k_i = \bsem{\r_1^i} b$ for all $i \geq 0$ by induction on $i$.
    \begin{itemize}
        \item $k_0 = b = \bsem{\r_1^0}b$.
         \item $k_{i+1} = \add{\bsem{\r_1}}k_i = \add{\bsem{\r_1}}(\bsem{\r_1^i} b) = \bsem{\r_1^{i+1}} b$ by the inductive hypothesis.
    \end{itemize}
    Hence $k_i \in \L$ for all $i$ by Definition~\ref{de:expressive-language}~\ref{de:expressive-language:sem-close}.
    Then by inductive hypothesis we have that for all $i \geq 0$ each $\vdash \as{b'} \ \r_1 \as{\bsem{\r_1}b'}$ is derivable for $b' \in \ebase{k_i} \cap \L_{\base}$, 
    hence again by applying the rule (join) for $b' \in \ebase{k_i} \cap \L_{\base}$ (which again by Definition~\ref{de:expressive-language}~\ref{de:expressive-language:quant} is dense), since $k_{i+1} = \add{\bsem{\r_1}}k_i = \bigsqcup_{b' \in \ebase{k_i} \cap \L_{\base}} \bsem{\r_1}b'$,
    we have that 
    $\vdash \as{k_i} \ \r_1 \ \as{k_{i+1}}$ is derivable for all $i \geq 0$. 
    Then we can derive
    \begin{equation*}
             \inference[(iter)]{
                \forall i \in \N &  \vdash \as{k_i} \ \r_1 \ \as{k_{i+1}} & \bigoplus_i k_i = k_b 
              }{
                \vdash \as{k_0} \ \r_1^* \ \as{k_b}
                }
            \end{equation*}
         which is, as $k_0 = b$ and $\bigoplus_i k_i = \bigoplus_i \bsem{\r_1^i} b$,
         \begin{equation*}
             \inference[(iter)]{
                \forall i \in \N & \vdash \as{k_i} \ \r_1 \ \as{k_{i+1}} & \bigoplus_i \bsem{\r_1^i} b = k_b
              }{
                \vdash \as{b} \ \r_1^* \ \as{k_b}
                }
                \tag*{\qedhere}
       \end{equation*}
    \end{itemize}
\end{proof}

Condition~\ref{de:expressive-language:sem-close} in Definition~\ref{de:expressive-language} can be in turn expanded to closure with respect to elementary commands and $\oplus$, as stated in the following lemma. 

\begin{lemmarep}\label{le:compl-closure-semantics-oplus}
  Let $(D, \sqsubseteq, \oplus, \base)$ be an interpretation monoid
  and assume that
  properties~\ref{de:expressive-language:quant}
  and~\ref{de:expressive-language:join} of
  Definition~\ref{de:expressive-language} hold. If for all elementary
  commands $\bsem{\e}(\L_{\base}) \subseteq \L$ and $\L$ is closed
  under the monoid operation $\oplus$ %
  then also~\ref{de:expressive-language:sem-close}
  holds, i.e., $\bsem{\r}(\L_{\base}) \subseteq \L$.
\end{lemmarep}

\begin{proof}
  By structural induction on the command syntax. Let $b \in \L_{\base}\subseteq \base$. 
  \begin{itemize}
  \item If $\r = \e$, then $\bsem{\e}b \in \L$ by hypothesis. 
  \item If $\r = \r_1 + \r_2$, then $\bsem{\r_1 + \r_2}b = \bsem{\r_1}b \oplus \bsem{\r_2}b \in \L$ since by inductive hypothesis $\bsem{\r_1}b, \bsem{\r_2}b \in L$ 
  and $\L$ is closed under $\oplus$. 
\item Let $\r = \r_1; \r_2$. By inductive hypothesis $\bsem{\r_1}b \in \L$. Hence, by Definition~\ref{de:expressive-language}~\ref{de:expressive-language:quant}, if we let $h = \bsem{\r_1}b$ %
it holds $h = \bigsqcup_{b \in \ebase{h} \cap \L_{\base}} b$. 
Therefore:  
\begin{align*}
    \bsem{\r_1 ; \r_2}b
    & = \add{\bsem{\r_2}}(\bsem{\r_1}b)\\
    & = \add{\bsem{\r_2}}(h)\\
    & = \sqcup_{b \in \ebase{h} \cap \L_{\base}} \bsem{\r_2}(b) & \mbox{[by Definition~\ref{de:expressive-language}~\ref{de:expressive-language:quant} and monotonicity of $\bsem{\r_2}$]}%
  \end{align*}
  By inductive hypothesis, for each $b \in \ebase{h} \cap \L_{\base}$, %
  we have $\bsem{\r_2}(b) \in \L$. Hence, since $\L$ is closed by joins, %
  $\bsem{\r_1 ; \r_2}b  =  \sqcup_{b \in \ebase{h} \cap \L_{\base}} \bsem{\r_2}(b) \in \L$, as desired.
  \item If $\r = \r_1^*$, then $\bsem{\r_1^*}b = \bigoplus_{i \geq 0} \bsem{\r_1^i}b$. By inductive hypothesis for all $i \geq 0$ we have $\bsem{\r_1^i}b \in \L$, 
  and we conclude since $\L$ is closed under $\oplus$. \hfill \qedhere
  \end{itemize}
\end{proof}
\end{toappendix}

\subsection{Hyperlogics}
\label{ss:hypersemantics}

The \emph{hypermonoid} $(\pow{C}, \subseteq, \otimes, \Ir{\pow{C}})$
over a  complete lattice $(C, \leq)$ (Example~\ref{ex:hyper-interpretation}) serves as the foundation for hypersemantics.
Assume that a semantics $\sem[C]{\cdot}$ is given on the underlying
lattice $C$, defined inductively, i.e., with the trivial basis $C$
(\S~\ref{inst:simple-inductive}), and define the semantics of elementary commands on basis
elements as $\bsem{\e}(\{c\}) \triangleq \{\sem[C]{\e}c\}$. 
By induction on
$\r$, it is easily seen that $\bsem{\r}(\{c\}) = \{\sem[C]{\r}c\}$. Since $\pow{C}$ is completely
meet-distributive and the basis consists of irreducibles, $\sem{\r}$ is additive
(Lemma~\ref{le:basic-prop-sem-mon}~\ref{le:basic-prop-sem-mon:distrib}), so that the
full semantics $\sem{\r}$ acts pointwise, i.e., for $h \in \pow{C}$,
$\sem{\r} h = \{\sem[C]{\r}c \mid c \in h\}$.

In particular, when $(C, \leq) = (\pow{\Sigma}, \subseteq)$,  we obtain the \emph{collecting hypersemantics}: on elementary commands $\e \in \Ecom$ and basis elements $\{X\} \in \Ir{\pow{\pow{\Sigma}}}$, the semantics is defined via the strongest postcondition:
  $\bsem{\e}(\{X\}) \triangleq \{\post{\e}(X)\}$.
This extends to the full semantics:
for any $\r \in \Reg$ and hyperproperty $h \in \pow{\pow{\Sigma}}$, it holds
  $\sem{\r} h = \{\post{\r}(X) \mid X \in h\}$ (cf. Lemma~\ref{le:basic-prop-sem-mon}~\ref{le:basic-prop-sem-mon:distrib}). %
Since $\post{\r}$ is additive, condition~\eqref{eq:subadd-oplus} holds with equality.
The logic $\APPL$ on the hypermonoid $(\pow{\pow{\Sigma}}, \subseteq, \otimes, \Ir{\pow{\pow{\Sigma}}})$ yields a proof system for hyperproperties, which is complete and, therefore, 
equivalent to the Hyper Hoare Logic (HHL) of \citet{DM:HHL}.
Axioms for elementary commands in~\cite{DM:HHL}, i.e., (Skip), (Assume), (Assign) are based on weakest-precondition transformers 
and, therefore, are axioms in $\APPL$. Rules (Seq), (Choice), (Cons), (Iter) of~\cite{DM:HHL} occur with the same name 
in $\APPL$, while rule (Exist) in~\cite{DM:HHL} corresponds to rule (join) in $\APPL$. %
The density condition of rule (join) is vacuous as all non-empty sets are dense in the hypermonoid.

\begin{figure*}[b]
  \hspace{-20pt}
      \resizebox{0.9\textwidth}{!}{
        \inference[(join)]{
        \inference[(choice)] {\vdash \as{P_0} \ \1 \ \as{P_0} & \vdash \as{P_0} \ \mathsf{inc} \ x \ \as{P_1}}{\vdash \as{P_0} \ \1 + (\mathsf{inc} \ x) \ \as{P_0 \otimes P_1}}
        & \inference[(choice)]{\vdash \as{P_2} \ \1 \ \as{P_2} 
        & \vdash \as{P_2} \ \mathsf{inc} \ x \ \as{P_3}}{\vdash \as{P_2} \ \1 + (\mathsf{inc} \ x) \ \as{P_2 \otimes P_3}}
        }{\vdash 
        \as{P_0 \cup P_2} \ \1 + (\mathsf{inc} \ x) \ \as{(P_0 \otimes P_1) \cup (P_2 \otimes P_3)}
        }
      }
      \caption{Derivation for Example~\ref{ex:hyper-muller}, where $P_v \triangleq \{\{\sigma_v\}\} \in \pow{\pow{\Sigma}}$ for $v\in\Z$. %
      }
    \label{fig:hyper-muller}
  \end{figure*}

\begin{example}\label{ex:hyper-muller}
  As a motivating example, 
  \citet[Example 1]{DM:HHL} consider the program $\mathsf{s} \triangleq \1 + (\mathsf{inc} \ x)$.
  For a value $v \in \mathbb{Z}$, we write $\state{v}\in \Sigma$
  for the state 
  where $x$ has value $v$ and all other variables have value $0$, and consider the singleton hyperproperty
  $P_v \triangleq \{\{\state{v}\}\} \in \pow{\pow{\Sigma}}$.  
  They show that the triple
    $\as{P_0 \cup P_2} \, \mathsf{s} \, \as{(P_0 \otimes P_1) \cup (P_2 \otimes P_3)}$
  is derivable in HHL, critically relying on the rule (Exist) and the indexing of hyperproperties by shadow variables.
  By definition of $\otimes$ in \cite[Definition~6]{DM:HHL},  the %
  triple is
  \(
    \as{\{\{\state{0}\}, \{\state{2}\}\}} \, \mathsf{s} \, \as{\{\{\state{0},\state{1}\}, \{\state{2},\state{3}\}\}}
    \),
    as $(P_0 \otimes P_1) \cup (P_2 \otimes P_3) =\{\{\state{0},\state{1}\}, \{\state{2},\state{3}\}\}$.
  In Figure~\ref{fig:hyper-muller} we show that 
  this triple is derivable in $\APPL$
 using (join), that mirrors %
(Exist) of~\cite{DM:HHL}, 
to split the analysis into computations starting from $P_0$ and $P_2$.
  \hfill$\lozenge$
\end{example}

\begin{remark}\rm 
 Rule (join) is pivotal for completeness of hyperproperty logics: it corresponds to
  (Exist) in HHL~\cite{DM:HHL} and Outcome Logic~\cite{Zil:toplas}. 
  \citet{CousotW25} include no primitive join/exist rule, folding the same
  functionality into their $\mathsf{if}$ (rule (69)) and $\mathsf{loop}$ (rule (70)). Thus,
  (join) is essential in hyperlogics yet redundant in Hoare logic, for an
  algebraic reason: with $\oplus = \otimes$ the monoidal operation does not preserve
  additivity, whereas with $\oplus = \cup$ it does (cf.\ Proposition~\ref{pr:nojoin}).
  \hfill$\lozenge$
\end{remark}

\subsection{Incorrectness Logic}
\label{ss:incorrectness}
The instance of $\APPL$ over the interpretation monoid
$(\pow{\Sigma}, \supseteq, \cup, \pow{\Sigma})$ of Example~\ref{ex:incorrectness-interpretation}
is an incorrectness logic in the style of~\citet{OHearn20}. Where~\cite{OHearn20} separates
normal termination ({\color{green!70!black}$\mathsf{ok}$}) from errors
({\color{red}$\mathsf{er}$}), we omit the distinction, focusing on under-approximating
reachable states as in recent work~\cite{AscariBGL25,VerschtK25}.
Since the basis is $\pow{\Sigma}$ itself, we have $\mon{\base} = \pow{\Sigma}$, and
Definition~\ref{de:semantics-monoid} yields the ordinary collecting semantics
$\sem{\r}X = \bsem{\r}X$.
Condition~\eqref{eq:subadd-oplus} holds with equality, since the collecting semantics is
additive: $\bigcup_{i \in I} \sem{\r}X_i = \sem{\r}(\bigcup_{i \in I} X_i)$.

Relative completeness already guarantees that $\APPL$
matches the expressive power of incorrectness logic, but most rules of~\cite{OHearn20}
 are directly recoverable. A key observation is that for any $\r$
and $h$, the trivial triple $\vdash \as{h}\,\r\,\as{\top}$ is always derivable by following
the structure of $\r$ with (basic), (seq), (choice), (iter). Since $\top = \varnothing$ in
$(\pow{\Sigma},\supseteq)$, this is exactly the \textsc{Empty-Under-Approximate} rule
$[h]\,\r\,[\varnothing]$ in \cite{OHearn20}, from which the remaining rules follow
(for the details see Remark~\ref{rem:ohearn} in the Appendix).
Finally, note that under the reverse-inclusion order, a set is dense when it has a minimum element, so, 
by monotonicity of the semantics, the (join) rule is redundant. %

\begin{toappendix}

\begin{remark}[Recovering Incorrectness rules from~\cite{OHearn20}]\rm 
\label{rem:ohearn}
We detail how the key rules from~\cite{OHearn20} (referred to here using \textsc{Small Caps}) are recovered in $\APPL$.

\textsc{Choice}: This rule is recovered by applying the $\APPL$ rule 
(choice) and using the trivial approximation for the unchosen branch. For example, to derive $[h] \ \r_1 + \r_2 \ [k]$ from $[h] \ \r_1 \ [k]$, we use the premise $\vdash \as{h} \ \r_1 \as{k}$ alongside the trivial premise $\vdash \as{h} \ \r_2 \as{\top}$. Since the monoid operation is $\cup$ and $\top = \varnothing$, the conclusion yields $k \cup \varnothing = k$.
    
\textsc{Disjunction}: This is a direct instance of the (meet) rule. In the reversed lattice, the meet operation $\sqcap$ corresponds to set union $\cup$. Thus, applying (meet) to premises with postconditions $k_1$ and $k_2$ yields the postcondition $k_1 \cup k_2$.

\textsc{Iterate-Zero} corresponds to $\vdash \as{h} \ \r^* \ \as{h}$, %
derived by applying (rec-star)
to the trivial premises $\vdash \as{h} \ \r^* \ \as{\top}$ and $\vdash \as{\top} \ \r \ \as{\top}$: since 
$\top = \varnothing$, 
in the conclusion $h \oplus \top = h \cup \varnothing = h$. 
The condition~\eqref{eq:subadd-oplus} that makes the rule (rec-star) sound holds, as the collecting semantics is additive.

 \textsc{Iterate-Non-Zero}: This rule %
 coincides with (rec-star) : both rules use the expansion $\r^*; \r$ (looping then executing the body). The premise $[h] \ \r^*; \r \ [k]$ of 
 \textsc{Iterate-Non-Zero} corresponds to the premises $\vdash \as{h} \ \r^* \ \as{k'}$ and $\vdash \as{k'} \ \r \ \as{k}$ of (rec-star), 
 with side condition $h \cup k \supseteq k$ always satisfied.

\textsc{Backwards-Variant} coincides with the infinitary rule (iter).
\end{remark}
\end{toappendix}

\begin{figure*}[t]
      \resizebox{\textwidth}{!}{
    $\inference[(iter)]{
            \inference[(choice)]{
              \dots 
            }{\vdash \as{h_0} \ \r_1 + \r_2 \ \as{h_1}} 
            & 
            \inference[(choice)]{
              \dots
            }{\vdash \as{h_1} \ \r_1 + \r_2 \ \as{h_2}} 
            & \set{0, 2, 1000} \subseteq h_0 \cup  h_1 \cup h_2
          }{\vdash \as{h_0} \ (\r_1 + \r_2)^* \ \as{\set{0, 2, 1000}}}$
      }
    \caption{Derivation for Example~\ref{ex:incorr-LCL-example}. %
    $h_0 \triangleq \set{0, 999}$, $h_1 \triangleq \set{1, 998, 1000}$, $h_2 \triangleq \set{0, 2, 997, 999}$.} %
    \label{fig:incorr-LCL-example}
  \end{figure*}

  \begin{example}
  \label{ex:incorr-LCL-example}
  Consider a variation of~\cite[Example 5.2]{BruniGGR23}, by taking the program 
  $\mathsf{u} \triangleq {(\r_1 + \r_2)}^*$, where %
    $\r_1 \triangleq (\mathit{x} > 0?); \mathsf{dec} \ x$ and
    $\r_2 \triangleq (\mathit{x} < 1000?); \mathsf{inc} \ x$. 
  We refute the safety specification $(\mathit{x} \neq 2)$ 
  by deriving the incorrectness triple
    ${\vdash \as{\set{0, 999}} \ \mathsf{u} \ \as{\set{0, 2, 1000}}}$
  in $\APPL$ instantiated with the interpretation monoid for incorrectness logic. %
  Figure~\ref{fig:incorr-LCL-example} shows the
  start of the derivation, with $h_1 \triangleq \sem{\r_1}h_0 \cup \sem{\r_2}h_0$ and
  $h_2 \triangleq \sem{\r_1}h_1 \cup \sem{\r_2}h_1$. \hfill$\lozenge$
 \end{example}

\subsection{Abstract interpretation for interpretation monoids}\label{ss:abstract-semantics}

Just as an interpretation domain can be abstracted via a Galois insertion
(\S~\ref{ss:ai-interpretation}), we now give conditions under which an interpretation
monoid can be abstracted: this underpins the abstraction of hyperproperties described in \S~\ref{sec:abs-hyper}.

\begin{lemmarep}[Abstract interpretation monoid]
  \label{le:abs-monoid}
  Let $(C, \sqsubseteq, \oplus, \base)$ be an interpretation monoid
  and $\langle \alpha,\gamma \rangle: C \rightleftarrows A$ a Galois
  insertion. Assume the abstract monoidal operator $\oplus^A$ is
  complete, i.e.\ $\alpha \comp \oplus = \oplus^A \comp \alpha$. Then
  $(A, \sqsubseteq_A, \oplus^A, \alpha(\base))$ is an interpretation
  monoid.
\end{lemmarep}

\begin{proof}
  First, recall that by
  Lemma~\ref{le:base-transf}~\ref{le:base-transf:base},
  $\alpha(\base)$ is indeed a basis for $A$, that we refer as $\base_A$.
  
  Moreover, completeness of the abstraction means that
  \begin{center}
    $\alpha(\bigoplus_{i \in I} h_i) = \bigoplus^A_{i \in I}
    \alpha(h_i)$, 
  \end{center}
  hence $\alpha$ is a (surjective) monoid homomorphism.
  Therefore
  \begin{equation}
    \label{eq:homo}
    \alpha(\mon{\base}) = \mon{\alpha(\base)}
  \end{equation}

  Let $\kappa_A$ be the weight of $\mon{\base_A}$ w.r.t.\
  $\base_A$.  Observe that $\kappa_A \leq \kappa_C$.
  In fact, let $h_A \in \mon{\base_A}$. By the above,
  $h_A = \alpha(h)$ for some $h \in \mon{\base}$. By definition of
  weight, there exists $Y \subseteq \base$ such that
  $|Y| \leq \kappa_C$ and $h = \sqcup Y$. Therefore if we let
  $Y_A = \alpha(Y)$ it holds that $|Y_A| \leq |Y| \leq \kappa_C$ and
  $\sqcup Y_A = \sqcup \alpha(Y) = \alpha(\sqcup Y) = \alpha(h) =
  h_A$. Hence $\kappa_A \leq \kappa_C$.
  
  From these facts the validity of the conditions in
  Definition~\ref{de:interpretation-monoid}~(iii)
  for $A$ readily follows. In fact, 
   if $Y \subseteq \mon{\alpha(\base)}$, with
     $|Y| \leq \kappa_A$, by \eqref{eq:homo}, $Y = \alpha(X)$ for
     $X \subseteq \mon{\base}$, $|X| \leq |\kappa_A$. Since
     $\kappa_A \leq \kappa_C$, the join $\sqcup X \in
     \mon{\base}$. Therefore
     \begin{align*}
       \bigsqcup_A X
       = \bigsqcup_A \alpha(X)
       = \alpha(\sqcup X) \in \alpha(\mon{\base}) = \mon{\alpha(\base)}
       \tag*{\qedhere}
     \end{align*}

\end{proof}

The comparison between concrete and abstract logics from \S~\ref{ss:ai-interpretation} and \S~\ref{sec:abs-sublogic} extends naturally 
to the monoidal setting. Let $(C, \sqsubseteq, \oplus, \base)$ be an interpretation monoid, 
let $\langle \alpha,\gamma \rangle : C \rightleftarrows A$ be a Galois insertion, 
and assume that $A$ is the logical language, i.e., $\L = A$. 
Then every abstract derivation lifts to a concrete one, and the converse holds whenever $\alpha$ reflects density.

\begin{propositionrep}[Concrete vs.\ abstract logics]\label{prop:abstract-logic-monoid}
  Under the hypotheses of Lemma~\ref{le:abs-monoid}, let $a, b \in A$ and $\r\in \Reg$. If\/
  $\vdash_A \as{a} \ \r \ \as{b}$ is derivable in {\rm $\APPL_A$}, then
  $\vdash \as{\gamma(a)} \ \r \ \as{\gamma(b)}$ is derivable in {\rm $\APPL$}. If $\alpha$
  reflects density or $C$ is completely meet-distributive, the converse holds too.
\end{propositionrep}

\begin{proof}
  The results immediately follows from
  Propositions~\ref{prop:sound-abstract-monoid}
  and~\ref{prop:equivalence-monoid}.
  \hfill
\end{proof}

\begin{toappendix}
  \begin{theoremrep}[$\APPL_A$ soundness]
  \label{prop:sound-abstract-monoid}
  Let $a, b \in A$ and $\r\in \Reg$. If
  $\vdash_A \as{a} \ \r \ \as{b}$ is derivable in {\rm $\APPL_A$},
  then $\sem{\r}{\gamma(a)} \sqsubseteq \gamma(b)$, and, consequently, 
  $\vdash \as{\gamma(a)} \ \r \ \as{\gamma(b)}$ holds.
\end{theoremrep}

\begin{proof}
  Immediate consequence of the fact that for all abstract rules, once
  abstract properties are replaced by their concretisation, are
  instances of the corresponding concrete rules. For instance, the
  abstract rule (choice) requires in the premises
  $b_1 \oplus^A b_2 \sqsubseteq_A b$, i.e.,
  $\alpha(\gamma(b_1) \oplus \gamma(b_2)) \sqsubseteq_A b$. Thus, by
  adjointness
  \begin{center}
    $\gamma(b_1) \oplus \gamma(b_2) \sqsubseteq  \gamma(b)$.
  \end{center}
  which is a
  valid premise for (choice) in the concrete.

  In the case of rule (join), we exploit Lemma~\ref{le:base-transf}~\ref{le:base-transf:density} for deducing that the premise
  $\set{a_i \mid i \in I}$ being dense in $A$ implies that
  $\set{\gamma(a_i) \mid i \in I}$ is dense in $C$.
  \hfill
\end{proof}

\begin{propositionrep}[Equivalence of $\APPL$ \& $\APPL_A$]\label{prop:equivalence-monoid}
  Let $a, b \in A$ and $\r\in \Reg$. If $\alpha$ reflects density then 
  $\vdash_A \as{a} \ \r \ \as{b}$ is derivable in {\rm $\APPL_A$}
  iff $\vdash \as{\gamma(a)} \ \r \ \as{\gamma(b)}$ is derivable in {\rm $\APPL$}.
\end{propositionrep}

\begin{proof}
  This is analogous to the proof of
  Proposition~\ref{prop:equivalence}. In fact, under the hypotheses,
  all concrete rules, if used on (concretisations of) abstract
  elements can be seen as instances of the abstract ones. For
  instance, the concrete rule (choice) would be
 \[
    \inference[(choice)]{\vdash \as{\gamma(a)} \ \r_1 \ \as{\gamma(b_1)} & \vdash_A \as{\gamma(a)} \ \r_2 \ \as{\gamma(b_2)} & \gamma(b_1)\oplus \gamma(b_2) \sqsubseteq \gamma(b)}
    {\vdash \as{\gamma(a)} \ \r_1 + \r_2 \ \as{\gamma(b)}}
\]
  The request in
  $\gamma(b_1) \oplus \gamma(b_2) \sqsubseteq  \gamma(b)$, by adjointness, implies
  \begin{center}
    $b_1 \oplus^A b_2  = \alpha(\gamma(b_1) \oplus \gamma(b_2)) \sqsubseteq_A = b$.
  \end{center}
   which is a
   valid premise for (choice) in the abstract system.

   The case of rule (join) is treated exactly as for Proposition~\ref{prop:equivalence}.
 \hfill
\end{proof}
\end{toappendix}

\subsection{Abstraction of Hyperproperties}
\label{sec:abs-hyper}

We give two instances of $\APPL$ for hyperproperties through the framework of
\S~\ref{ss:abstract-semantics}. Fix a complete lattice $(C, \leq)$ and the instance of
$\APPL$ over the hypermonoid $(\pow{C}, \subseteq, \otimes, \Ir{\pow{C}})$
(\S~\ref{ss:hypersemantics}). Since $\pow{C}$ is completely meet-distributive with an irreducible
basis, every non-empty subset is dense, a fact we use throughout. For intuition the reader
may safely assume $(C, \leq) = (\pow{\Sigma}, \subseteq)$.

\subsubsection{Down-closed Hyperproperties}\label{ss:subset-close}

Down-closed hyperproperties are widely
studied~\cite{ADST:HSASTIF,MP:VBSCH,CousotW25,MP:HSFF} for their support of inductive
approximation of hypersemantics. Following the observation
by~\citet[\S~16]{CousotW25} that they are abstractions of general hyperproperties, we place
them in our framework. Let
$\pow{C}^\downarrow \triangleq \set{X \subseteq C \mid X = \down X}$ be the sublattice of
down-closed sets, presented as a Galois insertion via
$\alpha^\downarrow(X) \triangleq \down{X}$ and $\gamma^\downarrow(h) \triangleq h$. For
$(C, \leq) = (\pow{\Sigma}, \subseteq)$ this is the order-ideal abstraction
of~\cite{CousotW25}. It turns out that the BCA $\oplus^\downarrow$ of $\oplus$ is $\oplus$ itself, hence
complete. Thus, we obtain the abstract interpretation monoid
$(\pow{C}^\downarrow, \subseteq, \oplus, \Ir{\pow{C}^\downarrow})$ and the logic
$\APPL_{\pow{C}^\downarrow}$, where (join) applies freely as all sets are dense
(for details see~\S~\ref{ss:appendix-down-closed} in the Appendix).

\begin{toappendix}
  \subsection{Additional Material for Abstraction of Hyperproperties}
  \subsubsection{Additional Material on Down-closed Hyperproperties}\label{ss:appendix-down-closed}
\begin{proposition}\label{pr:abstract-down}
  The structure
  $(\pow{C}^\downarrow, \subseteq, \oplus, \Ir{\pow{C}^\downarrow})$
  is an abstract interpretation monoid.
\end{proposition}

\begin{proof}
  The result follows by Lemma~\ref{le:abs-monoid}, recalling that
  $\oplus^\downarrow = \oplus$ is complete. The abstract basis
  obtained as the image under $\alpha^\downarrow$ of
  $\Ir{\pow{C}} = \{\{c\} \mid c \in C\}$ is
  $\{ \down{c} \mid c \in C \}$, which is the irreducible basis of the
  abstract domain.
  \hfill
\end{proof}

The elements of the abstract basis $\Ir{\pow{C}^\downarrow} = \{ \down{c} \mid c \in C \}$ are complete primes in $\pow{C}^\downarrow$, so $\pow{C}^\downarrow$ is prime algebraic and therefore %
completely distributive.
Hence, the density condition of the (join) rule holds for any non-empty subset.
Moreover, since the abstract domain $\pow{C}^\downarrow$ is completely distributive with an irreducible basis, the abstraction $\alpha^\downarrow$ trivially reflects density.
By Proposition~\ref{prop:abstract-logic-monoid}, we obtain the following equivalence result, where $\vdash$ and $\vdash_\downarrow$ denote derivability in $\APPL_{\pow{C}}$ and $\APPL_{\pow{C}^\downarrow}$, respectively.

\begin{corollary}[Equivalence of $\APPL_{\pow{C}}$ \& $\APPL_{\pow{C}^\downarrow}$]
For all down-closed hyperproperties 
$h, k\in \pow{C}^\downarrow$ and commands $\r\in \Reg$:
  $\vdash \as{h} \ \r \ \as{k} 
  \ \Leftrightarrow \ 
  \vdash_\downarrow \as{h} \ \r \ \as{k}$. %
\end{corollary}
\end{toappendix}

  \begin{figure*}[t]
     \hspace{-20pt}
     {\footnotesize
     \resizebox{\textwidth}{!}{
      \inference[(cons)]{
        \set{[0,0]} \subseteq T_2 
        &
      \inference[(inv-star)]{ 
        \inference[(cons)]{\inference[(join)]{
          \forall n \in \N 
        & 
          \as{\set{[2n,2n]}} \ \ti \ \as{\set{[2n+2,2n+2]}}
        }{
          \vdash \as{T_2} \ \ti \ \as{T_2^+
          }} & 
          T_2^+ \sqsubseteq T_2 
          }{\vdash \as{T_2} \ \ti \ \as{T_2%
          }}  & 
          \bigoplus_{i \in \N} T_2 = \set{[0, +\infty]}}
      {\vdash \as{T_2} \ \ti^* \ \as{\set{[0,+\infty]}}}}
      {\vdash \as{\set{[0,0]}} \ \ti^* \ \as{\set{[0,+\infty]}}}
     }
     }
  \caption{Derivation for Example~\ref{ex:down-ex} in $\pow{\Int}^\downarrow$, with $T_2 \triangleq \set{[2n,2n]}_{n \in \N}$ and $T_2^+ \triangleq \set{[2n,2n]}_{n \in \N \setminus \set{0}}$}
    \label{fi:int-down} 
\end{figure*}

\begin{example}
  \label{ex:down-ex}
  Consider the program
  $\ti \triangleq ((x\,\mathsf{mod}\,2 = 1?);x:=x-2) + (x:=x+2)$ from
  Example~\ref{ex:incompleteness}, where
  $\as{[0,0]}\,\ti^*\,\as{[0,+\infty]}$ is not derivable in
  $(\Int, \sqsubseteq, \Ir{\Int})$ without the infinitary (iter), even
  though $\sem{\ti^*}[0,0] = [0,+\infty]$.  Nevertheless, by
  shifting to the level of hypersemantics, it is possible to overcome
  this limitation and establish the target assertion using the
  finitary rules for iteration.  Over the hypermonoid
  $(\pow{\Int}, \subseteq, \otimes, \Ir{\pow{\Int}})$ (defined as
  in~\S~\ref{ss:hypersemantics}), in the down-closed abstraction
  monoid $\pow{\Int}^\downarrow$ we obtain the derivation of
  Fig.~\ref{fi:int-down}, which avoids (iter), where, for readability,
  down-closed sets are shown by their maximal elements. \hfill
  $\lozenge$
\end{example}

\begin{toappendix}
\subsubsection{Pointwise Abstraction}\label{ss:pointwise}
Given a Galois insertion $C \galoiS{\alpha}{\gamma} A$, we consider the abstraction for hyperproperties that applies the underlying abstraction $\alpha$ pointwise to the elements of a hyperproperty.
Formally, we define the Galois insertion
$\pow{C} \galoiS{\alpha_p}{\gamma_p} \pow{A}$, 
where, for $X \in \pow{C}$, we let $\alpha_p(X) \triangleq \set{\alpha(c) \mid c \in X}$, and for $Y \in \pow{A}$, we let $\gamma_p(Y) \triangleq \set{X \in \pow{C} \mid \alpha_p(X) \in Y}$.
This transformation is referred to as homomorphic semantic abstraction in~\cite{CousotW25} and as the lifting transformer in~\cite{MP:VBSCH}.
The abstract basis of $\pow{A}$, obtained by abstracting the irreducible basis of $\pow{C}$, is itself an irreducible basis, as singletons $\set{c} \in \Ir{\pow{C}}$ are mapped directly to singletons $\set{\alpha(c)} \in \Ir{\pow{A}}$.
Since the abstract lattice is a powerset, it is completely distributive. Consequently, as in 
\S~\ref{ss:subset-close}, the density requirement of the
(join) rule is vacuously satisfied, and $\alpha_p$ reflects density.
Therefore, by Proposition~\ref{prop:abstract-logic-monoid}, the abstract program logic 
$\APPL_{\pow{A}}$
is equally expressive to the concrete system over abstract points,  
that is, for all $Y, Y' \in \pow{A}$ and any 
command $\r$:
\[
  \vdash \as{\gamma_p(Y)} \ \r \ \as{\gamma_p(Y')} \iff \vdash_{\pow{A}} \as{Y} \ \r \ \as{Y'} \,.
\]

Moreover, $\alpha_p$ is complete with respect to $\oplus$ (Lemma~\ref{le:oplus-pointwise}).
Hence, the abstract logic over $(\pow{A}, \subseteq, \oplus, \Ir{\pow{A}})$ is sound and complete with respect to the corresponding abstract semantics.
In fact, the domain $\pow{\Int}^\downarrow$, discussed in Example~\ref{ex:down-ex}, 
is derived from the hypermonoid of concrete hyperproperties $\pow{\pow{\Sigma}}$ by composing the pointwise abstraction (induced by the interval abstraction $\pow{\Sigma} \to \Int$) with the down-closure abstraction.

\begin{lemmarep}
  \label{le:oplus-pointwise}
  Let $\pow{C} \galoiS{\alpha_p}{\gamma_p} \pow{A}$ be the pointwise abstraction defined above. %
  Then for every indexed family $\set{X_i}_{i \in I}$ in $\pow{C}$ it holds:
  \begin{center}
    $ \alpha_p\left (\bigoplus_{i \in I} X_i \right ) = \bigoplus_{i \in I} \alpha_p(X_i)$
  \end{center}
\end{lemmarep}

\begin{proof}
  Given an indexed family $\set{X_i}_{i \in I}$ in $\pow{C}$ we have
  \begin{align*}
    \alpha_p\Bigl (\bigoplus_{i \in I} X_i \Bigr )
    & = \alpha_p \Bigl (\Bigl \{ \bigsqcup_{i \in I} c_i \mid \forall i \in I.\, c_i \in X_i \Bigr \}\Bigr ) \\
    & = \Bigl \{ \alpha\Bigl (\bigsqcup_{i \in I} c_i \Bigr ) \mid \forall i \in I.\, c_i \in X_i \Bigr \} \\
    & = \Bigl \{\bigsqcup_{i \in I} \alpha(c_i) \mid \forall i \in I.\, c_i \in
      X_i \Bigr \}\\
    & = \bigoplus_{i \in I} \alpha_p(X_i)
    \tag*{\qedhere}
  \end{align*}
\end{proof}

\end{toappendix}

\subsubsection{Intervals of Hyperproperties}\label{ss:intervals-hyper}

Assume $(C, \leq)$ completely distributive and let 
$\langle \alpha,\gamma \rangle:(C, \leq) \rightleftarrows (A, \leq_{A})$
be a Galois insertion. %
Let $\sem[C]{\r} : C \to C$  and
$\sem[A]{\e}$ be a concrete semantics on $C$ and the corresponding abstract semantics (both with trivial basis and thus defined inductively, cf.~\S~\ref{inst:simple-inductive}).

We abstract $(\pow{C}, \subseteq, \otimes, \Ir{\pow{C}})$ over the product domain $C \times A$, with the following 
order: 
\begin{equation}\label{eq:order-interval}
  \tag{O$_1$}
\interval{c}{a} \sqsubseteq_I \interval{c'}{a'} \defiff  c' \leq c \land a \leq_A a'\, .
\end{equation}
The Galois insertion 
$\langle \alpha_I,\gamma_I \rangle:\pow{C} \rightleftarrows C \times A$ is defined as follows:
for all $X \in \pow{C}$ and $(c,a) \in C \times A$, 
\begin{center}
 $ \textstyle \alpha_I(X) \triangleq \left [ \wedge X,  \vee^{A} \alpha(X) \right ] \, ,
  \quad
  \gamma_I([ c,a ] ) \triangleq \set{ x \in C \mid c \leq x \leq \gamma(a) }\, .$
\end{center}

The abstraction $\alpha_I$ is complete for $\oplus$. Hence,
we can consider the abstract interpretation monoid $(C \times A, \sqsubseteq, \oplus^I, \alpha_I(\Ir{\pow{C}}))$, 
and the corresponding logic $\APPL_I$ (details are in the Appendix~\ref{ss:intervals-hyper-appendix}).
The abstract basis $\alpha_I(\Ir{\pow{C}})$ of $C \times A$, obtained by abstracting the irreducible basis of $\pow{C}$ (i.e., singleton sets), consists of the intervals $[ c, \alpha(c) ]$, 
for $c \in C$.
Notably, the abstract operations required by the program logic can be  efficiently defined on this interval representation, 
similarly to standard numerical intervals, as follows: $\bigoplus^I_{j \in J} \interval{c_j}{a_j} = \interval{\bigvee_{j \in J} c_j}{\bigvee_{j \in J} a_j}$ and $\sem{\r}^I \interval{c}{a} = \interval{\isem[C]{\r}c}{\isem[A]{\r} a}$.

\begin{toappendix}
  \subsubsection{Additional Material for Intervals of Hyperproperties}\label{ss:intervals-hyper-appendix}
\begin{lemmarep}[Properties of $C \times A$]\label{le:hyperproperties-intervals}\ \
  \begin{enumerate}[label=\rm (I\arabic*)]
  \item
    \label{le:hyperproperties-intervals:BCA}
    For all elementary commands $\e \in \Ecom$ and $\interval{c}{a} \in C \times A$, we have
      $\sem{\e}^I\interval{c}{a} = \interval{\isem[C]{\e} c}{\isem[A]{\e} a}$.
    
  \item\label{le:hyperproperties-intervals:oplus}
    Let $[c_j,a_j] \in C \times A$ for $j \in J$. Then %
      $\displaystyle 
      {\bigoplus_{j \in J}}^I \interval{c_j}{a_j} =
      \textstyle
      \big[\bigvee_{j \in J} c_j, \bigvee_{j \in J} a_j\big] %
    $. %
  \end{enumerate}
\end{lemmarep}

\begin{proof}
  \ref{le:hyperproperties-intervals:BCA} First observe that
  $\interval{c}{a} \in C \times A$, the set of basis element below
  $\interval{c}{a}$ are
  \begin{equation}
    \label{eq:base-interval}
    \ebase{\interval{c}{a}} = \set{\interval{c'}{\alpha(c')} \mid c
      \leq c' \ \land\ c' \leq \gamma(a)}
  \end{equation}
  Therefore
  \begin{align*}
    \bsem{\e}^I\interval{c}{\alpha(c)}%
    & = \alpha_I\left (\sem{\e} \gamma_I(\interval{c}{\alpha(c)}) \right )\\
    & = \alpha_I\left (\set{\isem[C]{\e}x \mid c \leq x \leq \gamma(\alpha(c))}\right )\\ 
    & = [\wedge \set{\isem[C]{\e}x \mid c \leq x \leq \gamma(\alpha(c))}, %
    \vee \set{\alpha(\isem[C]{\e}x) \mid c \leq x \leq \gamma(\alpha(c))}] \\
    & \quad \mbox{[since $\isem[C]{\e}$ is monotone]}\\
    & = \interval{\isem[C]{\e}c}{\alpha(\isem[C]{\e}\gamma(\alpha(c)))}\\
    & = \interval{\isem[C]{\e}c}{\isem[A]{\e}\alpha(c)}
  \end{align*}

  \medskip
  
  Finally, 
  \begin{align*}
    \sem{\e}^I \interval{c}{a}
    & = \bigsqcup
    \set{\bsem{\e}^I\interval{c'}{\alpha(c')}
      \mid \interval{c'}{\alpha(c')} \in \ebase{\interval{c}{a}}}%
    & \mbox{[by equation~\eqref{eq:base-interval}]}\\
    & = \sqcup \set{\bsem{\e}^I\interval{c'}{\alpha(c')}
      \mid c
      \leq c' \ \land\ c' \leq \gamma(a)}\\
    & = \sqcup \set{\interval{\isem[C]{\e}c'}{\isem[A]{\e}\alpha(c')}
      \mid c
      \leq c' \ \land\ c' \leq \gamma(a)}\\
    & = \interval{\isem[C]{\e}c}{\isem[A]{\e}\alpha(\gamma(a))}\\
    & = \interval{\isem[C]{\e}c}{\isem[A]{\e}a}
  \end{align*}
  \bigskip

  \ref{le:hyperproperties-intervals:oplus}
  \begin{align*}
    {\bigoplus_{j \in J}}^I\interval{c_j}{a_j} %
    & = \alpha_I \left (\bigoplus_{j \in J} \gamma_I(\interval{c_j}{a_j})\right )\\
    & = \alpha_I \left (\bigoplus_{j \in J} \set {x_j \mid c_j \leq x_j \leq \gamma(a_j)}\right )\\
    & = \alpha_I \left (  \set { \bigvee_{j \in J}x_j \mid \forall j \in J.\ c_j \leq x_j \leq \gamma(a_j)} \right )\\
    & = \interval{\bigvee_{j \in J} c_j}{\alpha\left (\bigvee_{j \in J} \gamma(a_j) \right )}\\
    & = \interval{\bigvee_{j \in J} c_j}{{\bigvee_{j \in J}} \alpha(\gamma(a_j))}\\
    & = \interval{\bigvee_{j \in J} c_j}{{\bigvee_{j \in J}} a_j}
    \tag*{\qedhere}
  \end{align*}
\end{proof}
\end{toappendix}

This abstract logic features judgements of the form $\vdash_I \as{\interval{c}{a}} \ \r \ \as{\interval{c'}{a'}}$, which are interpreted as asserting that $c' \leq \isem[C]{\r} c$ and $\isem[A]{\r} a \leq_{A} a'$. This effectively combines under\hyp{}approximation 
on the interval lower bound $c$ and over\hyp{}approximation on the interval upper bound $a$.

\begin{toappendix}
Since $\pow{C}$ is completely distributive and possesses an irreducible basis, $\alpha_I$ trivially reflects density. Consequently, by Proposition~\ref{prop:equivalence-monoid}, the abstract program logic 
has the same expressive power as the concrete logic on interval properties.

Furthermore,
the abstraction is complete with respect to $\oplus$ (Lemma~\ref{le:oplus-intervals}).
Therefore, $(C \times A, \sqsubseteq_I, \oplus^I, \alpha_I(\Ir{\pow{C}}))$ constitutes an interpretation monoid. The abstract logic is thus sound and complete with respect to the induced abstract semantics $\sem{\cdot}^I$, which is given by (Lemma~\ref{le:oplus-intervals}):
$\sem{\r}^I\interval{c}{a} = \interval{\isem[C]{\r} c}{\isem[A]{\r} a}$.

\begin{lemmarep} 
  \label{le:oplus-intervals}
  If $(C, \leq)$ is distributive then the abstraction
  $\alpha_I : \pow{C} \to C \times A$ is complete for $\oplus$.
  Therefore  $(C \times A, \sqsubseteq, \oplus^I, \alpha_I(\Ir{\pow{C}}))$ is an interpretation monoid and for all $\interval{c}{a}$ it holds $\sem{\r}^I \interval{c}{a} = \interval{\isem[C]{\r}c}{\isem[A]{\r} a}$.
\end{lemmarep}

\begin{proof}

  For the first part, we have to show that
  \begin{center}
    $\alpha_I\left (\bigoplus_{j \in J} X_j \right ) = \bigoplus_{j \in J}^I \alpha_I(X_j)$
  \end{center}

  Observe that
  \begin{align*}
    \alpha_I\Bigl (\bigoplus_{j \in J} X_j \Bigr ) %
    & = \alpha_I \Bigl (\Bigl \{\bigvee_{j\in J}c_j \mid \forall j\in J.\, c_j \in X_j \Bigr \} \Bigr )\\
    & = \Bigl [\bigwedge \Bigl \{\bigvee_{j\in J}c_j \mid \forall j\in J.\, c_j \in X_j\Bigr \}, 
    \alpha \Bigl ( \bigvee \Bigl \{\bigvee_{j\in J}c_j \mid \forall j\in J.\, c_j \in X_j\Bigr \} \Bigr )
     \Bigr ]\\
    & \qquad \mbox{[by complete distributivity of $C$]}\\
    &
      = \Bigl [\bigvee_{j\in J} \wedge X_j,
      \alpha(\bigvee_{j\in J} \vee X_j)
      \Bigr ]\\ & \qquad 
      \mbox{[by additivity of $\alpha$]} \\
    &
      = \Bigl [\bigvee_{j\in J} \wedge X_j,
      \bigvee_{j\in J} \alpha(\vee X_j)
      \Bigr ]\\
    &
      = \bigoplus_{j \in J}^I \interval{\wedge X_j
      }{\alpha(\vee X_j)}\\
    &  = \bigoplus_{j \in J}^I \alpha_I(X_j)
  \end{align*}
  as desired.
  
  \bigskip
  
  For the second part, an inductive proof on the structure of $\r$ shows that
  \begin{align*}
    \bsem{\r}^I\interval{c}{\alpha(c)} = \interval{\isem[C]{\r}c}{\isem[A]{\r} \alpha(c)}
  \end{align*}
  The base case is given by
  Lemma~\ref{le:hyperproperties-intervals}~\ref{le:hyperproperties-intervals:BCA}. The
  inductive cases are routine.  For instance, in the case of
  sequential composition, $\r = \r_1; \r_2$, we have
  \begin{align*}
    \bsem{\r_1; \r_2}^I \interval{c}{\alpha(c)} %
    & = \add{\bsem{\r_2}^I} \bsem{\r_1}^I\interval{c}{\alpha(c)}\\
    & \qquad \mbox{[by ind. hyp. $\bsem{\r_1}^I\interval{c}{\alpha(c)}=\interval{\isem[C]{\r_1} c}{\isem[A]{\r_1} \alpha(c)}$ and use \eqref{eq:base-interval}]}\\
    & = \sqcup \set{\sem{\r_2}^I\interval{c'}{\alpha(c')} \mid \isem[C]{\r_1} c \leq c' \leq \gamma(\isem[A]{\r_1} \alpha(c))}\\
    & \qquad \mbox{[by ind. hyp. on $\r_2$ $\bsem{\r_2}^I\interval{c'}{\alpha(c')}=\interval{\isem[C]{\r_2} c'}{\isem[A]{\r_2} \alpha(c')}$]}\\
    & = \sqcup \set{\interval{\isem[C]{\r_2} c'}{\isem[A]{\r_2} \alpha(c')} \mid \isem[C]{\r_1} c \leq c' \leq \gamma(\isem[A]{\r_1} \alpha(c))}\\
    & = \interval{\isem[C]{\r_2} \isem[C]{\r_1} c}{\isem[A]{\r_2} \alpha(\gamma(\isem[A]{\r_1} \alpha(c)))}\\
    & = \interval{\isem[C]{\r_2} \isem[C]{\r_1} c}{\isem[A]{\r_2}\isem[A]{\r_1} \alpha(c)}\\
    & = \interval{\isem[C]{\r_1;\r_2} c}{\isem[A]{\r_1; \r_2} \alpha(c)} %
  \end{align*}

  The fact that
  $\sem{\r}^I \interval{c}{a} = \bigsqcup
  \set{\bsem{\r}^I\interval{c'}{\alpha(c')} \mid
    \interval{c'}{\alpha(c')} \in \ebase{\interval{c}{a}}} =
  \interval{\isem[C]{\r}c}{\isem[A]{\r} a}$ is proven exactly as in
  Lemma~\ref{le:hyperproperties-intervals}~\ref{le:hyperproperties-intervals:BCA}.
  \hfill
\end{proof}
\end{toappendix}

   \begin{figure*}[t]
\centering
      \resizebox{\textwidth}{!}{
        \inference[(iter)]{
          \inference[(choice)]{
            \inference[(seq)]{
                \vdash \as{a_0} \ \mathsf{(\mathit{x} > 0?)} \ \as{a_0}
              & %
                \vdash \as{a_0} \ \mathsf{dec} \ x \ \as{a_1}
            }{\vdash \as{a_0} 
              \ \r_1 \ \as{a_1}}
            & 
            \inference[(seq)]{
                \vdash \as{a_0} \ \mathsf{(\mathit{x} < 1000?)} \ \as{a_0}
              & %
                \vdash \as{a_0} \ \mathsf{inc} \ x \ \as{a_2}
            }{\vdash \as{a_0} \ \r_2 \ \as{a_2}}
          }{\vdash \as{a_0} 
          \ \r_1 + \r_2 \ \as{\interval{\set{0, 2, 998, 1000}}{\interval{0}{1000}}}}
        }{\vdash \as{a_0} 
  \ \mathsf{u} \ \as{\interval{\set{0, 2, 1000}}{\interval{0}{1000}}}}
        }
    \caption{Derivation for the program $\mathsf{u} \triangleq ((\mathsf{(\mathit{x} > 0?)}; \mathsf{dec}\ x) + (\mathsf{(\mathit{x} < 1000?)}; \mathsf{inc}\ x))^*$ in Example~\ref{ex:hyper-intervals}, where $a_0 \triangleq \interval{\set{1, 999}}{\interval{1}{999}}$, 
    $a_1 \triangleq \interval{\set{0, 998}}{\interval{0}{998}}$, and $a_2 \triangleq \interval{\set{2, 1000}}{\interval{2}{1000}}$. 
     The side condition for (iter) consists of $\set{0, 2, 1000} \subseteq \set{0, 2, 998, 1000}$ and 
  $\interval{1}{999} \sqcup \interval{0}{1000} \subseteq \interval{0}{1000}$.}
    \label{fig:hyper-intervals}
  \end{figure*}

  \begin{example}
\label{ex:hyper-intervals}
  Consider the program $\mathsf{u} \triangleq (\r_1 + \r_2)^*$
  from Example~\ref{ex:incorr-LCL-example},
  where $\r_1 = \mathsf{(\mathit{x} > 0?)}; \mathsf{dec}\ x$ and $\r_2 = \mathsf{(\mathit{x} < 1000?)}; \mathsf{inc}\ x$, and the hyperset 
    $h = \{\{1, 500, 999\}, \{1, 2, 495, 999\}, %
    \{1, 100, 999\}\}$. It turns out that $h$ is 
  abstracted as $\alpha_I(h) = \interval{\set{1, 999}}{\interval{1}{999}}$.
  Figure~\ref{fig:hyper-intervals} illustrates the derivation of the triple %
    $\vdash \as{\interval{\set{1, 999}}{\interval{1}{999}}} \ \mathsf{u} \ \as{\interval{\set{0, 2, 1000}}{\interval{0}{1000}}}$. %
  \hfill $\lozenge$
\end{example}

Space limitation precludes a full treatment, but we mention that this framework, 
instantiated to ${(C, \leq)} = (\pow{\Sigma}, \subseteq)$ yields a sound and complete logic for certifying 
that an %
abstract evaluation is
the BCA, intended as the abstraction of the concrete semantics, which is a 
problem recently tackled in~\cite{GiacobazziR25}.  For each abstract value $a \in A$, the abstract semantics $\isem[A]{\r}a$ 
is the BCA of $\r$ at $a$ if and only if $\vdash_I \as{\interval{\gamma(a)}{a}} \ \r \ \as{\interval{c'}{a'}}$ is derivable 
for $c' \in \pow{\Sigma}$ and $a' \in A$ is such that $\alpha(c') =a'$. %

  \begin{corollaryrep}[A Logic for BCA]
    \label{cor:bca}
  Let $(C, \leq)$ be a complete lattice.
  For all $a \in A$, if the triple $\vdash_I \as{\interval{\gamma(a)}{a}} \ \r \ \as{\interval{c'}{a'}}$ is derivable for $c' \in C, a' \in A$ such that $\alpha(c') =a'$, then $a'$ is the BCA of $\r$ at $a$, i.e., $\isem[A]{\r}a = \alpha(\isem[C]{\r}\gamma(a))$.
  Moreover, if $C$ is completely distributive, then the converse also holds.
\end{corollaryrep}

\begin{proof}
  Assume that $\vdash_I \as{\interval{\gamma(a)}{a}} \ \r \ \as{\interval{c'}{a'}}$ is derivable and $\alpha(c') = a'$. Then by soundness of the abstract logic (Proposition~\ref{prop:sound-abstract-monoid}), $c' \leq \isem[C]{\r} \gamma(a)$ and $\isem[A]{\r} a \leq_{A} a'$ and thus
  \begin{center}
    $\isem[A]{\r}a \leq_A a' = \alpha(c') \leq_A \alpha(\isem[C]{\r} \gamma(a)) \leq_A \isem[A]{\r}a$
  \end{center}
  where the last inequality is by soundness of the abstract inductive
  semantics. Thus we deduce
  $\isem[A]{\r}a = \alpha(\isem[C]{\r} \gamma(a))$ as desired.

  If $(C, \leq)$ is completely distributive, then, as observed in
  \S~\ref{ss:intervals-hyper}, the abstract logic is sound and
  complete for the abstract semantics
  $\sem{\r}^I\interval{c}{a} = \interval{\isem[C]{\r} c}{\isem[A]{\r}
    a}$. Thus, one can derive
  \begin{center}
    $\vdash_I \as{\interval{\gamma(a)}{a}} \ \r \
    \as{\interval{\isem[C]{\r}\gamma(a)}{\isem[A]{\r}a}}$\, ,
  \end{center}
  which is a
  triple of the right form when
  $\isem[A]{\r}a = \alpha(\isem[C]{\r} \gamma(a))$ holds.
  \hfill
\end{proof}

This corollary lets us prove in the logic that $\mathsf{p} = (x{\neq}0?);(x{=}0?)$ on input
$[-1,1]$ yields the BCA (cf. Example~\ref{ex:intervals-BCA-hyper} in the Appendix). Finally, taking
the trivial abstraction $A = \set{\top}$ makes the upper-bound condition of~\eqref{eq:order-interval}
vacuous, so that the logic keeps only the under-approximation constraint
$c' \leq \isem[C]{\r}c$, reducing the proof system to an incorrectness logic.

\begin{toappendix}
\begin{example}\label{ex:intervals-BCA-hyper}
  As asserted in Example~\ref{ex:interval-BCA}, we can verify that the computation of the program $\mathsf{p} = (x \neq 0?); (x = 0?)$ on the input $[-1, 1]$ is the BCA.
  Let $C = \pow{\Z}$ and $A = \Int$. For any non-empty $S \subseteq \Z$, let $\alpha(S) = \interval{\inf S}{\sup S}$, and for $[a, b] \in \Int$, let $\gamma([a, b]) = \set{z \in \Z \mid a \leq z \leq b}$.
  We can indeed derive the triple
    $\vdash \as{\interval{\set{-1, 0, 1}}{[-1, 1]}} \ \mathsf{p} \ \as{\interval{\varnothing}{\varnothing}}$, 
  as illustrated in Figure~\ref{fig:intervals-hyper}.
  \hfill$\lozenge$
\end{example}

\begin{figure*}
    \hspace{-25pt}
     \begin{tabular}{c}
      \small
       $\inference[(join)]{
          \as{\interval{\set{-1, 0}}{[-1, 0]}} \ (x \neq 0?); (x = 0?) \ \as{\interval{\varnothing}{\varnothing}}
         \vdash \as{\interval{\set{0, 1}}{[0, 1]}} \ (x \neq 0?); (x = 0?) \ \as{\interval{\varnothing}{\varnothing}}}
       {\vdash \as{\interval{\set{-1, 0, 1}}{[-1, 1]}} \ (x \neq 0?); (x = 0?) \ \as{\interval{\varnothing}{\varnothing}}}$
      \end{tabular}
      \caption{Derivation for Example~\ref{ex:intervals-BCA-hyper}.}
      \label{fig:intervals-hyper}
\end{figure*}
\end{toappendix}

\section{Related Work}\label{sec:related}

Recent years have witnessed a growing interest in program logics for reasoning about
a wide range of properties. This includes multiple variants of correctness
logics~\cite{VerschtK25,Cousot24}, several forms of incorrectness
logic~\cite{VerschtK25,OHearn20,AscariBGL25,RDB:LRAPB,RaadBDO22,LeRVBDO22,abs-2502-14626,abs-2411-06416,ABGR26},
and combinations of correctness and incorrectness
reasoning~\cite{BGGR21,BruniGGR23,ZilbersteinSS24}. The landscape has further expanded to
logics for properties over abstract domains~\cite{BGGR21,BruniGGR23,GiacobazziM10,CCLB:AIFR},
logics for reasoning about program analyses~\cite{GiacobazziR25,popl26}, logics for
hyperproperties~\cite{GiacobazziM04,DM:HHL,CousotW25,ZhangZK024}, quantitative program
logics~\cite{ZhangK22}, and unifying logical
frameworks~\cite{zil-popl26,Zil:toplas,ZDS:OL,ZilbersteinKST25}. We focus here on the work
most closely related to ours.

\paragraph{Program logics and abstract interpretation.}
The idea of a program logic whose assertions are elements of an abstract interpretation
domain was first explored by~\citet{CCLB:AIFR}. 
As discussed in \S\ref{sec:intro}, they observed that the disjunction
rule is unsound over a generic abstract domain unless the
concretisation map $\gamma$ preserves joins. Moreover, since their
framework does not require best abstractions (i.e., Galois
connections), the conjunction rule may likewise be unsound when
$\gamma$ does not preserve meets. Our density condition pinpoints
exactly when the disjunction rule can be recovered: a family is dense
precisely when its join is computed without over-approximation, which
is the join-preservation of $\gamma$ localised to that family.

Ferrarini's MSc thesis~\cite{FERRARINI} develops a generalised Hoare logic parametric in the
abstract domain, using complete lattices as assertion languages, and already highlights the
difficulty with the disjunction rule. This logic is an instance of our framework, obtained by
taking the interpretation domain with domain and basis coinciding with the whole abstract domain.

Local Completeness Logic (LCL), proposed by~\citet{BruniGGR23}, is
also parametric in an abstract domain $A$ and combines under- and
over-approximating reasoning. Unlike ours, LCL postconditions are
concrete under-approximations $c$ whose abstraction in $A$
over-approximates the exact postcondition. 
Moreover, LCL is intrinsically incomplete unless the abstraction $A$
is trivial.
Both logics nonetheless follow the same design principle: replacing a global property by a local side condition, checked in the derivation. In LCL, completeness
of the abstract transfer functions becomes their completeness on the current input. In
$\APPL$, additivity of the semantics, i.e., exactness of every join, becomes join exactness on the current cover, i.e., density.%

\paragraph{Disjunctive completion and trace partitioning.}
Disjunctive completion~\cite{CC79,GiacobazziR98,FR99} adds to an abstract domain the missing
joins, i.e., the concrete unions of its elements, so that its concretisation preserves
joins. In our setting, disjunctive domains are degenerate cases: being completely
distributive, with their irreducible basis every family is dense
(cf.\ Lemma~\ref{le:basic-prop-sem:distrib}), so that (join) is unconstrained, as in
Hoare logic. Rather than disjunctively completing the domain, $\APPL$ keeps it fixed and recovers
exactness locally, on the covers used in a derivation, which avoids both the blow-up of
the disjunctive completion  and the re-engineering of its transfer functions.

The (join) rule of $\APPL$ is also related to the trace partitioning abstract
domain~\cite{HandjievaT98,RivalM07}, which addresses the precision loss
incurred when a non-disjunctive analysis merges information from different branches. In
an interval analysis of $\mathbf{if}~x<0~\mathbf{then}~x:=-1~\mathbf{else}~x:=1$, joining
$x\in[-1,-1]$ from the \emph{then} branch with $x\in[1,1]$ from the \emph{else} branch
yields
$x\in[-1,1]$, including the spurious value $0$. Trace
partitioning keeps the two branches apart, recording that the \emph{then} branch yields
$x\in[-1,-1]$ and the \emph{else} branch $x\in[1,1]$, thus multiplying abstract
states globally. The (join) rule achieves local, state-level disjunctive precision within
a derivation step, without changing the control-flow graph or refining the
 domain. Moreover, its cases need not follow control flow: the loop body of the example in
\S~\ref{sub:octagons} has no branch on the sign of $x$, and the precision gain obtained there rests
on the octagonal relation $x+y=0$, which the same split cannot exploit over intervals.

\paragraph{Program logics for hyperproperties.}
Dardinier and M\"uller~\cite{DM:HHL} introduced Hyper Hoare Logic, generalising Hoare logic
from state properties to arbitrary hyperproperties and allowing one to both prove and
disprove hyperproperties in the same logic. This logical system is sound and complete, with instances
for $\forall\exists$ and $\exists\forall$ hyperproperties. As shown in
Example~\ref{ex:hyper-muller}, the (join) rule of $\APPL$ can be seen as a recast of their
(Exist) rule within our framework.

Cousot and Wang~\cite{CousotW25} study hyperproperties through Cousot's calculational
design~\cite{Cousot24}. From a core of rules valid for general hyperproperties, 
their framework supports
abstraction by design and yields sound and complete proof systems for $\forall\exists$,
$\forall\forall$, and $\exists\forall$ hyperproperties.

Ferrarini's MSc thesis~\cite{FERRARINI} also derives a logic for hyperproperties, which however fails to
be complete. In light of the present work, the incompleteness can be traced to the monoidal
operation~$\oplus$ of \S~\ref{ss:hypersemantics} being forced to coincide with the lattice
join.

Mastroeni and Pasqua~\cite{MP:VBSCH} apply abstract interpretation-based static analysis to
hyperproperty verification, but restrict to down-closed hyperproperties, an
instance of our framework (cf.~\S~\ref{ss:subset-close}).

Outcome Logic (OL)~\cite{ZDS:OL} handles programs with qualitative and quantitative effects
and gives a natural account of hyperproperties~\cite{Zil:toplas}. Although $\APPL$ with its
monoidal operator and OL share similar algebraic structure, neither is an instance of the
other. OL rests on a semantic framework with stronger requirements, notably
left-distributivity of the monoidal operator over sequential composition. This makes OL
unsuitable for abstraction (cf. Example~\ref{ex:intervals-kleene-fail}) and makes some of
its structural rules, such as (Choice) in~\cite{Zil:toplas}, unsound in $\APPL$. Conversely,
$\APPL$ does not handle quantitative semantics, and including them may require extending
the framework.

\section{Conclusion and Future Work}\label{sec:conclusion}
We %
introduced $\APPL$, a Hoare-style program logic that places 
abstract interpretation on proof-theoretic foundations. The logic 
is parametric in an interpretation domain, i.e.\ a complete lattice equipped with a join-generating basis, 
and treats the choice of the assertion language itself as an abstraction, following the principles of abstract interpretation. $\APPL$ is sound over an arbitrary abstract domain and relatively complete whenever the assertion language is expressive enough. 

By extending the logic with a monoidal operator, 
we recover classical Hoare logic, incorrectness logic, and Hyper Hoare Logic.
This yields a principled explanation of when and why Hoare-style reasoning 
can be lifted to abstract domains and hyperproperties. 
A distinctive feature of $\APPL$ is that 
density governs the soundness of the Hoare-style disjunction rule over non-disjunctive domains.
Moreover, $\APPL$ internalises
the use of widening and narrowing through inference rules, so that an abstract interpretation run, 
with its ascending and descending phases, becomes a single derivation in the logic; the operators and 
their application strategy remain external inputs.

\paragraph{Future Work.}
Our framework accommodates, for a fixed abstract domain, a whole
hierarchy of abstract semantics whose precision is set by the chosen
join-basis. This semantic choice has a direct logical counterpart: the
basis fixes the alphabet of properties available to decompose a
precondition, and thus determines how---and how finely---the
(join) rule can be applied.  The trivial basis (i.e., the entire
domain) recovers the standard inductive semantics; smaller bases yield
strictly finer ones; and the irreducible basis, when it exists, gives
maximal precision at a potentially prohibitive cost.  This trade-off
turns the choice of basis into a genuine \emph{proof-search}
problem. Rather than committing to a static basis, an effective
procedure should adapt the dense decomposition on the fly, splitting a
precondition only where precision is actually needed, typically at
guards and other control-flow junctions, and coarsening elsewhere to
stay tractable. Trace partitioning~\cite{HandjievaT98,RivalM07} is a
natural source of inspiration here: it suggests heuristics that let
the logic recover path-sensitivity by introducing dense splits lazily,
driven by the branches the analysis takes. 
A natural question is what guarantees such a strategy could offer. For instance, whether
one can characterise, for a given program and target property, a \emph{minimal} dense
decomposition sufficient to derive it, and whether it can be discovered incrementally
during proof search.
A related open problem is to characterise the domains and programming primitives for which finite dense covers suffice. 
They do in all our examples, but not in general: 
the proof of %
completeness applies
(join) to possibly infinite covers.

A second direction concerns hyperproperties. By integrating
abstraction systematically, our framework provides a principled
foundation for \emph{designing} abstractions of hyperproperties, a
largely unexplored setting. We have taken first steps by recovering
and generalising known constructions, such as the down-closed
hyperproperties and homomorphic semantic abstractions
of~\cite{CousotW25,MP:VBSCH}. These are only two items in a
larger design space made accessible by the interpretation monoid structure.

A third direction targets a challenge identified by
\citet{CousotW25}: in hypersemantics the logical order is typically
subset inclusion, so once $\vdash \as{\set{X_1, X_2}} \ \r \ \as{H}$
is derived, the exact outputs $\sem{\r}X_1$ and $\sem{\r}X_2$ must
be in $H$. A postcondition can then be weakened only by
\emph{adding} elements, 
never by approximating those already present, 
precisely the move that over-approximation needs. We
plan to tackle this issue by injecting into the framework an \emph{approximation
  preorder} weaker than the lattice order, under which a hyperproperty
can be weakened both by adding elements and by replacing elements with
their approximations.  Canonical instances could be the Hoare preorder
(replacing elements by over-approximations) and the Egli--Milner
ordering (for convex hyperproperties). 
As hinted in \S~\ref{sec:nondet} when nondeterminism is modeled via a monoidal operator, an \emph{approximation preorder} weaker than the lattice order, for which $\oplus$ is inflationary, can help in recovering iteration semantics via a fixpoint operator, thereby restoring the (widen) and (narrow) rules.

\section*{Disclosure of Use of AI in Research}
The authors used a generative AI assistant to review selected proofs, construct and verify specific examples and counterexamples, and polish the manuscript's phrasing. All technical content and findings were independently verified by the authors, who take full responsibility for the final output.

\bibliographystyle{ACM-Reference-Format}
\bibliography{biblio}

\clearpage

\appendix

\makeatletter
\let\ACM@origsection\section
\makeatother
\end{document}